\documentclass[11pt, onecolumn]{IEEEtran} %12pt, draftcls, onecolumn
\pdfoutput=1
\usepackage{amsmath}
\usepackage{amssymb}
\usepackage{amsfonts}
\usepackage{amssymb,amsmath,amsthm, amsfonts}
\usepackage{epsfig}
\usepackage{subfigure}
\usepackage{calc}
\usepackage{color}
\usepackage[all]{xy}
\usepackage{bm}
\usepackage{cite}
\usepackage{enumerate}
\usepackage{hyperref}

\newtheorem{defn}{Definition}
\newtheorem{thm}{Theorem}[section]
\newtheorem{cor}[thm]{Corollary}
\newtheorem{prop}{Proposition}

\newtheorem{lem}[thm]{Lemma}
\newtheorem{conj}[thm]{Conjecture}
\newtheorem{constr}[thm]{Construction}
\newtheorem{note}{Remark}

\newcommand{\bit}{\begin{itemize}}
\newcommand{\eit}{\end{itemize}}
\newcommand{\bcor}{\begin{cor}}
\newcommand{\ecor}{\end{cor}}
\newcommand{\beq}{\begin{equation}}
\newcommand{\eeq}{\end{equation}}
\newcommand{\beqn}{\begin{equation*}}
\newcommand{\eeqn}{\end{equation*}}
\newcommand{\bea}{\begin{eqnarray}}
\newcommand{\eea}{\end{eqnarray}}
\newcommand{\bean}{\begin{eqnarray*}}
\newcommand{\eean}{\end{eqnarray*}}
\newcommand{\ben}{\begin{enumerate}}
\newcommand{\een}{\end{enumerate}}
\newcommand{\bdefn}{\begin{defn}}
\newcommand{\edefn}{\end{defn}}
\newcommand{\bnote}{\begin{note}}
\newcommand{\enote}{\end{note}}
\newcommand{\bprop}{\begin{prop}}
\newcommand{\eprop}{\end{prop}}
\newcommand{\blem}{\begin{lem}}
\newcommand{\elem}{\end{lem}}
\newcommand{\bthm}{\begin{thm}}
\newcommand{\ethm}{\end{thm}}
\newcommand{\bconj}{\begin{conj}}
\newcommand{\econj}{\end{conj}}
\newcommand{\bconstr}{\begin{constr}}
\newcommand{\econstr}{\end{constr}}
\newcommand{\bpf}{\begin{proof}}
\newcommand{\epf}{\end{proof}}

\begin{document}

\title{Outer Bounds on the Storage-Repair Bandwidth Tradeoff of Exact-Repair Regenerating Codes} 
 \author{Birenjith Sasidharan, N. Prakash, M. Nikhil Krishnan, Myna Vajha, Kaushik Senthoor, 
\ \\ 
 and P. Vijay Kumar
\ \\
\ \\
(email: biren@ece.iisc.ernet.in, prakashn@mit.edu, \{nikhilkm,myna,kaushik.sr,vijay\}@ece.iisc.ernet.in) 
\ \\
%Department of ECE, Indian Institute of Science, Bangalore, India.
\ \\
%(email: \{\}@ece.iisc.ernet.in)
\thanks{This research is supported in part by the National Science Foundation under Grant 1421848 and in part by an India-Israel UGC-ISF joint research program grant. Birenjith Sasidharan would like to thank the support of TCS Research Scholar Programme Fellowship awarded to him. N. Prakash was a PhD student at IISc, Bangalore, and also an intern at NetApp, Bangalore during the duration of this work. M. Nikhil Krishnan and Myna Vajha would like to thank the support of Visvesvaraya PhD Scheme for Electronics \& IT awarded by Department of Electronics and Information Technology, Government of India. A portion of the material in this paper was presented in part at the 2014 IEEE
International Symposium on Information Theory \cite{SasSenKum_isit}, and in part at the 2015 IEEE
International Symposium on Information Theory \cite{PraKri_isit}.}}
\date{\today}
\maketitle

\begin{abstract}
In this paper three outer bounds on the storage-repair bandwidth (S-RB) tradeoff of regenerating codes having parameter set $\{(n,k,d), (\alpha,\beta)\}$ under the exact-repair (ER) setting are presented. The tradeoff under the functional-repair (FR) setting was settled in the seminal work of Dimakis et al. that introduced the framework of regenerating codes as well as a subsequent paper by Wu. While it is known that the ER tradeoff coincides with the FR tradeoff at the extreme points of the tradeoff, known respectively as the minimum-storage-regenerating (MSR) and minimum-bandwidth-regenerating (MBR) points, its characterization on the interior points remains open. 

The first outer bound presented here termed as the {\em repair-matrix bound}, in conjunction with a recent code construction known as {\em improved layered codes} characterizes the normalized ER tradeoff for the case of $(n,k=3,d=n-1)$. The repair-matrix bound is derived by building on top of the techniques introduced by Shah et al. and applies to every parameter set $(n,k,d)$. It was earlier proved by Tian that the ER tradeoff lies strictly away from the FR tradeoff for the specific case $(n=4,k=3,d=3)$.  The repair-matrix bound shows that a non-vanishing gap exists between the ER and FR tradeoffs for {\em every} parameter set $(n,k,d)$.

The second outer bound builds upon a bound due to Mohajer and Tandon and improves the bound using the very same techniques introduced in the Mohajer-Tandon paper and for this reason, is termed here as the {\em improved Mohajer-Tandon} bound. While for $d=k$ the improved Mohajer-Tandon bound performs on par with the Mohajer-Tandon bound, for $d>k$ there is a significant improvement in the region of the tradeoff away from the MSR point.   In the vicinity of the MSR point however, the repair-matrix bound outperforms the improved Mohajer-Tandon bound.

In the third and final result, we restrict our focus to linear codes, and present an outer bound for the normalized ER tradeoff applicable to linear codes for the case $k=d$. In conjunction with the well-known class of {\em layered codes}, our third outer bound characterizes the normalized ER tradeoff in the case of linear codes for the case $k=d=n-1$. This bound is derived by analyzing the rank-structure of a parity-check matrix for a linear ER code.
\end{abstract}


\section{Introduction} \label{sec:intro} 

\subsection{Regenerating Codes} \label{sec:regenerating_codes} 

%In a distributed storage system, the amount of data transmitted across the data network for the purpose of repairing a failed node, termed as the {\em repair bandwidth}, is desired to be much less in comparison with the total amount of data stored in the system. Regenerating codes are a class of codes that was introduced \cite{DimGodWuWaiRam} to address the problem of minimizing repair bandwidth. 

In the regenerating-code framework~\cite{DimGodWuWaiRam}, all symbols are drawn from a fixed finite field $\mathbb{F}$ whose size is the power of a prime. The size of the field does not play an important role in the present paper and for this reason does not appear in our notation for the field.   Data pertaining to a file comprised of $B$ symbols is encoded into a set of $n\alpha$ coded symbols and then stored across $n$ nodes in the network with each node storing $\alpha$ coded symbols.  A data collector should be able to retrieve the file downloading entire data from any $k$ nodes. Furthermore, $k$ is the minimum such number that allows reconstruction of the file. In the event of a node failure{\footnote{Though regenerating codes are defined for the case of single node-failures, there are later works that looked into the case of simultaneous failure of multiple nodes, and studied cooperative repair in such a situation\cite{ShuHu,KerScoStr}. However in this paper, we focus only on single node-failures.}}, node repair is accomplished by having the replacement node connect to any $d$ nodes and download $\beta \leq \alpha$ symbols from each node with $\alpha \leq d \beta < B$. These $d$ nodes are referred to as helper nodes. From the minimality of $k$, it can be shown that $d$ must lie in the range
\bean
k \leq d \leq n-1.
\eean
The quantity $d\beta$ is called as the repair bandwidth. Here one makes a distinction between functional and exact repair.  By functional repair (FR), it is meant that a failed node will be replaced by a new node such that the resulting network continues to satisfy the data-collection and node-repair properties defining a regenerating code.   An alternative to functional repair is {\em exact repair} (ER) under which one demands that the replacement node store precisely the same content as the failed node.  From a practical perspective, ER is preferred at least for two reasons. Firstly, the algorithms pertaining to data collection and node repair remain static for the ER case.  Secondly if the ER code is linear, then it permits the storage of data in systematic form, which facilitates operations under paradigms such as MapReduce~\cite{mapreduce}.  We will use ${\cal P}_{\text{f}}$ to denote the {\em full parameter set} ${\cal P}_{\text{f}} = \{(n,k,d), (\alpha,\beta)\}$ of a regenerating code and use ${\cal P}$ when we wish to refer to only the parameters $(n,k,d)$.

\begin{figure}[ht]
\begin{minipage}[b]{0.5\linewidth}
\centering
\includegraphics[height=1.5in]{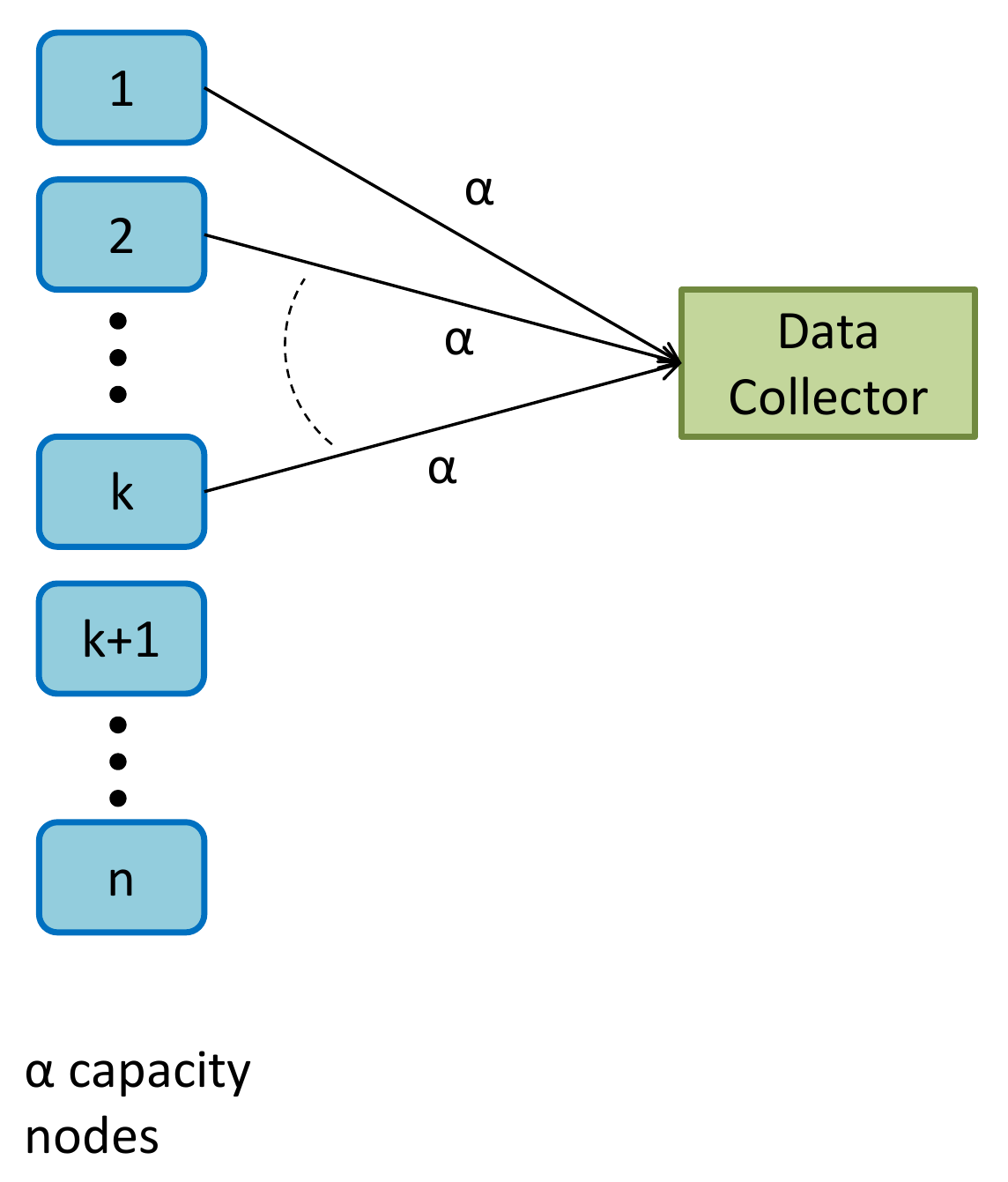}
\caption{Data collection.}
\label{fig:data_collection}
\end{minipage}
\hspace{-0.5cm}
\begin{minipage}[b]{0.50\linewidth}
\centering
\includegraphics[height=1.5in]{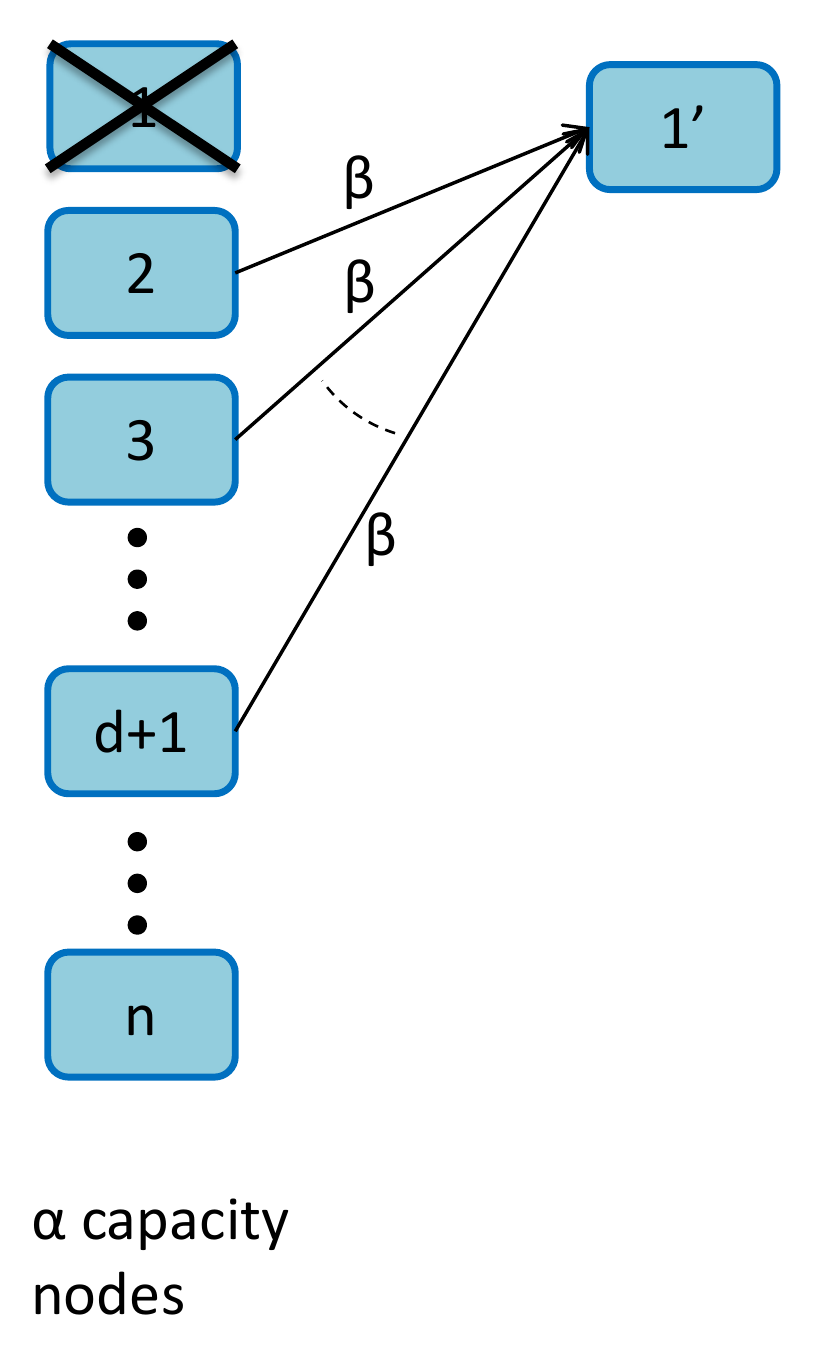}
\caption{Node repair.}
\label{fig:node_repair}
\end{minipage}
\hspace{-0.5cm}
\end{figure}

\subsection{The Storage-Repair Bandwidth Tradeoff} \label{sec:intro_tradeoff}

A cut-set bound based on network-coding concepts, tells us that given a code parameter set ${\cal P}_{\text{f}}$, the maximum possible size $B$ of a regenerating code is upper bounded as~\cite{DimGodWuWaiRam},
\bea \label{eq:cut_set_bd}
B & \leq & \sum_{\ell =0}^{k-1} \min\{\alpha,(d-\ell)\beta\} .
\eea
The derivation of the bound in \eqref{eq:cut_set_bd} makes use of only FR constraints, and therefore it is valid for both FR and ER codes. An FR code ${\cal \hat{C}}$ is said to be optimal if the file size $\hat{B}$ of ${\cal \hat{C}}$ achieves the cut-set bound in \eqref{eq:cut_set_bd} with equality, and further, that if either $\alpha$ or $\beta$ is reduced, equality fails to hold in \eqref{eq:cut_set_bd}. The existence of such codes has been shown in~\cite{DimGodWuWaiRam}, using network-coding arguments related to multicasting~\cite{Wu}. In general, we will use ${\cal \hat{C}}$, $\hat{B}$ etc to denote symbols relating to an optimal FR code while reserving ${\cal {C}}$, $B$ etc. to denote symbols relating to an ER code. 

Given ${\cal P}$ and $B$, there are multiple pairs $(\alpha,\beta)$ that satisfy \eqref{eq:cut_set_bd}. It is desirable to minimize both $\alpha$ as well as $\beta$ since minimizing $\alpha$ reduces storage requirements, while minimizing $\beta$ results in a storage solution that minimizes repair bandwidth.  It is not possible to minimize both $\alpha$ and $\beta$ simultaneously and thus there is a tradeoff between choices of the parameters $\alpha$ and $\beta$.  This tradeoff will be referred to as Storage-Repair Bandwidth (S-RB) tradeoff under functional repair.  Since much of the emphasis of the current paper is upon the distinction between the S-RB tradeoffs under functional and exact repair, we will use FR tradeoff and ER tradeoff to refer respectively, to the two tradeoffs. The two extreme points in the FR tradeoff are termed the {\em minimum storage regeneration} (MSR) and {\em minimum bandwidth regeneration} (MBR) points respectively. The parameters $\alpha$ and $\beta$ for the MSR point on the tradeoff can be obtained by first minimizing $\alpha$ and then minimizing $\beta$ to yield
\bea \label{eq:MSR} 
B = k \alpha, \ \ \alpha = (d-k+1)\beta .
\eea
Reversing the order leads to the MBR point which thus corresponds to
\bea \label{eq:MBR} 
B = \left( dk - {k \choose 2} \right)\beta,  \ \ \alpha = d \beta .
\eea
\begin{figure}[ht]
\centering
\includegraphics[width=3.5in]{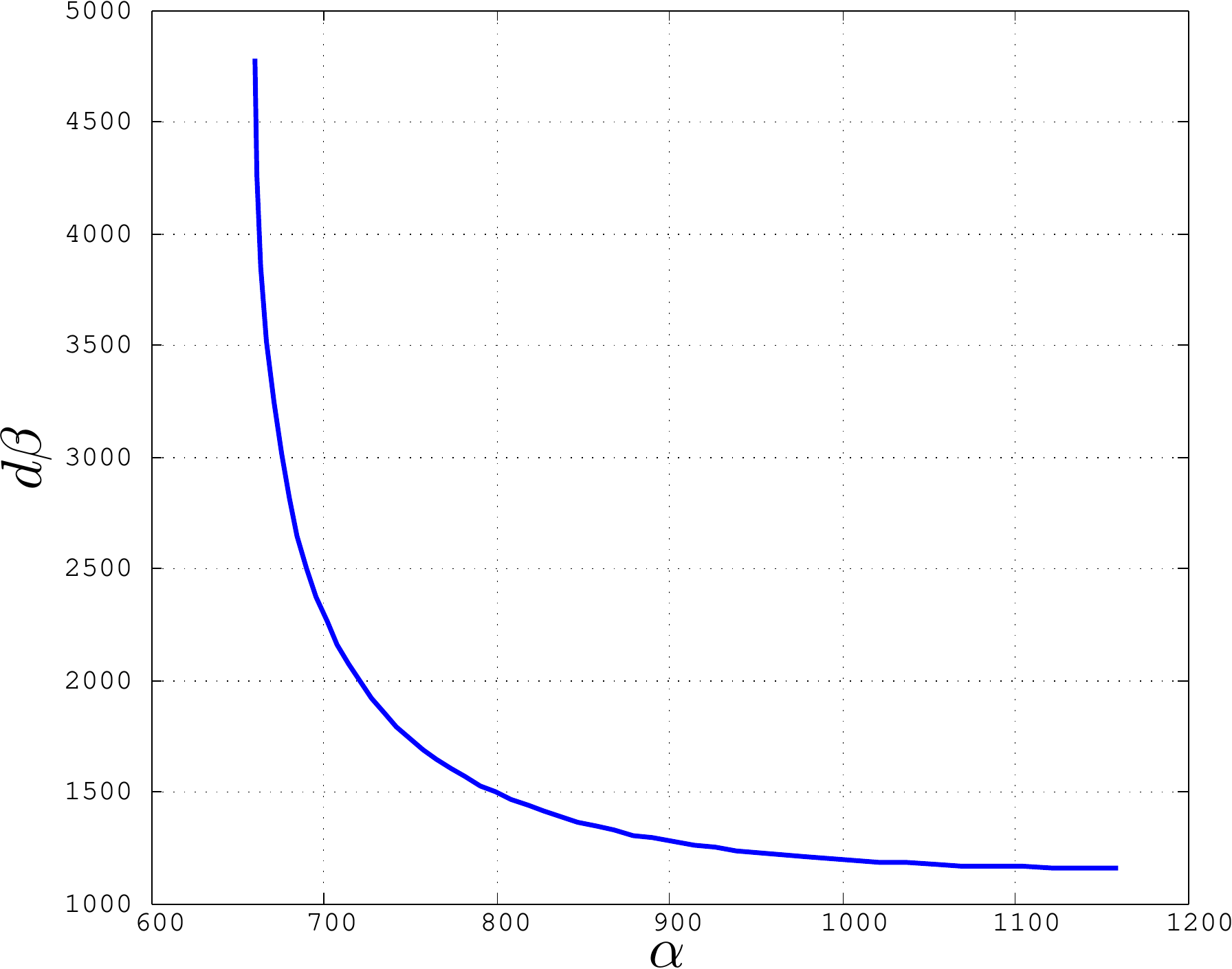}
\caption{FR Tradeoff. Here $(n=60, k=51, d=58, B=33660)$.}
\label{fig:tradeoff}
\end{figure}

The remaining points on the tradeoff will be referred to as {\em interior points}. As the tradeoff is piecewise-linear, there are $k$ points of slope discontinuity, corresponding to 
\bean
\alpha = (d-\mu)\beta, \ \ \mu \in \{0, \cdots k-1 \}.
\eean
Setting $\mu=k-1$ and $0$ respectively, yields the MSR and MBR points. The remaining values of $\mu \in \{1, \cdots k-2\}$ correspond to interior points with slope-discontinuity. Interior points where there is no slope discontinuity can be specified by setting,
\bea
\nonumber \alpha & = & (d-\mu)\beta - \theta , \ \theta \in [0, \beta) \\
\label{eq:op_point}& = & (d-\mu)\beta - \nu \beta , \ \nu \in [0 , 1),
\eea
with $\mu \in \{0,1,\ldots, k-2\}$. When $\mu = k-1$, we always set $\nu =0$. We will refer to the pair $(\alpha, \beta)$ as an {\em operating point} of the regenerating code. The tradeoff between $\alpha$ and $d\beta$ is plotted in Fig.~\ref{fig:tradeoff} for $(n=131,k=120,d=130)$ and file size $B=725360$. 

The results in the present paper pertain to the ER tradeoff.  Several ER code constructions \cite{RasShaKum_pm, CadJafMalRamSuh, PapDimCad, SuhRam, ShaRasKumRam_ia, ShaRasKumRam_rbt, TamWanBru} are now available that correspond to the MSR and the MBR points of the FR tradeoff. Thus the end points of the ER tradeoff coincide with those of the FR tradeoff. However, characterization of the interior points of the ER tradeoff remains an open problem in general. 

\subsection{The Normalized ER Tradeoff and ER-Code Symmetry \label{sec:norm_tradeoff}}

For a given parameter set ${\cal P} = (n,k,d)$, there are several known constructions for an ER code, each of which is valid only for a restricted set of file sizes. Since the ER tradeoff for a fixed $(n,k,d)$ varies with file size $B$, comparison across code constructions is difficult.   For this reason, we normalize $(\alpha, \beta)$ by the file size $B$. The tradeoff between $\bar{\alpha}=\frac{\alpha}{B}$ and $\bar{\beta}=\frac{\beta}{B}$ thus obtained for a fixed  value of $(n,k,d)$, will be referred to here as the {\em normalized ER tradeoff}. The tuple $(\bar{\alpha}, \bar{\beta})$ is referred to as the {\em normalized operating point} of a regenerating code. Throughout the remainder of this paper, we will work only with the normalized version of the ER tradeoff. 

Given a regenerating code ${\cal C}$ associated to parameter set ${\cal P}$ and file size $B$, the parameters of the code are clearly invariant to coordinate (i.e., node) permutation.   Given an ER code ${\cal C}$, we can vertically stack the $n!$ codewords obtained by encoding independent files using all possible node permutations of ${\cal C}$.  The resultant stack of $n!$ codewords may be regarded as a single new ER regenerating code ${\cal C'}$ where the parameters $(n,k,d)$ remain the same, but where the parameters $(\alpha,\beta)$ and $B$ are each scaled up multiplicatively, by a factor of $n!$.  It is clear that ${\cal C'}$ is symmetric in the sense that the amount of information contained in a subset $A \subset [n]$ of nodes depends only upon the size $|A|$ of $A$, and not upon the particular choice of nodes lying in $A$. This symmetry carries over even in the case of repair data transferred by a collection $D$ of $d=|D|$ nodes for the replacement of a fixed node. Such codes will be referred to as {\em symmetric} ER codes. Since the normalized values $(\bar{\alpha}, \bar{\beta})$ of ${\cal C'}$ remain the same as that of ${\cal C}$, there is no change in operating point on the normalized ER tradeoff in going from ${\cal C}$ to ${\cal C'}$. Thus, given our focus on the normalized tradeoff, it is sufficient to consider {\em symmetric} ER codes. This observation was first made by Tian in \cite{Tia}.

\subsection{Results \label{sec:results}}

Though the complete characterization of normalized ER tradeoff for every parameter set remains an open problem, much progress has been made. It was shown in \cite{ShaRasKumRam_rbt}, that apart from the MBR point and a small region adjacent to the MSR point, there do not exist ER codes whose $(\alpha, d\beta)$ values correspond to coordinates of an interior point on the FR tradeoff. However, the authors of \cite{ShaRasKumRam_rbt} did not rule out the possibility of approaching the FR tradeoff asymptotically i.e., as the file size $B \rightarrow \infty$. It was first shown by Tian in \cite{Tia} that the ER tradeoff lies strictly away from the FR tradeoff. This was accomplished by using an information theory inequality prover~\cite{ITIP} to characterize the normalized ER tradeoff for the particular case of $(n,k,d)=(4,3,3)$ and showing it to be distinct from the FR tradeoff.  The results in the \cite{Tia} were however, restricted to the particular case $(n,k,d)=(4,3,3)$.

That the ER tradeoff lies strictly above the FR tradeoff for {\em any} value of the parameter set $(n,k,d)$ was first shown in \cite{SasSenKum_isit}. The first result in the present paper is to show an outer bound on the normalized ER tradeoff for every parameter set $(n,k,d)$, and is stated in Thm.~\ref{thm:bound1}. We refer to this outer bound as the {\em repair-matrix bound}. This outer bound in conjunction with a code construction appearing in \cite{SenSasKum_itw}, characterizes the normalized ER tradeoff for the parameter set $(n,k,d)$ for $k=3$, $d=n-1$ and any $n \geq 4$.   

Two outer bounds on the normalized ER tradeoff appeared subsequently in \cite{Duursma2014} and \cite{Duursma2015}.  In \cite{Duursma2014}, the author presents two bounds on the ER file size. In the first bound, he builds on top of the techniques presented in \cite{Tia} and derives a bound that applies to a larger set of parameters. The second bound is obtained by taking a similar approach as in \cite{SasSenKum_isit}, and is shown to improve upon the one given in \cite{SasSenKum_isit}. In \cite{Duursma2015}, the author provides an upper bound on ER file size, that is non-explicit in general. However for the case of linear codes, the bound can be computed to obtain an explicit expression for any parameter set $(n,k,d)$. A second paper by Tian, \cite{Tia_544}, characterizes the ER tradeoff for $(n=5,k=4,d=4)$ with the help of a class of codes known as the {\em layered codes} introduced in \cite{TiaSasAggVaiKum}. A different approach adopted to derive an outer bound on the normalized ER tradeoff is presented in \cite{MohTan}. In \cite{MohTan}, Mohajer et al. derived an outer bound for general $(n,k,d)$ that turns out to be optimal for the special case of $(n,k=n-1,d=n-1)$ in a limited region of $\bar{\beta} \leq \frac{2\bar{\alpha}}{k}$ close to the MBR point.  Optimality follows from the fact that a code construction due to Goparaju et al. in \cite{GopRouCal_isit} meets their outer bound in the region $\bar{\beta} \leq \frac{2\bar{\alpha}}{k}$.  We will refer to this outer bound in \cite{MohTan} as the {\em Mohajer-Tandon bound}.

%{\color{red}   Seems like we are making a stronger claim with the improved Mhajer-Tandon bound than we had earlier imagined ?}  

The second result of the present paper is an improvement upon the Mohajer-Tandon bound for the case $k < d$. We make use of the very same techniques introduced in \cite{MohTan} to arrive at this improved bound. This bound is stated in Thm.~\ref{thm:bound2}, and we refer to it as the {\em improved Mohajer-Tandon bound}. While the improved Mohajer-Tandon bound performs better whenever $k < d$, it coincides with the Mohajer-Tandon bound when $k=d$. The repair-matrix bound still performs better than the improved Mohajer-Tandon bound in a region close to the MSR point. The theorem below essentially combines the repair-matrix bound and the improved Mohajer-Tandon bound.

%Thus both the improved Mohajer-Tandon bound and the repair-matrix bound can be combined together to obtain an outer bound on the normalized ER tradeoff. To our knowledge, this yeilds the best known outer bound on the normalized ER tradeoff. This is stated in the following theorem.

\bthm \label{thm:bound3} Let 
\bean
B_1 = \sum_{i=0}^{k-1} \min\{\alpha, (d-i)\beta) - \delta,
\eean where $\delta$ is as defined in \eqref{eq:eps}, and it corresponds to the repair-matrix bound. Let $B_2$ be the expression on the RHS in \eqref{eq:soh_improved}, corresponding to the improved Mohajer-Tandon bound. Then the ER file size $B$ is bounded by,
\bean
B & \leq & \min\{B_1, B_2\}.
\eean
\ethm

The final result presented in this paper is under the restricted setting of linear codes. For the case of $(n \geq 4, k=n-1, d=n-1)$, we characterize the normalized ER tradeoff under this setting. This is done by deriving an explicit upper bound on the file size $B$ of a ER linear regenerating code for the case $k=d=n-1, n \geq 4$. The outer bound remains valid for the general case $k=d$ even when $d< n-1$. For the case of $(n,k=n-1,d=n-1)$, the outer bound matches with the region achieved by the layered codes. This result, which first appeared in\cite{PraKri_isit}, is stated below:

\begin{thm} \label{thm:new_bound_k_eq_d}
Consider an exact repair linear regenerating code, having parameters $(n, k = n-1, d = n-1), (\alpha, \beta), n \geq 4$. Then, the file size $B$ of the code is upper bounded by 
\begin{eqnarray} \label{eq:bound_rank_G}
B & \leq & \left \{ \begin{array}{rl} \left \lfloor \frac{r(r-1)n\alpha + n(n-1)\beta}{r^2+r}\right \rfloor , & \frac{d\beta}{r} \leq \alpha \leq \frac{d\beta}{r-1},  
 \ \ \ 2 \leq r \leq n - 2 \\  
(n-2)\alpha + \beta, & \frac{d\beta}{n-1} \leq \alpha \leq \frac{d\beta}{n-2} \end{array} \right. .
\end{eqnarray}
\end{thm}
We remark that there are no known instances of non-linear codes that violate the above outer bound derived under the linear setting. In an independent work \cite{ElyMohTan}, the authors also derive the normalized linear ER tradeoff for the case $(n,k=n-1,d=n-1)$, but the tradeoff is expressed in an implicit manner as the solution to an optimization problem. 

\begin{figure}[h!]
\begin{center}
  \subfigure[For $k=3$, $d=n-1$, codes in \cite{SenSasKum_itw} achieves our repair-matrix bound. The example here is $(n=6,k=3,d=5)$. ]{\label{fig:plot1a}\includegraphics[width=2.8in]{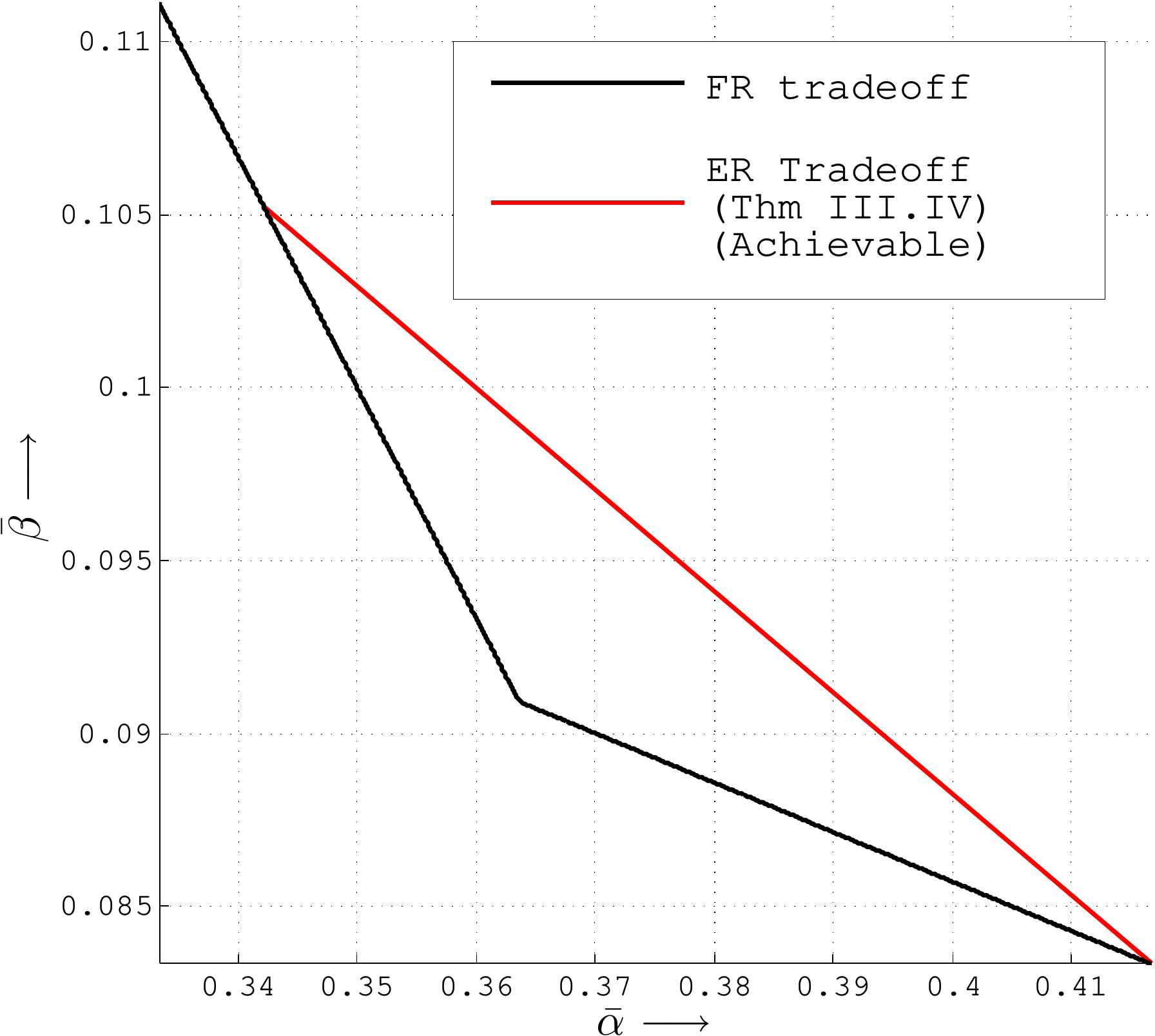}}
  \hspace{0.2in}
  \subfigure[For $k=d=n-1$, our outer bound matches the achievable region of layered codes, thus characterizing the tradeoff under linear setting. The example here is $(n=6,k=5,d=5)$.]{\label{fig:plot1b}\includegraphics[width=2.9in]{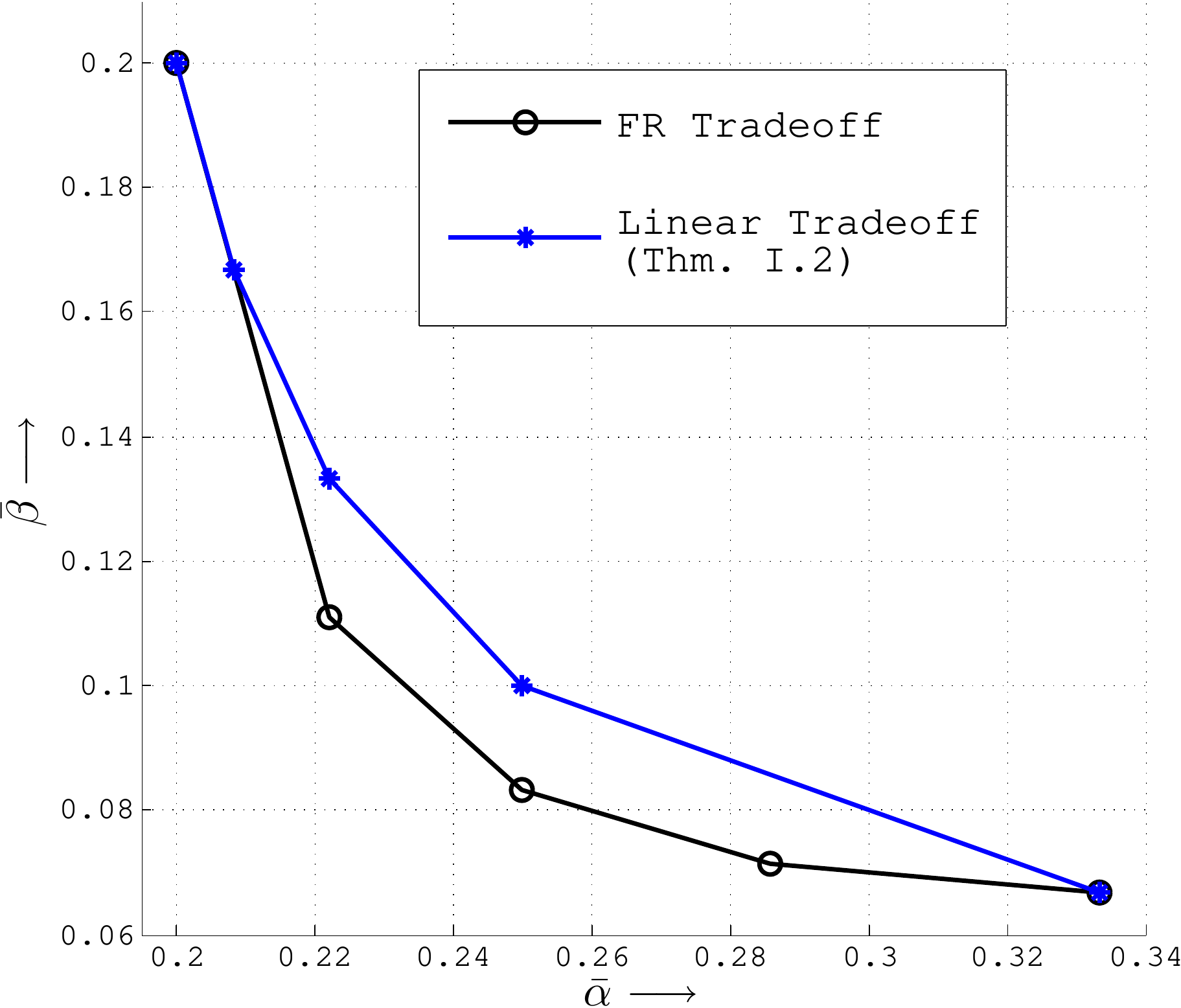}}
\caption{Characterization of normalized ER Tradeoff. \label{fig:plot1}}
\end{center}
\end{figure}

\begin{figure}[h!]
\begin{center}
  \subfigure[The example here is $(n=13,k=7,d=12)$. The combination of repair-matrix bound and improved Mohajer-Tandon bound performs better than other bounds given in the plot.]{\label{fig:plot2a}\includegraphics[width=2.9in]{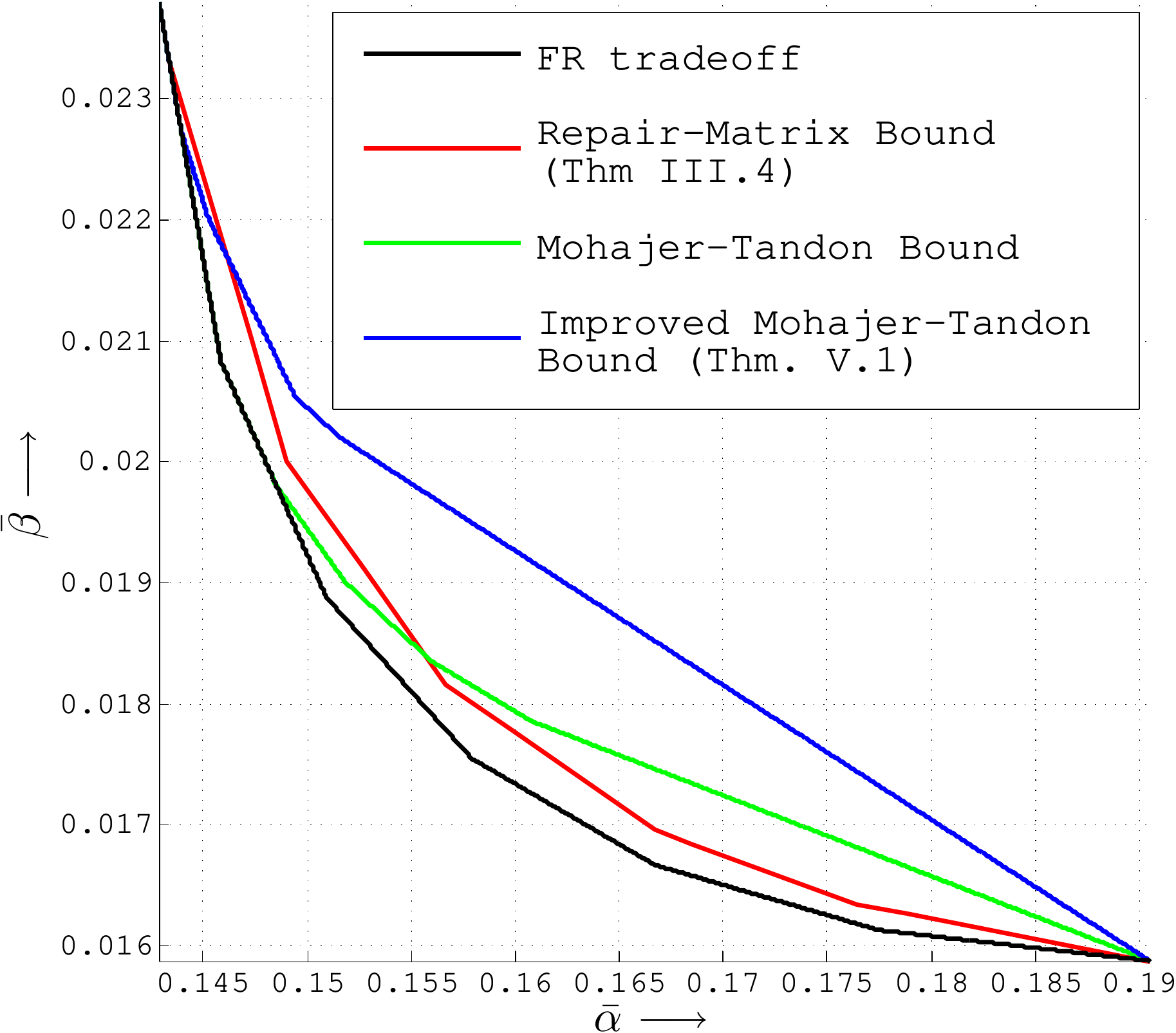}}
  \hspace{0.2in}
  \subfigure[Example here is $(n=6,k=d=5)$. When $k=d$, both Mohajer-Tandon and the improved Mohajer-Tandon bounds coincide.]{\label{fig:plot2b}\includegraphics[width=2.9in]{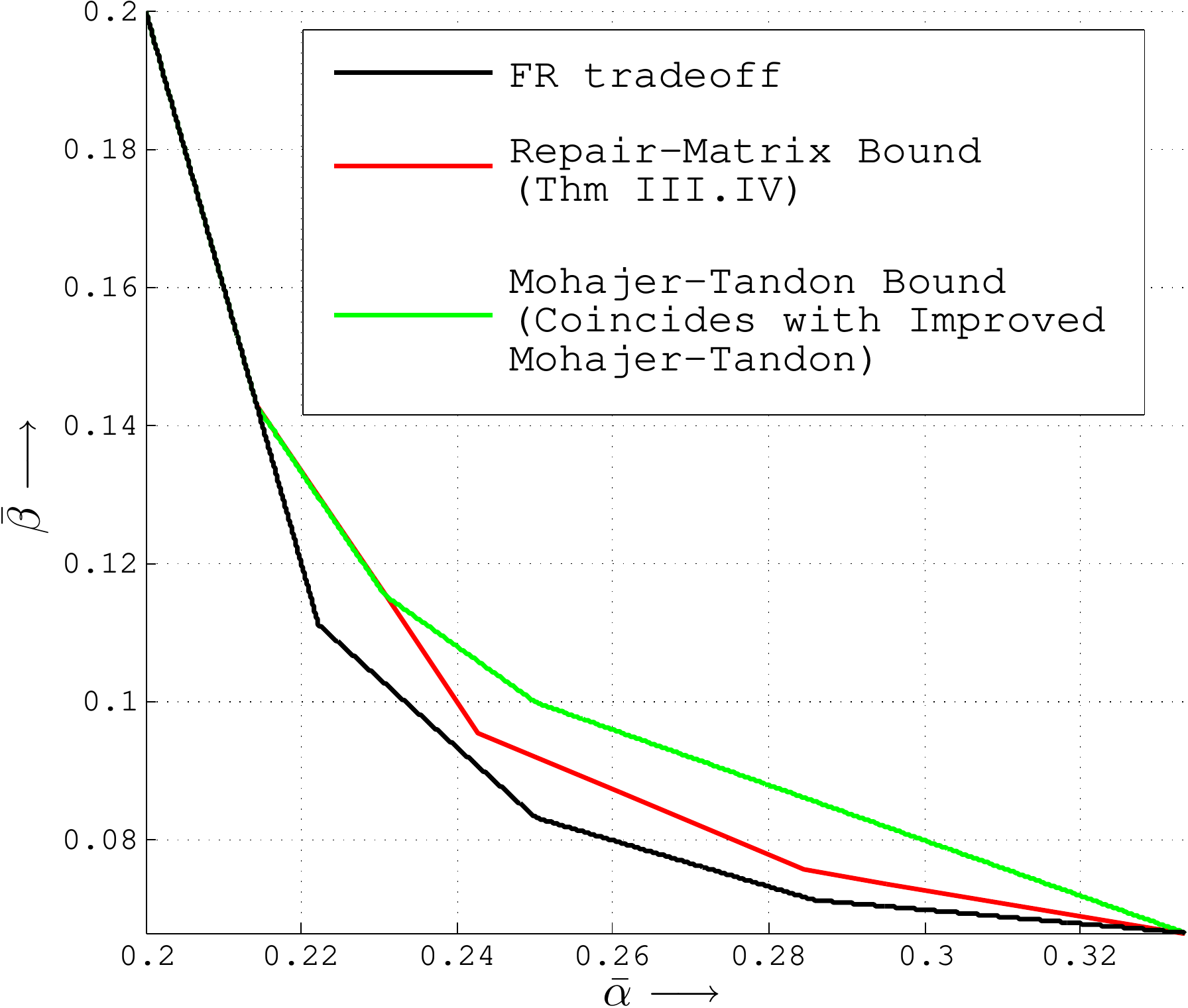}}
\caption{Performance comparison of various outer bounds.\label{fig:plot2}}
\end{center}
\end{figure}
In Fig.~\ref{fig:plot1}, we plot the cases in which our outer bounds characterize the normalized ER tradeoff. In Fig.~\ref{fig:plot2}, we do a performance comparison of various known bounds. 

%{\color{blue}  We need a statement somewhere which states that the present paper is a compilation of the results in repair-mx bound and Prakash-Krishnan bound together with an overview of results generated since and a recent improvement on the Mohajer-Tandon bound.} 

%\ben
%\item $k=d$, Mohajer bound is the best; repair-matrix bound coincides with that in a region close to msr
%\item $(n,n-1,n-1)$ Mohajer bound characterizes the tradeoff near mbr region $\beta \leq \frac{2\alpha}{k}$ using Goparaju inner bound.
%\item $k<d$ repair-matrix bound is better than Mohajer bound in near msr region. 
%\item Improved Mohajer bound is better than Mohajer bound whenever $k < d$, coincides with mohajer for $k=d$.
%\item But repair-matrix bound is better than even the improved mohajer bound in msr region.
%\item thus min of improved mohajer, and epsilon is the best bound so far.
%\item $(n,n-1,n-1)$ linear tradeoff is characterized.
%\een 

\subsection{Our Approach\label{sec:approach}}

The present paper derives outer bounds on the normalized ER tradeoff of a regenerating code with full-parameter-set ${\cal P}_f = \{(n,k,d),(\alpha, \beta)\}$. Since every ER code is an FR code, it is clear that the normalized ER tradeoff lies on or above and to the right, of the normalized FR tradeoff in the $(\bar{\alpha}, \bar{\beta})$-plane. When we say that the normalized ER tradeoff {\em lies above} the normalized FR tradeoff, we imply that for given $(n,k,d)$ there is at least one value of normalized parameter $\bar{\beta}_0$ such that the corresponding normalized values $\bar{\alpha}_{\text{ER}}$ and $\bar{\alpha}_{\text{FR}}$ satisfy $\bar{\alpha}_{\text{ER}} > \bar{\alpha}_{\text{FR}}$. An equivalent definition in terms of the file size $B$ is given as follows. For given $(n,k,d)$,  let $\hat{B}_0:=\hat{B}_{\text{opt}}(\alpha_0, \beta_0)$ denote the optimal FR file size at an operating point $(\alpha_0,\beta_0)$ with $\alpha_0 = (d-\mu)\beta_0 - \nu \beta_0$ as in \eqref{eq:op_point}. Thus $\left(\frac{\alpha_0}{\hat{B}_0},\frac{\beta_0}{\hat{B}_0}\right)$ is a point lying on the normalized FR tradeoff. Suppose that the maximum file size of an ER code as a function of $(\alpha,  \beta)$ is
\bean
B(\alpha, \beta) & = & \hat{B}(\alpha,\beta) - \epsilon(\alpha, \beta)
\eean
for some non-negative function $\epsilon(\alpha, \beta)$. Let $\epsilon_0 = \epsilon(\alpha_0,\beta_0)$. Then the normalized operating points $(\bar{\alpha}_{\text{ER}},\bar{\beta}_{\text{ER}})$ for an optimal ER code as given by 
\bean
\bar{\beta}_{\text{ER}} \ = \ \frac{\beta_0}{B(\alpha_0, \beta_0)} & = & \frac{1}{\frac{\hat{B}_0}{\beta_0} - \frac{\epsilon_0}{\beta_0}} \\
\bar{\alpha}_{\text{ER}} \ = \ \frac{\alpha_0}{B(\alpha_0, \beta_0)} & = & \frac{1}{\frac{\hat{B}_0}{\alpha_0} - \frac{\epsilon_0}{\alpha_0}} \ = \
 \frac{1}{\frac{\hat{B}_0}{\alpha_0} - \frac{\epsilon_0}{\beta_0} \frac{1}{(d-\mu-\nu)}} \\
\eean
will be bounded away from $\left(\frac{\alpha_0}{\hat{B}_0},\frac{\beta_0}{\hat{B}_0}\right)$ if $\left(\frac{\epsilon_0}{\beta_0}\right)$ does not vanish to zero. It follows that an upper bound on the file size $B$ of an ER code
\bean
B & \leq B_{\text{upper}} (\alpha, \beta),
\eean
such that 
\bea
\label{eq:nonvanishing} \lim_{\beta \rightarrow \infty} \frac{\hat{B}(\alpha,\beta) - B_{\text{upper}} (\alpha, \beta)}{\beta} & > & 0
\eea
for some $(\mu,\nu)$ will equivalently define a bound on the normalized ER tradeoff that lie strictly above the normalized FR tradeoff. Throughout the paper, our approach therefore will be to derive upper bounds on ER file size that satisfy the criterion in \eqref{eq:nonvanishing}.

%If the file size $B_0 = B(\alpha_0,\beta_0)$ achievable by an ER code at the same operating point is strictly less than that by an FR code such that
%\bea \label{eq:non-vanish}
%\frac{\hat{B}_{\text{opt}}(\alpha_0, \beta_0) - B_0 }{B_0} 
%\eea
%does not vanish to zero as $\beta_0 \rightarrow \infty$, then by considering the corresponding normalized operating point $(\frac{\alpha_0}{B_0}, \frac{\beta_0}{B_0})$, it follows that the normalized ER tradeoff lies above the normalized FR tradeoff. Therefore, any upper bound $B_{\text{upper}} (\alpha, \beta)$ on the ER file size as a function of operating point $(\alpha, \beta)$
%\bean
%B & \leq B_{\text{upper}} (\alpha, \beta),
%\eean
%such that $B_{\text{upper}} (\alpha_0, \beta_0)$ satisfies the condition \eqref{eq:non-vanish} when $B_0$ is replaced with $B_{\text{upper}} (\alpha, \beta)$, will lead to an outer bound on the normalized ER tradeoff. Throughout the paper, our approach will be to derive upper bounds on ER file size that satisfy the above criterion.

If the full parameter set of a regenerating code has $n>(d+1)$, then by restricting attention to a set of $(d+1)$ nodes, one obtains a regenerating code with $n=(d+1)$ with all other parameters remaining unchanged.  It follows from this that any upper bound on the size $B$ corresponding to full parameter set $\{(n=(d+1),k,d), (\alpha,\beta)\}$ continues to holds for the case $n>(d+1)$ with the remaining parameters left unchanged.  Keeping this in mind, we will assume throughout that $n=(d+1)$. 

A key technique used in the paper is to lower bound the difference $\epsilon = \hat{B}_{\text{opt}}(\alpha, \beta) - B(\alpha,\beta)$  between the file size of an optimal FR code and an ER code. The total information content in a regenerating code can be accumulated from a set $\{1,2,\ldots, k\}$ of $k$ nodes. The conditional entropy of the $(i+1)$-th node data conditioned on the data accumulated from previous $i$, $0 \leq i \leq k-1$ nodes is compared against the corresponding value of an optimal FR code, and the difference is defined to be $\omega_i$. It follows that $\epsilon$ is the sum of all $\{\omega_i\}_{i=0}^{k-1}$. Our approach is to relate  $\{\omega_i\}_{i=0}^{k-1}$ in terms of entropy of certain collections of repair data, and eventually find an estimate on $\epsilon$. Along the way, we construct a {\em repair matrix} as an arrangement of random variables corresponding to repair data in a $((d+1) \times (d+1))$-sized matrix. Many properties pertaining to the inherent symmetry of regenerating code become clear from the repair-matrix perspective, and we use it as a tool in our proofs. 

A different approach is used in deriving an upper bound on the ER file size of a linear regenerating code. Here we focus on a parity-check matrix $H$ of a linear ER code, and construct an augmented parity-check matrix $H_{\text{repair}}$ of size $(n\alpha \times n\alpha)$ that captures the exact-repair properties. A block-matrix structure is associated to $H_{\text{repair}}$, and thereby we identify $n$ thick columns $\{H_1, H_2, \ldots H_n\}$ of $H_{\text{repair}}$ with $H_i$ associated to the node $i$. Here we mean by a thick column a collection of $\alpha$ columns. Let us denote by $\delta_i$ the incremental rank added by $H_i$ to the collection of $(i-1)\alpha$ vectors in $\{H_j \mid 1 \leq j < i\}$. We estimate lower bounds on $\{\delta_i\}_{i=1}^{n}$ that will eventually lead to a lower bound on the rank of $H$. It is clear that the file size $B$ is the dimension of the code, and therefore a lower bound on the rank of $H$ results in an upper bound on the file size.

\subsection{Organization of the Paper}

In Sec.~\ref{sec:non-exist}, we describe the result of Shah et al. showing the non-existence of ER codes operating on the FR tradeoff. In Sec.~\ref{sec:bound1}, we present an upper bound on the ER file size. In Sec.~\ref{sec:oth_bounds}, we review the various upper bounds on ER file size that are known in the literature. In Sec.~\ref{sec:bound2}, we develop on the existing Mohajer-Tandon bound, and make an improvement upon that to get a better bound when $d > k$. In Sections~\ref{sec:linear_app},\ref{sec:544},\ref{sec:main_proof}, we focus on upper bounds on file size under linear setting. We characterize the normalized ER tradeoff for the case $(n,k=n-1,d=n-1)$ in Sec.~\ref{sec:main_proof}, while the proof techniques are illustrated for a particular case of $(n=5,k=4,d=4)$ in Sec.~\ref{sec:544}. In Sec.~\ref{sec:achievability}, we discuss the achievability of the outer bounds on normalized ER tradeoff derived at earlier sections. 

\section{The Non-existence of ER Codes Achieving FR tradeoff\label{sec:non-exist}}

As mentioned in Sec.~\ref{sec:results}, it was shown in \cite{ShaRasKumRam_rbt} that apart from the MBR point and a small region adjacent to the MSR point, there do not exist ER codes whose $(\alpha, d\beta)$ values correspond to coordinates of an interior point on the FR tradeoff. The theorem in \cite{ShaRasKumRam_rbt} due to Shah et al. is stated below.
\bthm \label{thm:shah_non_exist}(Theorem 7 in \cite{ShaRasKumRam_rbt}) For any given values of $(n,k \geq 3,d)$, ER codes do not exist for the parameters $(\alpha, \beta, B)$ lying at an interior point on the FR tradeoff except possibly for the case
\bea \label{eq:exception}
(d-k+1)\beta & \leq \ \alpha \ \leq & \left[(d-k+2) - \frac{d-k+1}{d-k+2}\right] \beta.
\eea
\ethm
The region 
\bean
\{ (\alpha, \beta) \mid (d-k+1)\beta \ \leq \ \alpha \ \leq \ \left[(d-k+2) - \frac{d-k+1}{d-k+2}\right] \beta\}
\eean on which the theorem does not claim the non-existence of ER codes is referred to as the {\em near-MSR region}.The Theorem~\ref{thm:shah_non_exist} however did not rule out the possibility of approaching the FR tradeoff asymptotically i.e., as the file size $B \rightarrow \infty$. As mentioned earlier, this question was answered by Tian in the negative in \cite{Tia} for the specific case when $(n,k,d)=(4,3,3)$. 

In this section, we will describe the approach taken by Shah et al. in proving Theorem~\ref{thm:shah_non_exist} in terms of the notation to be used in the present paper. We begin with some notation and definitions. Let $\mathcal{C}$ be an ER regenerating code over $\mathbb{F}$ having file size $B$ and full-parameter set ${\cal P}_f = \{(n,k,d), (\alpha, \beta)\}$.  We regard the message symbols as a collection of $B$ random variables taking on values in $\mathbb{F}$ and use $M$ to denote the $(1 \times B)$ random vector whose components are the $B$ message symbols.  We use $p_M(\cdot)$ to denote the joint probability distribution of the $M$ random variables.  All other random variables pertaining to the regenerating code are functions of the components of $M$, and satisfy probability distributions that are induced by $p_M$. 

We will use $[i], 1 \leq i \leq n$ to denote the set $\{1, 2, \ldots, i \}$ and define $[0]$ to be the empty set $\phi$. For $1 \leq i \leq j \leq n$, we use $[i \ j]$ to denote the set $\{i, i+1, \ldots, j \}$. Whenever we write $[i \ j]$ with $i > j$, it will be assumed to be the empty set.   On occasion, we will run into a set of random variables of the form $W_A$ where $A$ is the empty set,  $W_A$ should again be interpreted as the empty set. 
 
\subsection{The Repair Matrix and the Constraints Imposed By Exact-Repair\label{sec:rep_mat}}

As made clear in Sec.~\ref{sec:approach}, we assume that $n=d+1$ without loss of generality. Let $W_x, 1 \leq x \leq n$ denote the random variable corresponding to the contents of a node $x$.   Given a subset $A \subseteq [n]$, we use \bean
W_A & = & \{W_x \mid x \in A \} 
\eean
to denote the contents of nodes indexed by $A$.  Clearly, 
\bea
H(W_x) & \leq & \alpha. \label{eq:capacity_alpha}
\eea
Let $S_x^y$, $x,y \in [n], x \neq y$ denote the random variables corresponding to the helper data sent by the helper node $x$ to the replacement of a failed node $y$. This is well defined because under the assumption $n=(d+1)$, there is just one set of $d$ helper nodes for any failed node. Given a pair of subsets $X,Y \subseteq [n]$, we define $S_X^Y \ = \ \left\{ S_x^y \mid x \in X, y \in Y, x \neq y \right\}$. We use the short-hand notation $S_X$ to indicate $S_X^X$. From the definition of regenerating codes, it follows that
\bea
H(S_x^y) \ \leq \ \beta. \label{eq:capacity_beta}
\eea
In (\ref{eq:capacity_alpha}, \ref{eq:capacity_beta}), information is measured in units of $\log_2(|\mathbb{F}|)$ bits.  The collection of random variables $\{S_x^y \mid x \in [d+1], y \in [d+1], x \neq y \}$ can schematically be represented using a $(d+1) \times (d+1)$ matrix ${\cal S}$  with empty cells along the diagonal as shown in Fig.~\ref{fig:repairmatrix}.  The rows in this matrix correspond to the helper nodes and the columns to nodes undergoing repair. The $(x,y)$th entry of this matrix, thus corresponds to $S_x^y$.  We will refer to ${\cal S}$ as the {\em repair matrix}. The subset of ${\cal R}$ appearing below the diagonal and above the diagonal are denoted by ${\cal R}_L$ and ${\cal R}_U$ respectively.
%\begin{figure}[ht]
%\centering
%\includegraphics[height=2in]{IJICT_RepairMatrix_1.pdf}
%\caption{\scriptsize Illustrating repair matrix.}\label{fig:repairmatrix}
%\end{figure}

\begin{figure}[h!]
\begin{center}
  \subfigure[Illustration of the repair matrix.]{\label{fig:repairmatrix}\includegraphics[width=2.65in]{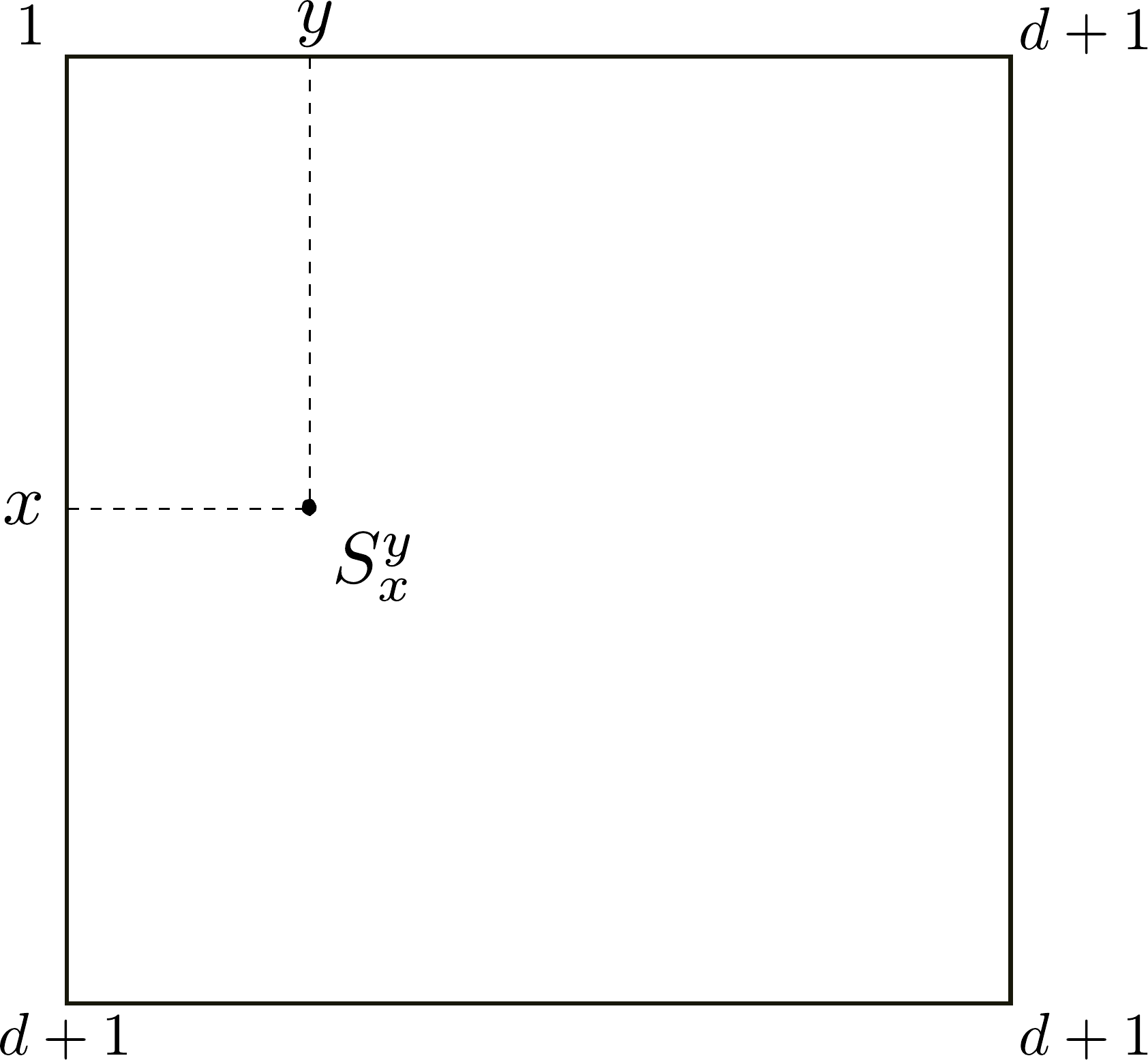}}
  \hspace{0.2in}
  \subfigure[The trapezoidal configuration]{\label{fig:trapezium}\includegraphics[width=2.9in]{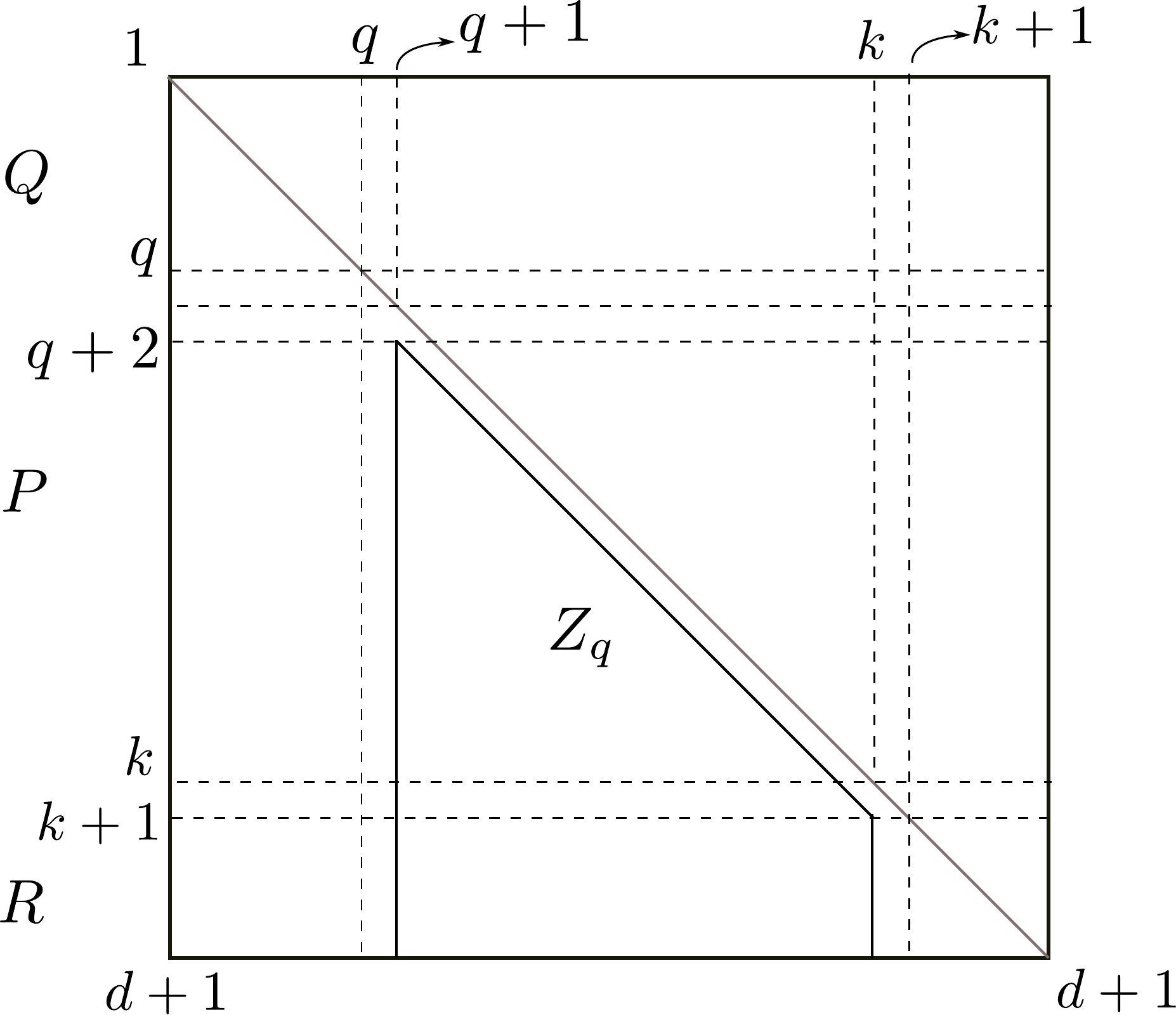}}
\caption{The repair matrix and the trapezoidal configuration}
\end{center}
\label{fig:repmat_trapez}
\end{figure}

Apart from the constraints given in \eqref{eq:capacity_alpha}, \eqref{eq:capacity_beta}, the requirements of data reconstruction and exact-repair impose  further constraints. The constraint due to data reconstruction is given by either of the following two equivalent statements:
\bea
H(W_A) & = & B, \ |A| \geq k \label{eq:data_collection_1}, \\
H(M \mid W_A) & = & 0, \ |A| \geq k \label{eq:data_collection_2}.
\eea
For every $i \in [n]$, the exact-repair condition imposes the constraint
\bea
H(W_i \mid S_{\cal D}^i) & = & 0, \ |{\cal D}| = d, \ i \notin {\cal D} \label{eq:exact_repair}.
\eea

\subsection{Trapezoidal Configurations in the Repair Matrix \label{sec:trapezium}}

Throughout the discussion taking place in Sections up to \ref{sec:bounds_file_size}, we will assume that there is a fixed numbering of the $n=(d+1)$ nodes in the network. In \eqref{eq:data_collection_1}, the file size $B$ is expressed as the joint entropy of a collection $k$ random variables $\{W_1, W_2, \ldots, W_k\}$. It is possible to express $B$ as the joint entropy of other subsets of random variables, in particular those involved in node repair. An example, important for the discussion to follow, appears below. Let $q$ be an integer lying in the range $0 \leq q \leq k$ and set 
\bean
Q & = & \{1,2,\cdots,q\} \\
P& = & \{q+1, q+2, \cdots, k\} \\
R & = & \{k+1, k+2, \cdots (d+1)\} .
\eean
Note that $Q,P,R$ are all functions of the integer $q$.  When $q=0$, we will set $Q$ to be the empty set $\phi$.   Note that $P = [k] \setminus Q$ and $R = [k+1 \ d+1]$. We define:
\bea
Z_q & = & {\cal R}_L \cap S_{[d+1]}^P \\
X_q  & = & {\cal R}_L \cap S_P.
\eea
Then we can write $B$ as:
\bea
 \nonumber B & = & H(W_Q, W_P) \\
 \label{eq:trapezoid1a} \nonumber & = & H(W_Q, W_P, Z_q)\\
\nonumber & = & H(W_Q, Z_q) + H(W_P \mid W_Q,Z_q) \\
\nonumber & = & H(W_Q, Z_q) 
\eea
where \eqref{eq:trapezoid1a} follows from the exact-repair condition \eqref{eq:exact_repair}. The collection $Z_q$ of random variables forms a trapezoidal region within the repair matrix as shown in Fig.\ref{fig:trapezium}. We refer to $(W_Q, Z_q)$, $q \in \{0, 1, \ldots, k\}$  as a {\em trapezoidal configuration}. The set $Z_q$ is said to be the {\em trapezoid} corresponding to the trapezoidal configuration $(W_Q, Z_q)$. It is clear that $Z_q = X_q \uplus S_R^P$. Next we proceed to define a {\em sub-trapezoid} of the trapezoid $Z_q$. Let $T =\{q+1, q+3, \ldots, q+t\} \subseteq P$ be a subset of size $0 \leq t \leq k-q$ of $P$. Then we define the subset $Z_{q,t}$ of $Z_q$ as:
\bean
Z_{q,t} & := & {\cal R}_L \cap S_{[d+1]}^T.
\eean 
The set $Z_{q,t}$ also forms a trapezoidal region in ${\cal R}$ and is called a sub-trapezoid of the trapezoid $Z_q$. Here again, we define $X_{q,t}$ as:
\bean
X_{q,t} & := & S_T \cap Z_{q,t},
\eean
and it follows that $Z_{q,t} = X_{q,t} \uplus S_{R\cup (P\setminus T)}^T$. A sub-trapezoid is illustrated in Fig.~\ref{fig:sub-trapezoid}. 
\begin{figure}[ht]
\centering
\includegraphics[width=2.8in]{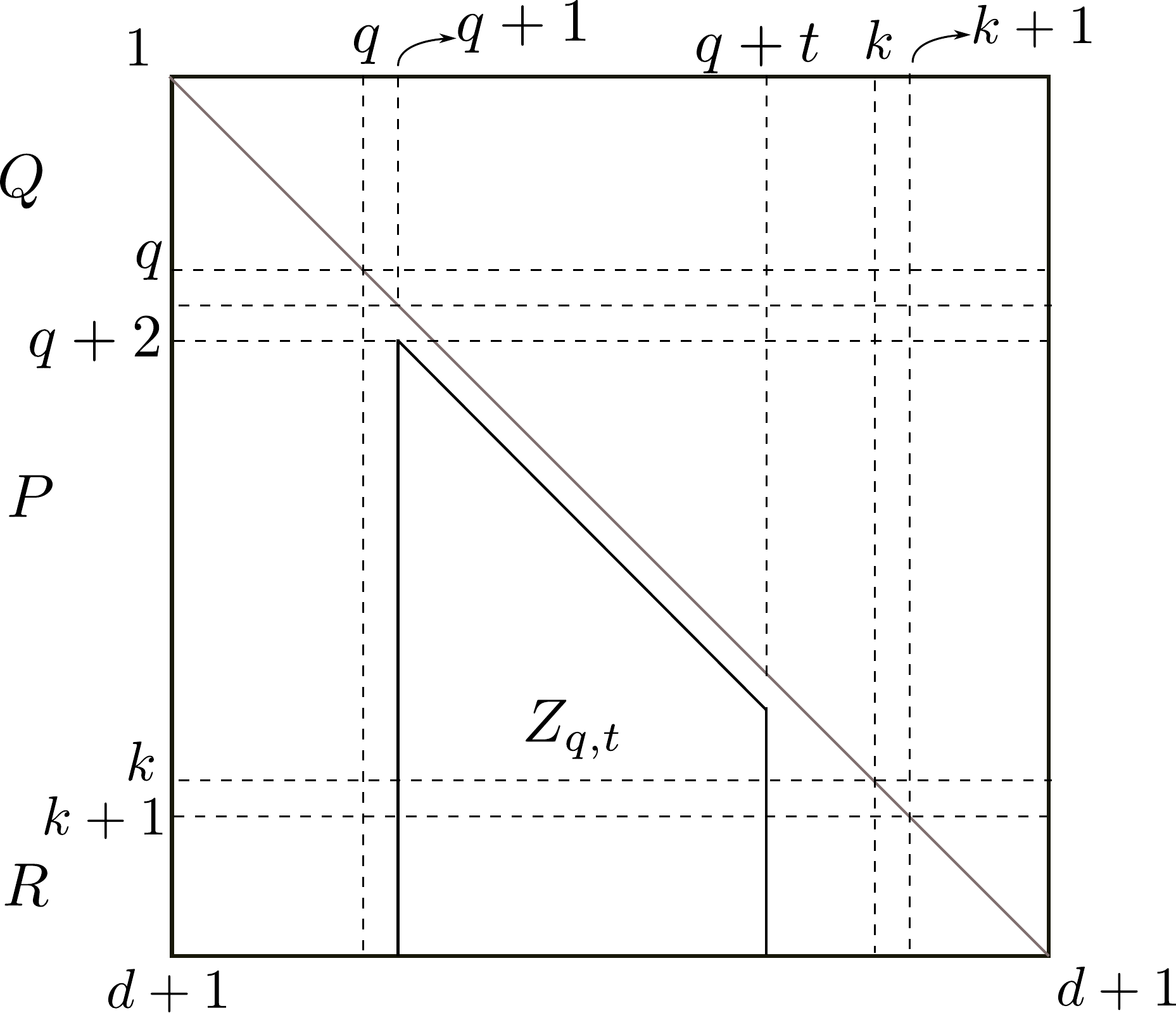}
\caption{Illustration of the sub-trapezoid $Z_{q,t}$.}
\label{fig:sub-trapezoid}
\end{figure}

For every trapezoidal configuration $(W_Q,Z_q)$ indexed by $q = 0, 1, \ldots, k$, we have the identity
\bea \label{eq:trapezoid2}
B & = & H(W_Q, Z_q),
\eea
and the corresponding inequality obtained by repeatedly applying the union bound $H(X_1,X_2) \leq H(X_1)+H(X_2)$, i.e., 
\bea \nonumber 
B & \leq & H(W_Q) + H(Z_q) \\ 
\label{eq:trapezoid_bound} & \leq  & H(W_Q) + H(X_q) + H(S_R^P) \\
\label{eq:Bq_bound}  & \leq & q\alpha + {k-q \choose 2}\beta + (d+1-k)(k-q)\beta .
\eea
We define for $q \in \{0,1,2\cdots, k\}$, the quantities: 
\bean
B_q & := & q\alpha + {k-q \choose 2}\beta + (d+1-k)(k-q)\beta .
\eean

\subsection{The Argument For Non-existence \label{sec:hnp}}

Let us consider an ER code operating at the point $(\alpha, \beta)$ satisfying $\alpha = (d-\mu)\beta$. For this value of $\alpha$, as shown below, the FR bound gives us $B_{\mu+1}$ as the upper bound on file size: 
\bean
B & \leq & \sum_{i=0}^{k-1}\min \{\alpha, (d-i)\beta\} \\
& = & (\mu+1)\alpha \ + \ \sum_{i=\mu+1}^{k-1}(d-i)\beta \\ 
& = & (\mu+1)\alpha \ + \ \sum_{j=0}^{k-\mu-2}(d-k+1+j)\beta \\ 
& = & (\mu+1)\alpha \ + \ (d-k+1)(k-\mu-1)\beta \ + \ {k-\mu-1 \choose 2} \beta \\
& = & B_{\mu+1} .
\eean
Thus if an ER code is optimal with respect to the FR tradeoff at the point $\alpha=(d-\mu)\beta$, from equations \eqref{eq:trapezoid2} and \eqref{eq:trapezoid_bound}, with $q=(\mu+1)$, one obtains that such a code must satisfy: 
\bea \label{eq:union_bd_cond}
H(Z_{\mu + 1} \mid W_{[\mu]}) & = & H(Z_{\mu + 1}) \ = \ {k-\mu - 1 \choose 2}\beta + (d+1-k)(k-\mu-1)\beta, 
\eea
i.e., the union bound on $Z_{\mu + 1}$ must hold with equality. That means that all the random variables in $Z_{\mu +1 }$ are mutually independent.   However, it is shown by Shah et al. in \cite{ShaRasKumRam_rbt} that this is not possible if an ER code lies at an interior point except for the near-MSR region and the MBR point.  To prove this result, the authors of \cite{ShaRasKumRam_rbt} focus on a subset $S_m^L$ of the repair matrix where $m \in [n]$ and $L \subseteq [n]$ are arbitrarily chosen from $[n]$ while satisfying the conditions $|L| := \ell < k$ and $m \notin L$. The subset $S_m^L$ is of course, the union of helper data sent by a single node $m$ to the nodes in $L$.    We can write
\bea
\nonumber H(S_m^L) & = & H(S_m^L \mid W_L) + I(S_m^L : W_L) \\
\label{eq:ia} & \leq & H(S_m^L \mid W_L) + I(W_m : W_L).
\eea
It can be shown that (see \cite{ShaRasKumRam_rbt})
\bea \label{eq:ia_cancel}
H(S_m^L \mid W_L) = 0 , \ \ell \geq \mu + 1, 
\eea
and that 
\bea \label{eq:rel_beta}
I(W_m : W_L) = \beta , \ \ell = \mu + 1.
\eea
As a consequence, we have that 
%The first term $H(S_m^L \mid W_L) $ on the RHS of \eqref{eq:ia} quantifies the interference in the helper data sent to the nodes in $L$, that can not be cancelled even with the knowledge of the contents of $(\ell -1)$ other nodes. The second term $I(S_m^L : W_L)$ quantifies the relevant information sent by the helper node $m$ to the the subset $L$. If an optimal FR code satisfies ER constraints, it is shown that (a) the interference can be completely cancelled when $\ell \geq (\mu + 1)$ i.e.,
%
%and (b) the mutual information between a single node and a collection of $\mu +1$ nodes is limited by $\beta$ i.e.,
%\bea \label{eq:rel_beta}
%I(W_m : W_L) = \beta , \ \ell = \mu + 1.
%\eea
%
%This set corresponds to a line segment in the repair matrix.
%
%Therefore we have
\bea \label{eq:row_beta}
H(S_m^L ) = \beta, \ \ell = \mu + 1.
\eea
It follows that  
\bean
H(S_m^J ) \leq \beta, \text{ for any $J \subseteq [n]$ with $\mid J \mid < \mu + 1$}. 
\eean
In particular this is true of $J$ is of size $|J|=2$.  On the other hand, optimality with respect to the FR bound assumes that each row in the trapezoidal region $Z_q$ has joint entropy equal to the number of repair random variables $S_x^y \in Z_q$ belonging to the row, times $\beta$.  The bottom row of the trapezoid has $(k-\mu-1)$ entries and thus we clearly have a contradiction whenever $(k-\mu-1)\geq 2$. The argument does not go through when $(k-\mu-1)\leq 1$, i.e., when $\mu \geq k-2$. This necessary condition on $\mu$ underlies the fact that the non-existence of ER codes do no hold good in the near-MSR region. The proof given here is for the case when $\alpha=(d-\mu)\beta$ is a multiple of $\beta$.  This proof can be extended to the general case $\alpha=(d-\mu)\beta - \theta$, for $0 < \theta < \beta$ as well.  In the next section, we will exploit this contradiction to derive an upper bound on the file size of an ER code. 

\section{An Upper Bound on the ER File Size \label{sec:bounds_file_size}}

In this section, we show that for {\em any} value of the parameter set $(n,k,d)$, the ER tradeoff lies strictly above the FR tradeoff, a result that was first established in \cite{SasSenKum_isit}. As explained in Sec.~\ref{sec:approach}, we do this by deriving a tighter bound on file size $B$ in the case of ER than is true under FR. 

As mentioned in Sec.~\ref{sec:hnp}, our approach to bounding the file size $B$ is based on deriving estimates for the joint entropy of subsets of the repair matrix. First, we assume the existence of an ER code having parameters $(n,k,d),(\alpha,\beta)$ whose file size $B$ is of the form $B=\hat{B}-\epsilon$ for some $\epsilon \geq 0$, where $\hat{B}$ is the file size of an optimal FR code having the same parameter set ${\cal P}$. Next, we proceed to estimate the joint entropy of the subset $Z_q$ corresponding to a trapezoidal configuration $(W_Q, Z_q)$. We estimate the joint entropy in two different ways and show that the two estimates are in contradiction unless the value of $\epsilon$ lies above a threshold value $\epsilon_{\min}$.  This allows us to replace $B-\epsilon_{\min}$ as the revised bound on the file size under ER.   We will also show that $\epsilon_{\min}$ does not vanish as $\beta \rightarrow \infty$.  

\subsection{Preliminaries \label{sec:prelim}}

Consider an optimal FR code ${\cal \hat{C}}$ possessing the same set of parameters ${\cal P}$ as the ER code ${\cal C}$. In what follows, given any deterministic or random entity associated with $\mathcal{C}$, we will use a hat to denote the corresponding entity in ${\cal \hat{C}}$. For example, $\hat{B}$ denotes the file size of ${\cal \hat{C}}$. With this, we can write
\bean
\sum_{i=0}^{k-1} \min\{\alpha,(d-i)\beta\} & = & \hat{B} \ = \ H(\hat{W}_{[k]}) \nonumber \\
& = & \sum_{i=0}^{k-1} H(\hat{W}_{i+1} \mid \hat{W}_{[i]} ) \\
& \leq & \sum_{i=0}^{k-1} \min\{\alpha,(d-i)\beta\} .
\eean
It follows that in an optimal FR code ${\cal \hat{C}}$, we must have
\bean
H(\hat{W}_{i+1} \mid \hat{W}_{[i]} ) & = & \min\{\alpha,(d-i)\beta\}, \ 0 \leq i \leq (k-1) .
\eean
Next, for $0 \leq i \leq k-1$, let us set:  
\bean
\gamma_i & = & \min \{ \alpha, (d-i)\beta\} , \\
\omega_i  & = & \gamma_i- H(W_{i+1} \mid W_{[i]}),
\eean
where $\omega_i$ measures the drop in the conditional entropy $H(W_{i+1} \mid W_{[i]})$ of an ER code in comparison with its value $H(\hat{W}_{i+1} \mid \hat{W}_{[i]})$ in the case of an optimal FR code. A plot of $\gamma_i$ as a function of $i$ for a given operating point $(\alpha, \beta)$ with $\alpha = (d-\mu)\beta -\theta$, appears in Fig.~\ref{fig:gamma}.   We also note the following identities:
\bea 
\label{eq:basic1} \epsilon & = & \sum_{i=0}^{k-1} \omega_i , \\
\label{eq:basic2} H(W_B \mid W_A)  & = &  \sum_{i =a}^{a+b-1} (\gamma_i - \omega_i),
\eea
where $A=[a]$ and $B=[a+1 \ a+b]$ and $0 \leq a \leq a+b \leq k$. The lemma below follows from these identities. 
\blem \label{lem:colsum} Let $(Q, Z_q)$ be a trapezoidal configuration for some $q \in \{ 0, 1, \ldots, k \}$, and let $Z_{q,t} \subseteq Z_q$ be a sub-trapezoid with $0 \leq t \leq k-q$. Then
\bean
H(Z_{q,t} \mid W_Q) & \geq & \sum_{i =q}^{q+t-1} (\gamma_i - \omega_i)
\eean
\elem
\bpf By the exact-repair condition, $H(Z_{q,t} \mid W_Q) $ is at least $H(W_{[q+1 \ q+t]} \mid W_Q) $ and the result follows from \eqref{eq:basic2}.
\epf
\begin{figure}[ht]
\centering
\includegraphics[height=2in]{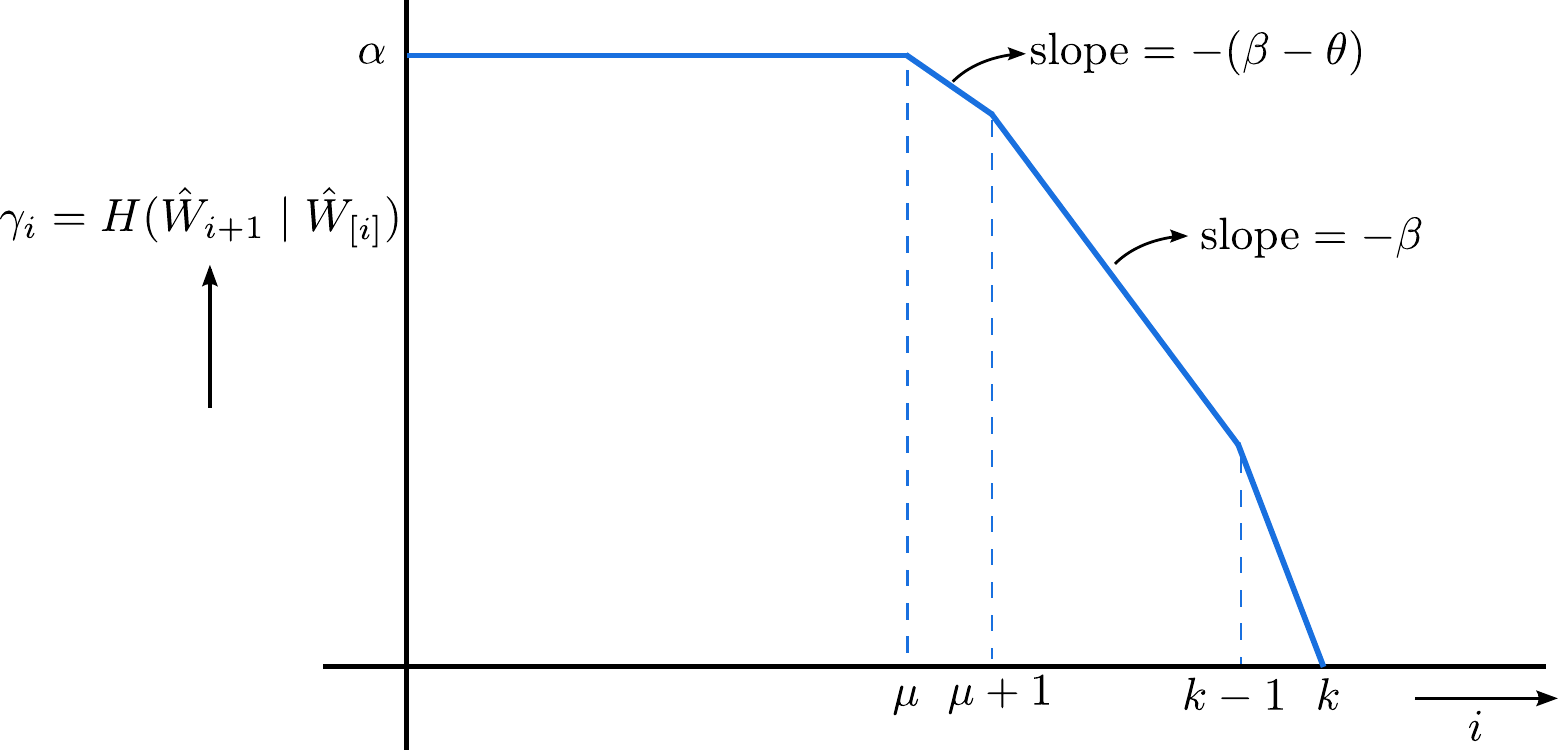}
\caption{The function $\gamma_i$ versus $i$ for $\alpha=(d-\mu)\beta - \theta$.}\label{fig:gamma}
\end{figure}

\subsection{Upper Bounds On Joint Conditional Entropies Of Repair Data \label{sec:ubounds_ent}}

Let $Q = [q]$, and $M, L$, be two mutually disjoint subsets of $[d+1]\setminus Q$ with $\ell := |L|$, and $m := |M|$. Then we can write
\bea \label{eq:ia_equation}
H(S_M^L \mid W_Q) & = & H(S_M^L \mid W_V, W_Q) + I(S_M^L : W_V \mid W_Q),
\eea
where in we take $V \supset L$ as a superset of $L$ with $V \cap M = \phi$ and $v:= |V|$. Our next objective is to estimate $H(S_M^L \mid W_V, W_Q)$ and $I(S_M^L : W_V \mid W_Q)$ in order to obtain an upper bound on $ H(S_M^L \mid W_Q) $.
\blem \label{lem:ubound} Suppose $\alpha = (d-\mu)\beta -  \theta$ with $\mu \in \{0,1,\ldots, k-1\}$ and $\theta \in [0,\beta)$ except when $\mu = k-1$. Then for $2 \leq \ell \leq v < k-q$,
\bean
H(S_M^L \mid W_V, W_Q) & \leq &  \left\{ \begin{array}{lc} \ell \theta + \ell \omega_{v-1+q} , &  v= \mu + 1 - q  \\
					\ell \omega_{v-1+q} , &  v > \mu + 1 -q . \end{array} \right.
\eean
\elem
\bpf Let $\ell_0 \in L$, and by symmetry $H(S_M^{\ell_0} \mid W_V, W_Q)$ is same for every $\ell_0 \in L$. Define $\tilde{V} = V \setminus \{\ell_0\}$. Then we have
\bean H(S_M^L \mid W_V, W_Q) & \leq & \ell H(S_M^{\ell_0} \mid W_V, W_Q) \\
 & = & \ell \{ H(S_M^{\ell_0},W_{\ell_0} \mid W_{\tilde{V}}, W_Q) - H(W_{\ell_0} \mid W_{\tilde{V}}, W_Q) \} \\
 & = & \ell \{ H(S_M^{\ell_0} \mid W_{\tilde{V}}, W_Q) + H(W_{\ell_0} \mid S_M^{\ell_0}, W_{\tilde{V}}, W_Q) - H(W_{\ell_0} \mid W_{\tilde{V}}, W_Q) \}  
\eean
By substituting bounds, we obtain for the case $v-1+q > \mu$
\bean
H(S_M^L \mid W_V, W_Q) & \leq & \ell \{ m\beta + (d-v +1-q-m)\beta - (d-v +1-q)\beta + \omega_{v - 1+q} \}  \\
& = & \ell \omega_{v - 1+q}  \ ,
\eean
and for the case $v - 1 + q = \mu$, 
\bean
H(S_M^L \mid W_L, W_Q) & \leq & \ell \{ m\beta + (d-v +1-q-m)\beta - (d-v +1-q)\beta + \theta + \omega_{v - 1+q} \}  \\
& = & \ell \theta + \ell \omega_{v - 1+q} \ .
\eean
\epf
We remark here that in \cite{Duursma2014} the quantity $H(S_M^L)$ is considered for obtaining a bound on ER file size. Our approach here is different in the sense that we estimate $H(S_M^L)$ in terms of $\{ \omega_i\}_{i=0}^{k-1}$. The second term in \eqref{eq:ia_equation} can also be easily estimated in terms of $\{\gamma_i, \omega_i \}_{i=0}^{k-1}$:
\bea
\nonumber I(S_M^L :W_V \mid W_Q) &  \leq & I(W_M :W_V \mid W_Q) \\
\nonumber & = & H(W_M \mid W_Q) - H(W_M \mid W_{Q\cup V}) \\
\label{eq:rel} & = & \left[ \sum_{i=q}^{q+m-1} (\gamma_i - \omega_i) \right] - \left[ \sum_{i=q+v}^{q+v+m-1} (\gamma_i - \omega_i) \right] .
\eea
The Lemma~\ref{lem:ubound} along with \eqref{eq:rel} allows us to bound $H(S_M^L \mid W_Q)$ from above given an operating point $\alpha = (d-\mu)\beta - \theta$. Calculations for the particular case of $q=0, m=1$ taking values for $v$ in $\{\mu+1, \mu+2\}$ result in the following corollary.
\bcor \label{cor:rowbound} Let $\alpha = (d-\mu)\beta - \theta$. Then for $m \notin L$ and $\ell = |L|$, we have
\bea
\label{eq:rb1} H(S_m^L) & \leq & \beta + (\ell -1)\theta  + (\ell -1)\omega_{\mu} + (\omega_{\mu} + \omega_{\mu+1}) , \ \ 2 \leq \ell \leq \mu +1 \\
\label{eq:rb2} H(S_m^L) & \leq & 2\beta - \theta + (\ell -1)\omega_{\mu+1} + (\omega_{\mu+1 + \omega_{\mu+2}) }, \ \ 2 \leq \ell \leq \mu +2 .
\eea
\ecor

\subsection{The Bound On ER File Size \label{sec:bound1}}

In this section, we make use of Lem.~\ref{lem:colsum} and Cor.~\ref{cor:rowbound} to derive an upper bound on the file size $B$ of an ER code. This will also translate to an outer bound for the ER tradeoff.

\bthm \label{thm:bound1} Let $B$ denote the file size of a ER regenerating code with full-parameter set ${\cal P}_f = \{(n,k,d),(\alpha,\beta)\}$. Let $\alpha = (d-\mu)\beta - \theta$. Then the ER file size $B$ is upper bounded by:
\ben
\item For $\mu=0, \ 0 < \theta < \beta$,
\bean
B & \leq & \hat{B} - \epsilon_1
\eean
\item For $\mu \in \{ 1, 2, \ldots, k-3 \}, \ 0 \leq \theta < \beta$,
\bean
B & \leq & \hat{B} - \max \{ \epsilon_0, \epsilon_1 \}
\eean
\item For $\mu=k-2, \ 0 \leq \theta < \left(\frac{d-k+1}{d-k+2}\right)\beta$,
\bean
B & \leq & \hat{B} - \epsilon_0, 
\eean
\een
where $\epsilon_0$ and $\epsilon_1$ are as given in Tab.~\ref{tab:eps_table}.
\ethm
\bpf The proof is relegated to the Appendix.
\epf

\begin{table}[h]
\centering

\begin{tabular}{||c|c||}  \hline
\hline
 &   \\
Regime of $(\mu,\theta)$ & Lower bounds $\epsilon_0$ , $\epsilon_1$  on $\epsilon = \hat{B}-B$ \\
 &   \\
\hline
\hline
 &   \\
$ \begin{array}{c} \mu \in \{ 1, 2, \ldots, k-2 \} \text{ for all } \theta \\
\text{For } \mu=k-2, \ \theta < \frac{d-k+1}{d-k+2}\beta
\end{array}$ & 
\large
$\begin{array}{lcl}
& & \text{Let } r_0 = \left\lfloor \frac{k-\mu}{\mu+1} \right\rfloor  \\
&& \\
\epsilon_0 & = & \left\{ \begin{array}{lc} \frac{(d-k+1)(k-\mu-1)(\beta - \theta) \ - \ \theta}{(d-k+1)(k-\mu) \ + \ 1}, & k-\mu < \mu+1. \\
& \\
\frac{\left(d-\frac{(\mu+1)(r_0+3)}{2}+2 \right)r_0\mu(\beta - \theta) \ - \ \theta}{\left(d-\frac{(\mu+1)(r_0+3)}{2}+2\right)r_0(\mu+1) \ + \ 1}, & k-\mu \geq \mu+1 . \end{array} \right. \end{array}$ \\
 &   \\
\hline
 &   \\
$ \begin{array}{c} \mu \in \{ 0, 1, \ldots, k-3 \} \text{ for all } \theta \\
\text{For } \mu=0, \ \theta \neq 0
\end{array}$ & 
\large
$\begin{array}{lcl}
& & \text{Let } r_1 = \left\lfloor \frac{k-\mu-1}{\mu+2} \right\rfloor  \\
&& \\
\epsilon_1 & = & \left\{ \begin{array}{lc} \frac{(d-k+1)\left[(k-\mu-3)\beta \ + \ \theta\right]}{(d-k+1)(k-\mu-1) \ + \ 1}, & k-\mu-1 < \mu+2. \\
& \\
\frac{\left(d-\frac{(\mu+2)(r_1+3)}{2}+2 \right)r_1 \left[\mu\beta \ + \ \theta\right] }{\left(d-\frac{(\mu+2)(r_1+3)}{2}+2\right)r_1(\mu+2) \ + \ 1}, & k-\mu-1 \geq \mu+2. \end{array} \right. \end{array}$ \\
 &   \\
\hline \hline
\end{tabular}
\caption{Lower Bounds on the quantity $\hat{B}-B$}
\label{tab:eps_table}
\end{table}

\bcor When $k \geq 3$, the normalized ER tradeoff is strictly away from the normalized FR tradeoff for all normalized operating points $(\bar{\alpha},\bar{\beta})$ with $\bar{\alpha} = (d-\mu)\bar{\beta} - \nu \bar{\beta}$ such that $(\mu, \nu)$ falls in the range $(\mu = 0, \ 0 \ < \ \nu \  < \ 1)$, $(\mu \in \{1,2,\ldots , k-3\}, 0 \ \leq \nu \ < \ 1)$ or $(\mu=k-2, 0 \ \leq \ \nu  < \frac{d-k+1}{d-k+2})$.
\ecor
\bpf We will show that the upper bound on the file size given in \ref{thm:bound1} satisfies the criterion in \eqref{eq:nonvanishing}. Let
\bea
\label{eq:eps} \delta = \left\{ \begin{array}{lc} \epsilon_1 & \mu = 0, \theta \neq 0 \\
							\max \{\epsilon_0, \epsilon_1\} & \mu \in \{1,2,\ldots , k-3 \} \\
							\epsilon_0 & \mu = k-2, \theta < \frac{d-k+1}{d-k+2}\beta \end{array} \right. 
\eea
Let $\alpha$ be related to $\beta$ as $\alpha = (d-\mu)\beta - \theta=(d-\mu)\beta - \nu\cdot\beta, \ \ \nu \in [0,1)$ by a fixed pair $(\mu,\nu)$ that falls in the range given. Then for a code with the file size $B$,
\bean
\frac{\beta}{B} & \geq & \frac{\beta}{\hat{B} - \delta }, \ \ \ \text{(using Thm.~\ref{thm:bound1})}\\
& = & \frac{\beta}{\hat{B}} \cdot \frac{1}{1-\left(\frac{\delta}{\hat{B}} \right)} \\
& = & \frac{\beta}{\hat{B}} \cdot \frac{1}{ 1 - \left(\frac{\delta}{ \beta \sum_{i=0}^{k-1} \min \{ (d-\mu)-\nu , (d-i) \} } \right)}  \\
& \geq & \frac{\beta}{\hat{B}} + \delta_0, 
\eean
for some $\delta_0 > 0$, determined by the constants $\frac{\epsilon_0}{\beta}$ and $\frac{\epsilon_1}{\beta}$. It can be seen that $\frac{\epsilon_0}{\beta}$ and $\frac{\epsilon_1}{\beta}$ are independent of $\beta, B$ and dependent only on the fixed values of $\mu,\nu,k$ and $d$. This completes the proof.
\epf

\section{Discussion On Various Known Upper Bounds On ER File Size \label{sec:oth_bounds}}

%The upper bound on the ER file size stated in Thm.~\ref{thm:bound1} was presented in an earlier conference submission \cite{SasSenKum_isit}. That was the first result that proved a non-vanishing gap between the optimal file size of an FR code and an ER code for every parameter set ${\cal P}$. It also translated to a non-vanishing gap between the ER tradeoff and FR tradeoff for any ${\cal P}$. Later other researchers came up with improved bounds on the ER file size using different techniques. 
In this section, we briefly review the results from \cite{Tia}, \cite{Tia_544}, \cite{Duursma2014}, \cite{Duursma2015},\cite{MohTan}, all of them involving upper bounds on the ER file size. While bounds provided in \cite{Duursma2014}, \cite{Duursma2015} are not explicit, those presented in \cite{Tia}, \cite{Tia_544}, \cite{MohTan} have got the form of explicit algebraic expressions.

\subsection{Review of the Bounds in \cite{Tia},\cite{Tia_544}}

In \cite{Tia}, Tian characterized the optimal ER file size for the case of $(n,k,d)=(4,3,3)$. This was the first result establishing a non-vanishing gap for ER file size in comparison with the optimal FR file size. For the case of $(n,k,d)=(4,3,3)$, there are four bounds
\bea \label{eq:433_trapezoid}
B & \leq & B_q, \ q = 0,1,2,3 ,
\eea 
that follow from considering all possible trapezoidal configurations. For a given operating point $\alpha = (d-\mu)\beta -\theta$, one of these bounds dominate over the others. By suitably modifying the information theory inequality prover software(see \cite{ITIP}, \cite{Yeu}), Tian was able to characterize a bound
\bean
3B & \leq & 4 \alpha + 6 \beta,
\eean
that is different from \eqref{eq:433_trapezoid}. Recently in \cite{Tia_544}, Tian made further progress with his computational approach to provide an upper bound on the ER file size for $(n,k,d)=(5,4,4)$. In both the case of $(4,3,3)$ and $(5,4,4)$, the bounds are achieved using the well-known class of layered codes\cite{TiaSasAggVaiKum}. These results are made part of the online collection of ``Solutions of Computed Information Theoretic Limits (SCITL)'' hosted at \cite{SCITL}. 

\subsection{Review of the Bound in \cite{Duursma2014} \label{sec:duursma1}}

In the second of two bounds presented in~\cite{Duursma2014}, Duursma considers the region $Z_q$ in a trapezoidal configuration $(Q, Z_q)$, and tiles the region using rectangular blocks corresponding to random variables $S_M^L$, with $m := |M|$, $\ell := |L|$. This approach is an extension of the tiling-with-line-segments method, introduced in \cite{SasSenKum_isit} and used in the present paper in the derivation of Thm.~\ref{thm:bound1}.  Duursma extends the upper bound given in \cite{SasSenKum_isit} to obtain a bound on $H(S_M^L)$, involving entropy expressions having a negative coefficient.  Various carefully-chosen alternative bounds on $B$ are used to cancel out these negative terms leading to the improved bound: 
\bea
\label{eq:gauss} B + \sum\limits_{(M,L) \in {\cal M}} \ell B \le B_q + \sum\limits_{(M,L) \in {\cal M}} (B_{r+m-1} + (\ell - 1)(B_{r+m-2} - \beta)),
\eea
where $m:=|M|$, $\ell:=|L|$ and $r \ge \ell$ for every choice of $(M,L)$. In \eqref{eq:gauss}, ${\cal M}$ denotes a set of possible tilings of the trapezoidal region $Z_q$ using rectangular blocks, and $B_q$ remains as defined in Section~\ref{sec:trapezium}. To obtain the best possible explicit bound, one would then proceed to minimize this expression over all possible tilings. It can easily be checked that the bound in \eqref{eq:gauss} is tighter than the one given in \eqref{eq:Bq_bound}, by a difference of at most $\beta$. 

\subsection{Review of the Bound in \cite{Duursma2015} \label{sec:duursma2}}

In \cite{Duursma2015}, Duursma augments the set of node random variables $\{W_i\}_{i=1}^{k}$ with another set of random variables $W'_{k+u}$  for $1 \leq u \leq \nu$ satisfying
\bea
\label{eq:aux_var} H(S_i^j | W'_{k+u}) \le H(S_i^j | W_{[i+1, k]} W'_{[k+1, k+u-1]}) \text{ for } 1\leq i < j \leq p,
\eea
for a given value of $p$, $0 \leq p \leq k$. The bound on file size $B$ is obtained as
\bean
(\nu + 1)B \le (\nu + 1)B_{k-p} + \sum\limits_{u = 1}^{\nu} \left(H(W'_{k+u}) - {p \choose 2}\beta \right),
\eean
where $B_{k-p}$ is as defined earlier.   This results in general, in an implicit bound as it is not clear how the random variables $\{W'_{k+u}\}_{u=1}^{\nu}$ can be constructed. However, restricting to linear codes, the author is able to construct the $\{W'_{k+u}\}$ resulting in an explicit bound for every parameter set $(n,k,d)$. This bound matches with the one proved in \cite{PraKri_isit} for the special case of $(k+1,k,k)$-linear ER codes.

\subsection{Review of the Bound in \cite{MohTan} \label{sec:mohajer}}

In this section, we give a complete description\footnote{We have simplified the proof to some extent, and therefore certain arguments differ from what is presented in \cite{MohTan}.} of the proof of the bound due to Mohajer et al. in \cite{MohTan}. We start with recalling the bound given in \eqref{eq:trapezoid2} for a trapezoidal configuration $(Q, Z_q)$,
\bea
\nonumber B & \leq & H(W_Q) + H(Z_q \mid W_Q) \\
\label{eq:soh_0}& = & H(W_Q) + H(X_q, S_R^P \mid W_Q),
\eea
where the sets $P$, $Q$, and $R$ are as defined in Sec.~\ref{sec:trapezium}. For convenience of notation, we modify the indexing of elements in sets $P$, $Q$ and $R$, without making any change in their respective sizes. Thus the sets $Q, P, R$ are defined by the same value of $q$, and hence the bound in \eqref{eq:trapezoid2} remains unaltered. With respect to the modified indexing, $Q = \{ -1, -2, \ldots, -q \}$, $P = \{ 1, 2, \ldots, p:=k-q \}$ and $R = \{ k+1, k+2, \ldots, d+1 \}$. Continuing from \eqref{eq:soh_0}, we write
\bea
\nonumber B & \leq & H(W_Q) + H(X_q, S_R^P \mid W_Q) \\
\label{eq:soh_1} &\leq & q \alpha + \underbrace{\sum\limits_{i=1}^p H\left(S_i^{\left[i-1\right]} \mid W_Q\right)}_{{\cal R}(p)} + H\left(S_R^P \mid W_Q \right).  
\eea
Instead of invoking the union bound as done in \eqref{eq:trapezoid_bound}, the entropic term ${\cal R}(p) := \sum\limits_{i=1}^p H\left(S_i^{\left[i-1\right]} | W_Q\right)$ is canceled out with the help of other expressions for $B$. In \eqref{eq:soh_2} that follows, the authors over-count conditional node entropy $H(W_i \mid W_{[i-1]})$ as $\alpha$, and later subtract out the error introduced in doing so. This leads to a different expression for $B$:
\bea
\label{eq:soh_2} B &=&  H(W_Q) + \sum\limits_{i=1}^p  H(W_i \mid W_Q) -  \sum\limits_{i=1}^p I(W_i ; W_{[i-1]} \mid W_Q) \nonumber \\
& \leq & q\alpha + p \alpha -  \sum\limits_{i=1}^p I(S_i^{[i-1]} ; S_{[i-1]}^i \mid W_Q) \nonumber \\
& = & k\alpha - \underbrace{\sum\limits_{i=1}^p H\left(S_i^{\left[i-1\right]} \mid W_Q\right)}_{{\cal R}(p)} - \underbrace{\sum\limits_{i=1}^p H\left(S_{\left[i-1\right]}^i \mid W_Q\right)}_{{\cal C}(p)} + \underbrace{\sum\limits_{i=1}^p H\left(S_{\left[i-1\right]}^i, S_i^{\left[i-1\right]} \mid W_Q\right)}_{{\cal J}(p)}.
\eea
While \eqref{eq:soh_2} allows cancellation of ${\cal R}(p)$ in \eqref{eq:soh_1}, it introduces new entropic terms ${\cal C}(p)$ and ${\cal J}(p)$. A third expression for $B$ is obtained by over-counting entropy of columns in the trapezoidal region $Z_q$ using union bound, and then subtracting out the error introduced in doing so.
\bea
\nonumber B & \leq & H(W_Q, S_{[d+1]}^P) \nonumber \\
\nonumber & \leq & q\alpha + \sum\limits_{i=1}^p H(S_{[d+1]}^i | W_Q ) - \sum\limits_{i=1}^p I\left(S_{[d+1]}^i; S_{[d+1]}^{[i-1]} \mid W_Q\right) \\
\label{eq:soh_31} & \leq & q\alpha + \sum\limits_{i=1}^p H\left(S_{[i-1]}^i \mid W_Q\right) + \sum\limits_{i=1}^p H\left(S_{[i+1 \ d+1]}^i \mid W_Q\right) - \sum\limits_{i=1}^p I\left(S_{[d+1]}^i; S_{[d+1]}^{[i-1]} \mid W_Q\right).
\eea
The following straightforward lemma is useful in producing a lower bound for $I\left(S_{[d+1]}^i; S_{[d+1]}^{[i-1]} \mid W_Q\right)$. 
\blem \label{lem:useful} Let $X,Y,Z,U$ be random variables such that $Z \ = \  f_1(X,U) \ = \ f_2(Y,U)$ for some deterministic functions $f_1$, $f_2$. Then
\bean
I(X:Y \mid U) & \geq & H(Z \mid U).
\eean
\elem
%\bpf
%The proof follows from applying chain rule in two different ways.
%\bean
%I(XZ: YZ \mid U) & = & I(X: Y \mid U) + I(X: Z \mid Y,U) + I(Z: YZ \mid X,U) \\
%& = & I(X:Y \mid U). \\
%I(XZ : YZ \mid U) & = & I(Z:Z \mid U) + I(Z:Y \mid Z, U) + I(X:Y Z \mid Z, U) \\
%& \geq & I(Z:Z \mid U) \ = \ H(Z \mid U).
%\eean
%\epf
By invoking Lem.~\ref{lem:useful} along with identifying $Z = \{S_{[i-1]}^i, S_i^{\left[i-1\right]} \}$, $X = S_{[d+1]}^i$, $Y = S_{[d+1]}^{[i-1]}$ and $U = W_Q$, it follows that 
\bea
\label{eq:soh_32} I\left(S_{[d+1]}^i; S_{[d+1]}^{[i-1]} | W_Q\right) & \geq & H\left(S_{[i-1]}^i, S_i^{\left[i-1\right]} | W_Q\right),
\eea
and substituting \eqref{eq:soh_32} back in \eqref{eq:soh_31}, the authors obtain the bound
\bea
\label{eq:soh_3} B & \leq & q \alpha + \underbrace{\sum\limits_{i=1}^p H\left(S_{\left[i-1\right]}^i \mid W_Q\right)}_{{\cal C}(p)} + \sum\limits_{i=1}^p H\left(S_{\left[i+1 \ d+1\right]}^i \mid W_Q\right) - \underbrace{\sum\limits_{i=1}^p H\left(S_{[i-1]}^i, S_i^{\left[i-1\right]} \mid W_Q\right)}_{{\cal J}(p)}.
\eea
Summation of \eqref{eq:soh_1} \eqref{eq:soh_2} and \eqref{eq:soh_3} eliminates ${\cal R}(p)$, ${\cal C}(p)$ and ${\cal J}(p)$, and results in the bound:
\bea
\label{eq:soh_final0}3B &\le& (3k-2p)\alpha + \sum\limits_{i=1}^p H\left(S_{\left[i+1 \ d+1\right]}^i \mid W_Q\right) + H\left(S_R^P \mid W_Q \right).
\eea
By applying union bound, it follows that
\bea
\label{eq:soh_final}B &\le & \min_{0 \leq p \leq  k} \frac{(3k-2p)\alpha + \frac{p(2(d-k)+p+1)\beta}{2} + (d-k+1)\min\{\alpha, p\beta\} }{3}.
\eea
To our knowledge, the bound in \eqref{eq:soh_final} due to Mohajer et al. remains the best known upper bound on ER file size in the region away from the MSR point. 

\section{An Improved Upper Bound on ER File Size \label{sec:bound2}}

In this section, we first propose an improvement over the bound in \cite{MohTan}, that is described in Sec.~\ref{sec:mohajer}. The authors of \cite{MohTan} apply union bound on the last two terms in \eqref{eq:soh_final0} to obtain the final bound. But it is possible to avoid the union bound for the term $H\left(S_R^P \mid W_Q \right)$ when $d >> k$. 

\begin{figure}[h!]
\begin{center}
 \includegraphics[width=2.8in]{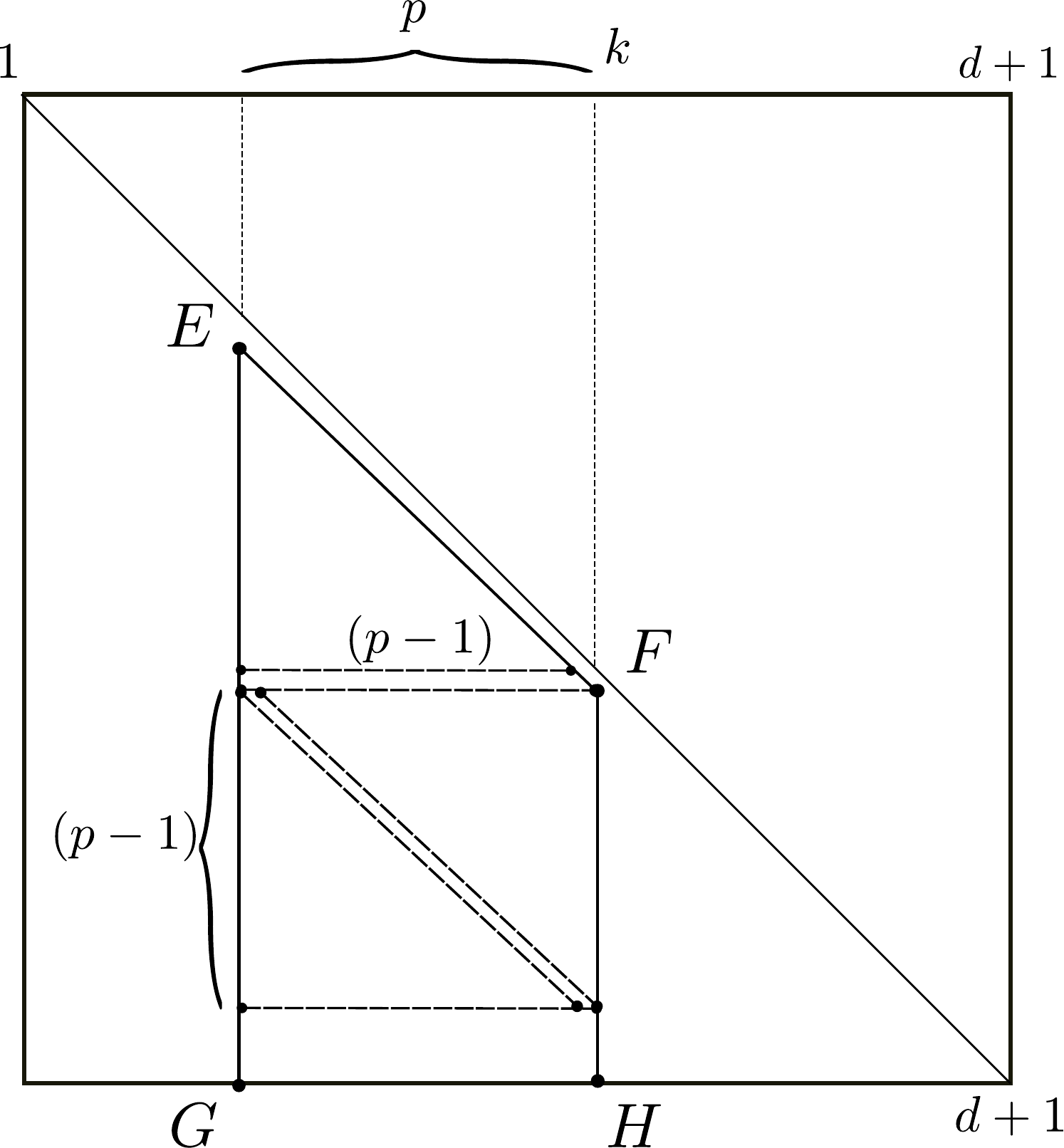}
\caption{The splitting up of the region corresponding to $S_R^P$. In this example, $a = 1$, $b > 0$.} \label{fig:soh_imp}
\end{center}
\end{figure}

The Fig.~\ref{fig:soh_imp} illustrates the region $S_R^P$ as it is viewed on the repair matrix. The rectangular region $S_R^P$, denoted by $\Gamma$, is of width $p$ and height $(d-k+1)$. Let us write
\bean
d-k+1 = a(p-1) + b, \ 0 \leq b < (p-1).
\eean
Then $\Gamma$ can be split into $(a+1)$ sub-rectangles $\Gamma_1, \Gamma_2, \ldots, \Gamma_{a+1}$ of equal width $p$, and $\Gamma_i, 1 \leq i \leq a$ have the same height $(p-1)$. The last sub-rectangle $\Gamma_{a+1}$ is of height $b$, and it vanishes in the case $b=0$. Each rectangle $\Gamma_i, 1 \leq i \leq a$ is further split into two isosceles right triangles $\Gamma_{i1}$,  $\Gamma_{i2}$ of base $(p-1)$ as illustrated in Fig.~\ref{fig:soh_imp}. By symmetry, we can write
\bea
\nonumber H(S_P^R | W_Q) &\leq & a H\left( \Gamma_1 | W_Q \right) + H(\Gamma_{a+1}|W_Q)\\
\nonumber & \leq & 2a H\left( \Gamma_{11} | W_Q \right) + H(\Gamma_{a+1}|W_Q)\\
\label{eq:imp_soh_1}& \leq & 2a \sum_{i=1}^{p} H\left( S_i^{[i-1]} | W_Q \right) + b\min\{\alpha, p\beta\} .
\eea
We improve upon the the bound in \eqref{eq:soh_1} by substituting \eqref{eq:imp_soh_1}, and obtain that
\bea
\label{eq:soh_rect} B &\le& q \alpha + (1+2a)\sum\limits_{i=1}^p H\left(S_i^{\left[i-1\right]} | W_Q\right) + b\min\{\alpha, p\beta\}.
\eea
This modification only affects the coefficient of the term ${\cal R}(p)$. The cancellation of the term ${\cal R}(p)$ is possible by appropriately scaling the bounds in \eqref{eq:soh_2} and \eqref{eq:soh_3}. This results in an improved bound whenever $a \geq 1$, and is stated in the theorem below. We refer to this bound as the {\em improved Mohajer-Tandon bound}.
\bthm \label{thm:bound2} The ER file size $B$ of regenerating code with full-parameter set ${\cal P}_f = \{(n,k,d),(\alpha,\beta)\}$ is  bounded by
\bea
\label{eq:soh_improved} B & \leq & \min_{0 \leq p \leq k} \frac{\alpha(2(k-p)(1+a)+k(1+2a)) + b\min\{\alpha, p\beta\} + \frac{(1+2a)p(2(d-k)+p+1)\beta}{2} }{3+4a},
\eea
where $d-k+1 = a(p-1) + b$ and $0 \leq b < (p-1)$.
\ethm
We remark that the improved Mohajer-Tandon bound relies upon the same techniques introduced by Mohajer et al. of coming up with various expressions for $B$ allowing one to cancel out entropic terms that are otherwise difficult to estimate. Our incremental contribution is limited to identifying the symmetry in certain entropic terms as seen in the pictorial depiction on a repair matrix, and leveraging upon this symmetry to avoid certain union bounds. When $d > k$, the bound in Thm.~\ref{thm:bound2} leads to an outer bound on normalized ER tradeoff, that lies above the one due to \eqref{eq:soh_final}. A principal result of the paper stated in Thm.~\ref{thm:bound3} follows by combining both the Thm.~\ref{thm:bound1} and the Thm.~\ref{thm:bound2}.

\section{A Dual-Code-Based Approach To Bounding the ER File Size for Linear Codes\label{sec:linear_app}}

In this section, we investigate the maximum possible ER file size under the restricted setting of linear regenerating codes. Let ${\cal C}_{\text{lin}}$ denote a linear ER code with full-parameter set ${\cal P}_f = \{(n, k, d), (\alpha, \beta)\}$. We will continue to use $B$ to denote the file size. By linear, we mean that (a) the encoding mapping that converts the $B$ message symbols to $n\alpha$ coded symbols is linear, (b) the mapping that converts the node data into repair data that is transmitted during the repair of a failed node is linear and furthermore, (c) the mappings that are involved during data collection from a set of $k$ nodes and regeneration of a failed node using repair data from a set of $d$ nodes are linear. A linear regenerating code can be viewed as a linear block-code with length $n\alpha$ over $\mathbb{F}$ such that every set of $\alpha$ symbols (taken in order without loss of generality) are bunched together to correspond to a node. 

\subsection{The Parity-Check Matrix And Its Properties\label{sec:pc}}

Since ${\cal C}_{\text{lin}}$ is a linear code, we can associate a generator matrix to the code. Let $G$ of size $(B \times n\alpha)$ denote a  generator matrix of ${\cal C}_{\text{lin}}$. Without loss of generality, we assume that the first $\alpha$ columns of $G$ generate the contents of the first node, the second $\alpha$ columns of $G$ generate the contents of the second node, and so on. The first $\alpha$ columns taken together will be referred to as the first thick column of $G$. Similarly, the second thick column consists of columns from $\alpha+1$ to $2\alpha$, and so on. Overall, we will have $n$ thick columns in $G$. Let $H$ denote a parity-check matrix having size $(n\alpha - B) \times n\alpha$. The row-space of $H$ is the dual code of ${\cal C}_{\text{lin}}$. The definition of thick columns directly carries over to $H$. For any set $S \subseteq [n]$, we write $H|_S$ to denote the restriction of $H$ to the thick columns indexed by the set $S$. From definitions, we have that
\bea  \label{eq:B_rank}
B & = & \textsl{rank}(G) \ = \ n\alpha - \textsl{rank}(H) .
\eea
By \eqref{eq:B_rank}, it is sufficient to obtain a lower bound on $\textsl{rank}(H)$ to bound $B$ from above. This is precisely the approach taken here. 

In the following two lemmas, we will translate the properties of data collection and exact-repair as properties of the parity-check matrix $H$. We remark here that these observations are already made in \cite{Duursma2014}. 
\begin{lem}[Data Collection] \label{lem:H_datacollection} Let $H$ be a parity-check matrix of an ER linear regenerating code. Then $\textsl{rank}\left({H|_S}\right) = (n-k)\alpha$, for any $S \subseteq [n]$ such that $|S| = n-k$. 
\end{lem}
\begin{proof}
This is a re-statement of Part $(1)$ of Proposition $2.1$ of \cite{Duursma2014}, and is equivalent to the data collection property.
\end{proof}
\begin{lem}[Exact Repair] \label{lem:H_repair}
Assume that $d= n-1$. Then the row space of $H$ of an ER linear regenerating code contains a collection of
$n\alpha$ vectors that can be arranged as the rows of an $(n\alpha \times n\alpha)$ matrix $H_{repair}$, which can be written in the block-matrix form:
\begin{eqnarray} \label{eq:Hrepair}
H_{repair} & = & \left[ \begin{array}{c|c|c|c} A_{1,1} & A_{1,2} & & A_{1,n} \\
\hline &&& \\
A_{2,1} & A_{2,2}  & \hdots & A_{2, n} \\
\hline &&& \\
& & \ \ \vdots \ \ & \\
\hline &&& \\
A_{n,1} & A_{n,2} &  & A_{n,n} \end{array}
\right],
\end{eqnarray}
where $A_{i,i}$ is defined to be the identity matrix $I_{\alpha}$ of size $\alpha$ and $A_{i,j}$ denotes an $\alpha \times \alpha$ matrix such that $\text{rank}\left(A_{i, j}\right) \leq \beta, 1 \leq i, j \leq n, i \neq j$.
\end{lem}
\begin{proof} The first $\alpha$ rows of the form
\bean
\left[ \begin{array}{c|c|c|c}  I_{\alpha} & A_{1,2} & \cdots & A_{1,n} \end{array}
\right]
\eean
can be obtained by the parity-check equations that are necessitated by the exact-repair requirement of the first node. In a similar manner, there must be parity-check equations that must allow repair of every node. These parity-check equations can be arranged to obtain the matrix $H_{repair}$. The requirements on the ranks of the sub-matrices $A_{ij}$ follow from the definition of regenerating codes, and the fact that $d=n-1$. In fact, the proof is indicated in Part $(2)$ of Proposition $2.1$ of \cite{Duursma2014}.
\end{proof}	

For the case of $d=k=n-1$, the matrix $H_{\text{repair}}$ as given in Lem.~\ref{lem:H_repair} satisfies the condition given in Lem.~\ref{lem:H_datacollection}, and therefore $H_{repair}$ by itself defines an $(n, k=n-1, d=n-1)(\alpha, \beta)$ regenerating code. Since $\textsl{rank}(H) \geq \textsl{rank}(H_{repair})$, and our interest lies in regenerating codes having maximal file size, we will assume that $H = H_{repair}$ while deriving a lower bound on $\textsl{rank}(H)$ for the case of $d=k=n-1$.

\subsection{A Proof Of FR Bound For ER Linear Codes Using Dual Code \label{sec:dimakis_via_dual}}

In this section, we will present a simple proof of the FR bound \eqref{eq:cut_set_bd} for ER linear regenerating codes. Our proof of Theorem \ref{thm:new_bound_k_eq_d} will be built up on the proof of \eqref{eq:cut_set_bd} that is presented here.

As earlier, let $\mathcal{C}_{\text{lin}}$ denote an $(n, k, d = n-1) (\alpha, \beta)$ linear regenerating code, and let the matrix $H$ generate the dual code of $\mathcal{C}$. The key idea of the proof is to obtain a lower bound on the column rank of the matrix $H$. We  use the notation  $\rho(.)$ to denote the rank of a matrix. Let us define the quantities $\delta_j, 1 \leq j \leq n$ as follows:

\begin{eqnarray}
\delta_1 & = & \rho(H|_{[1]}), \label{eq:proof_diamkis_aa}\\
\delta_j & = & \rho(H|_{[j]}) - \rho(H|_{[j-1]}),\ 2 \leq j \leq n. \label{eq:proof_diamkis_0}
\end{eqnarray}
Next, we make the following claims:
\begin{eqnarray} 
\label{eq:proof_dimakis_1} \delta_{j} &=& \rho(A_{j,j}) \\
 &=& \alpha, \ \ 1 \leq j \leq n-k \\
\label{eq:proof_diamkis_2}
\delta_{j} & \geq & (\alpha - (j-1)\beta)^+, \ n-k + 1 \leq j \leq n.
\end{eqnarray}
Here we have set $a^+ := \text{max}(a, 0)$. The first claim \eqref{eq:proof_dimakis_1} follows from the fact that any $n-k$ thick  columns of $H$ has rank given by $(n-k)\alpha$ as required by Lem.~\ref{lem:H_datacollection}. To show the second claim \eqref{eq:proof_diamkis_2}, one needs to first focus on the $j^{\text{th}}$ thick row of $H_{repair}$. By $j^{\text{th}}$ thick row, we mean the set of rows starting from $(j-1)\alpha + 1$ and reaching up to $j\alpha$ of $H_{repair}$. Next observe that 
\bea
\delta_j & \geq & \left(\rho(A_{j,j}) - \sum_{\ell=1}^{j-1}\rho(A_{j,\ell})\right)^+  \label{eq:dimakis_proof_3a} \\
\nonumber & = & \left(\rho(I_{\alpha}) - \sum_{\ell=1}^{j-1}\rho(A_{j,\ell})\right)^+  \\
& \geq & \left(\alpha - (j-1)\beta \right)^+,  \ n-k + 1 \leq j \leq n, \label{eq:dimakis_proof_3}
\eea
where \eqref{eq:dimakis_proof_3} holds true since $\rho(A_{i,j}) \leq \beta$ by Lem.~\ref{lem:H_repair}. Thus we have shown \eqref{eq:proof_diamkis_2}. Next, invoking \eqref{eq:proof_dimakis_1} and \eqref{eq:dimakis_proof_3}, we bound the column-rank of $H$ from below as:
\begin{eqnarray}
\text{rank}(H) & = & \sum_{j=1}^{n}\delta_j \label{eq:proof_dimakis_rkH_delta_eq} \\
& \geq & (n-k)\alpha + \sum_{j=n-k+1}^{n}\left(\alpha - (j-1)\beta \right)^+. \label{eq:dimakis_proof_4}
\end{eqnarray}
An illustration of arriving at \eqref{eq:dimakis_proof_3a} and \eqref{eq:dimakis_proof_4} is given in Fig. \ref{fig:rankH_computation}.
\begin{figure}[h]
	\centering
	\includegraphics[width=4in]{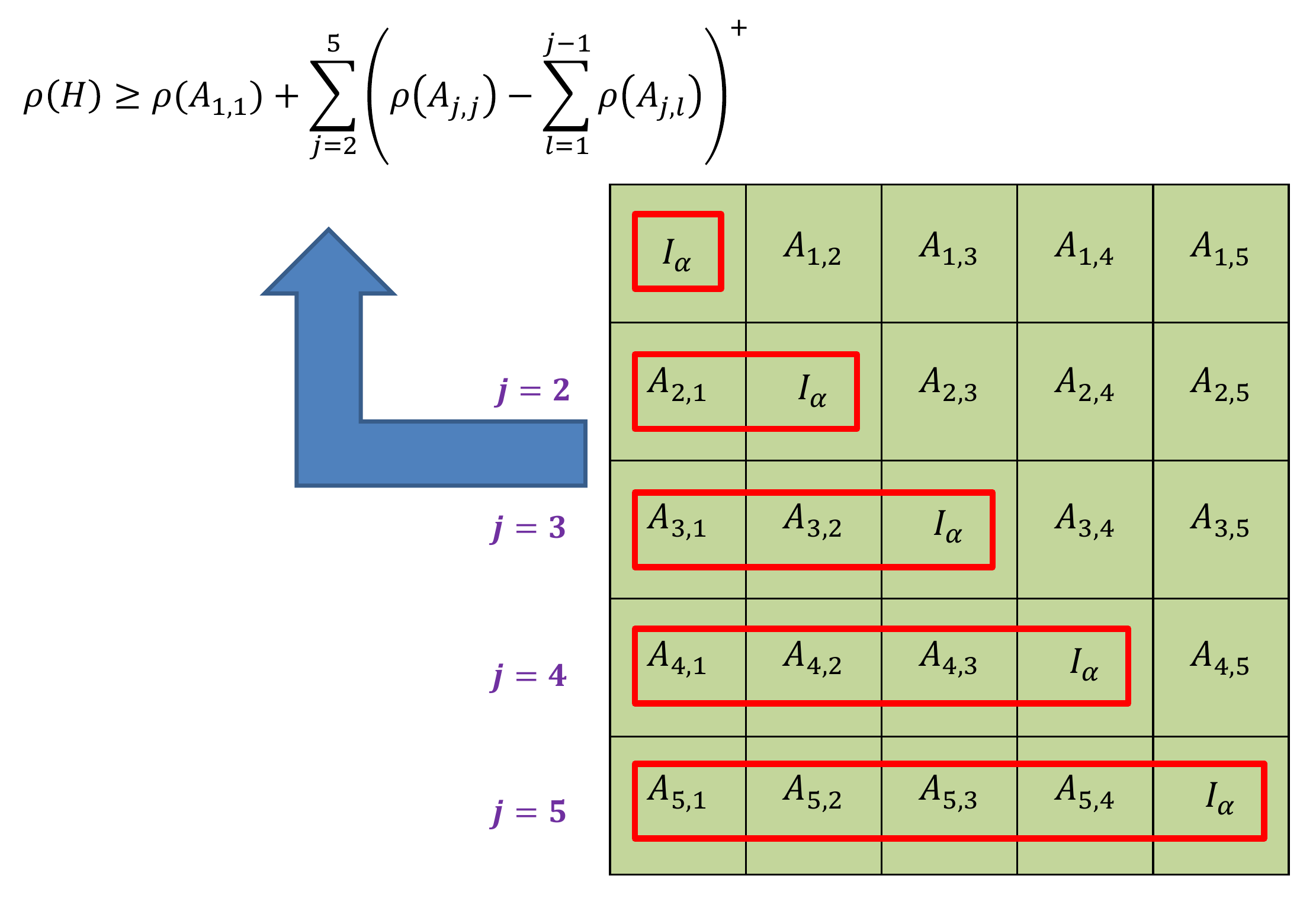}
	\caption{A lower bound on $\rho(H)$, for the case of $(n = 5, k = 4, d = 4)$. Each term indexed by $j$ in the summation correspond to a lower bound on the incremental rank $\delta_j$. This bound is obtained by looking at the sub-matrices in $j^{th}$ thick row.}
	\label{fig:rankH_computation}
\end{figure}
Consequently, it follows that 
\begin{eqnarray}
B &  =  & n\alpha - \rho(H) \\
& \leq & n\alpha - (n-k)\alpha  - \sum_{j=n-k+1}^{n}\left(\alpha - (j-1)\beta \right)^+ \\
& = & \sum_{j=0}^{k-1}\min(\alpha, (n-1-j)\beta) \label{eq:dimakis_proof_6}.
\end{eqnarray}
For $d< n-1$, the proof follows by first puncturing the code on any $(n-d-1)$ nodes to form a $(n'=d+1,k,d)$ ER linear regenerating code and then invoking the above analysis on the resultant new code. The way we express incremental ranks $\{\delta_j\}$ in \eqref{eq:proof_dimakis_1} and \eqref{eq:dimakis_proof_3a} will turn out to be useful in deriving a strong upper bound on the file size of linear ER codes in Sections \ref{sec:544} and \ref{sec:main_proof}.

\section{An Upper Bound On The File Size Of Linear ER Codes For The Case $(n=5,k=4,d=4)$\label{sec:544}}

In this section, we obtain a new upper bound on the file size of a linear ER code for parameters $(n=5,k=4,d=4)$. Taken along with the achievability using layered codes (see Sec.~\ref{sec:ach_lin}), we characterize the tradeoff for this case. As mentioned earlier, our technique is to lower bound the rank of the parity-check matrix $H$, leading to an upper bound on the file size by \eqref{eq:B_rank}. The lower bound on $\rho(H)$ that we derive here is in general tighter than what is obtained in \eqref{eq:dimakis_proof_4}. The principal result of this section is stated in Thm.~\ref{thm:544} below. Most of the ideas that are developed in the proof of Thm.~\ref{thm:544} will later be used in the next section to prove a general result for the case of $(n,k=n-1,d=n-1)$.

\begin{thm} \label{thm:544} Consider an ER linear regenerating code $\mathcal{C}_{\text{lin}}$ with full-parameter set $\{(n=5, k=4, d=4), (\alpha, \beta)\}$. Let $H$ denote a parity-check matrix of $\mathcal{C}_{\text{lin}}$. Then
	\begin{eqnarray} \label{eq:bound_rank_H_544}
	\rho(H) & \geq & \left \{ \begin{array}{c} \left \lceil \frac{10(\alpha - \beta)}{3} \right \rceil, \ 2\beta \leq \alpha \leq 4\beta \\  \\
	\left \lceil \frac{15\alpha - 10\beta}{6} \right \rceil, \ \frac{4}{3}\beta \leq \alpha \leq 2\beta \\  \\
	2\alpha - \beta, \ \beta \leq \alpha \leq \frac{4}{3}\beta \end{array} \right. .
	\end{eqnarray}
\end{thm}

Note that $\alpha = \beta$ and $\alpha = 4\beta$ correspond to the MSR and MBR points respectively for the case of $(n=5,k=4,d=4)$. Next, we observe that for a fixed $\beta$, the bound given in \eqref{eq:bound_rank_H_544} corresponds to a piecewise linear curve with $\alpha$ on the $X$-axis and $\rho(H)$ on the $Y$-axis. Non-linear ceiling operation $\lceil.\rceil$ is used in \eqref{eq:bound_rank_H_544} to enforce integrality requirements on $\rho(H)$.  However, it may be removed considering that $\rho(H)$ always takes integer values. We can view \eqref{eq:bound_rank_H_544} as a combination of the following three inequalities without paying attention to the limited range of $\alpha$:
\bea
\rho(H) & \geq &  \frac{10(\alpha - \beta)}{3}  \label{eq:proof_544_line1}\\
\rho(H) & \geq & \frac{15\alpha - 10\beta}{6} \label{eq:proof_544_line2} \\
\rho(H) & \geq & 2\alpha - \beta \label{eq:proof_544_en}.
\eea
Here \eqref{eq:proof_544_en} follows from \eqref{eq:dimakis_proof_4} since $\alpha \geq \beta$ and $\left(\alpha - (j-1)\beta \right)^+\ge 0$ for $3 \leq j \leq 5$. Therefore, we need to prove only the remaining two inequalities \eqref{eq:proof_544_line1} and \eqref{eq:proof_544_line2} to complete the proof of Thm.~\ref{thm:544}. We proceed to prove them by obtaining two lower bounds to the incremental thick-column-rank of $H$ that are stronger than what is given in \eqref{eq:dimakis_proof_3}. To make this point clear upfront, a comparison of the bounds in \eqref{eq:dimakis_proof_4} and \eqref{eq:bound_rank_H_544} is shown in Fig.~\ref{fig:544_rankH_comparison}. 

%(\eqref{eq:dimakis_proof_4} also exhibits this behavior).
\begin{figure}[h]
	\centering
	\includegraphics[width=4.5in]{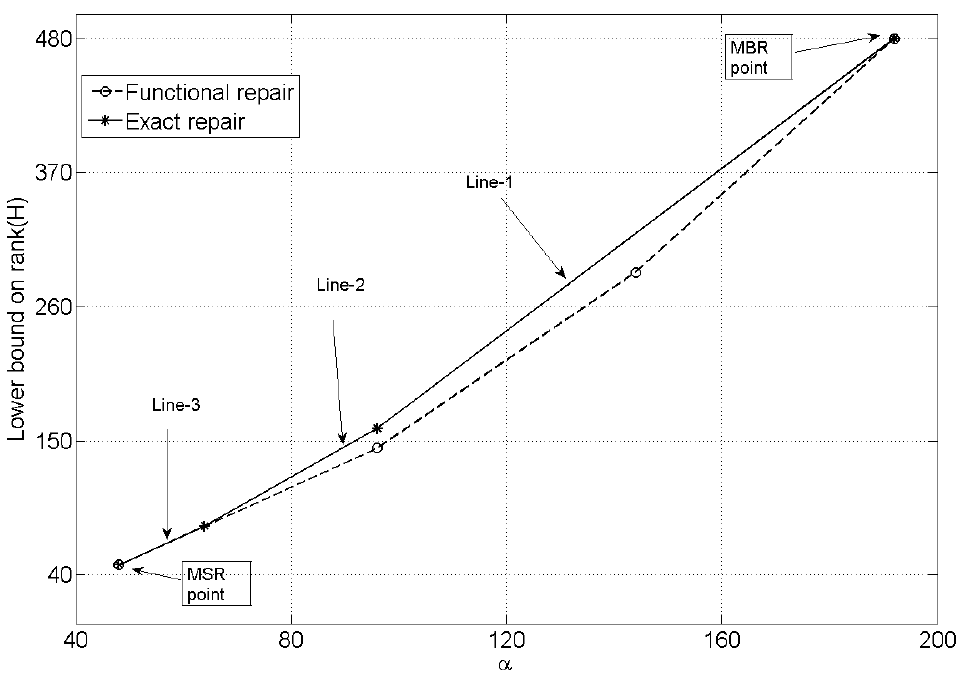}
	\caption{Comparison of the lower bounds on $\rho(H)$ as function of $\alpha$, for the case of $(n = 5, k = 4, d = 4)$ with $\beta = 48$. The dashed and the solid lines correspond to the cases of functional and exact repairs, respectively. See  \eqref{eq:dimakis_proof_4} and \eqref{eq:bound_rank_H_544} for the corresponding equations. Here, lines $1$, $2$ and $3$ are given by \eqref{eq:proof_544_line1}, \eqref{eq:proof_544_line2} and \eqref{eq:proof_544_en} respectively.} 
	\label{fig:544_rankH_comparison}
\end{figure}

\subsection{Proof of Theorem~\ref{thm:544} \label{sec:proof_thm_433}}

We begin with setting up some notation. For any matrix $B$ over $\mathbb{F}$, we denote by $\mathcal{S}(B)$ the column space of $B$. Note that $\rho(B)$ is the same as the dimension of the vector space $\mathcal{S}(B)$. Next, we define $H^{(5)} = H_{repair}$, where $H_{repair}$ is as defined in \eqref{eq:Hrepair}. Let the matrix $H^{(5)}_j$ denote the $j^{\text{th}}$ thick column of $H^{(5)}, 1 \leq j \leq 5$, i.e., $H^{(5)} = [H^{(5)}_1 \ H^{(5)}_2  H^{(5)}_3 \ H^{(5)}_4 \ H^{(5)}_5 ]$. Next, we define matrices $H^{(4)}_j,  2 \leq j \leq 5$ such that the columns of $H^{(4)}_j$ form a basis for the vector space $\mathcal{S}\left(H^{(5)}_j\right) \cap \mathcal{S}\left(H^{(5)}|_{[j-1]}\right)$. Next, we define $H^{(4)}$ as
\bean
H^{(4)} = [H^{(4)}_2  H^{(4)}_3 \ H^{(4)}_4 \ H^{(4)}_5 ]. 
\eean
For convenience of notation, we have used $H^{(4)}_2$ to denote the first thick column of $H^{(4)}$. Similarly, $H^{(3)}$ is obtained from $H^{(4)}$, where columns of $H^{(3)}_j$ form a basis for $\mathcal{S}\left(H^{(4)}_j\right) \cap \mathcal{S}\left(H^{(4)}|_{\{2,\ldots,j-1\}}\right)$:
\bean
H^{(3)} = [H^{(3)}_3 \ H^{(3)}_4 \ H^{(3)}_5 ] .
\eean
Let $A_{i,j}^{(\ell)}$ denote the  $i^{th}$ thick row of $H^{(\ell)}_j$. An illustration of the block-matrix representations of $H^{(5)}, H^{(4)}$ and $H^{(3)}$ is given in Fig.~\ref{fig:matrices_544}. 

The key idea in the proof lies on the observation that $\rho(H^{(5)}) \geq \rho(H^{(4)}) \geq \rho(H^{(3)})$. We will show that \eqref{eq:proof_544_line1} and \eqref{eq:proof_544_line2} are necessary conditions respectively for $\rho(H^{(5)}) \geq \rho(H^{(4)})$ and  $\rho(H^{(5)}) \geq \rho(H^{(4)}) \geq \rho(H^{(3)})$ to be true. The following remark underlines an important property of $\rho\left(A^{(\ell)}_{j,j}\right)$ preserved in the construction of $H^{(\ell)}_j$.
\begin{note}\label{rem:full_column_rank_diag_sub_matrix} The sub-matrices $A^{(\ell)}_{j,j}, 3 \leq \ell \leq 5, \ 	5-\ell+1 \leq j \leq 5$ have full column rank, and $\rho\left(A^{(\ell)}_{j,j}\right) = \rho\left(H^{(\ell)}_j\right)$.
\end{note}
\begin{figure}[h]
	\centering
	\includegraphics[width=6in]{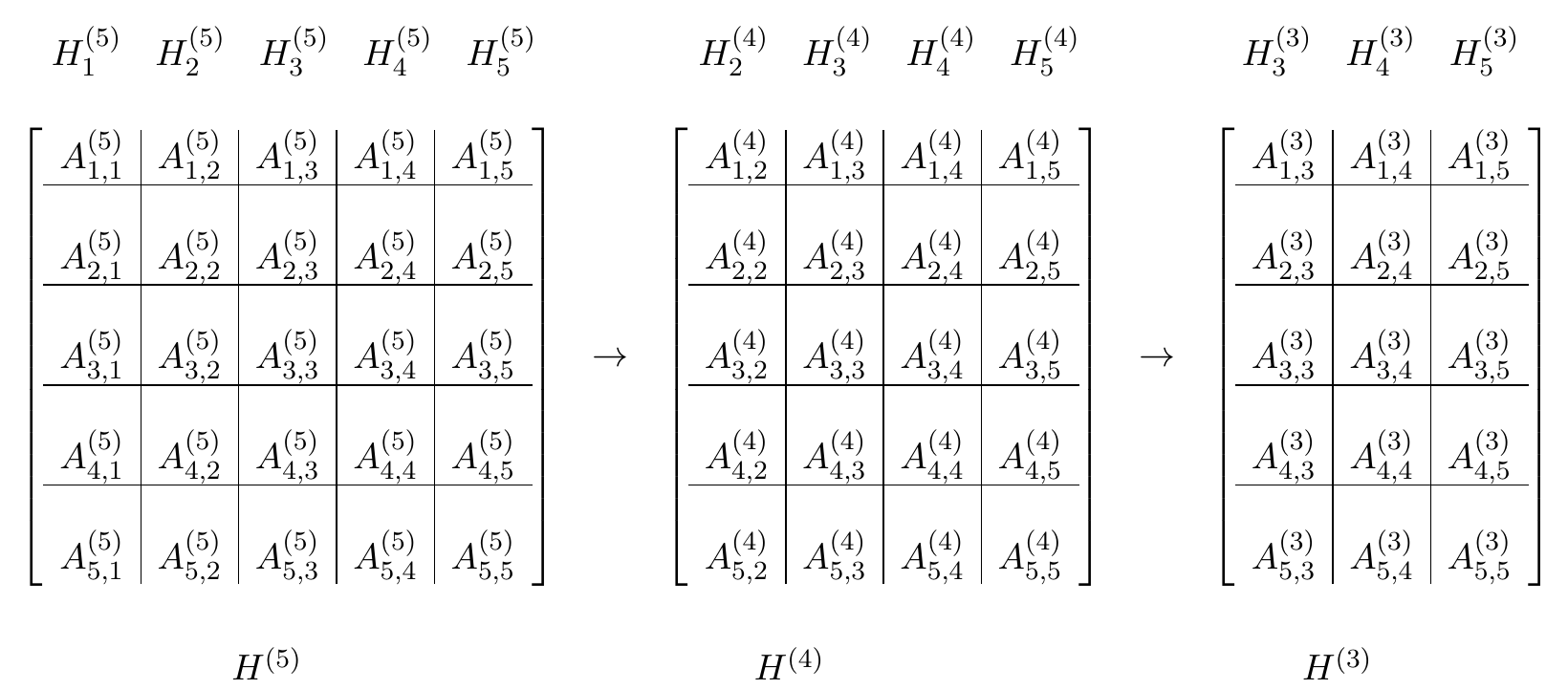}
	\caption{The matrices $H^{(5)}, H^{(4)}$ and $H^{(3)}$, and the associated block submatrix representations for the case  $n=5$. The matrix $H^{(5)} = H_{repair}$, $H^{(4)}$ is defined based on $H^{(5)}$, and $H^{(3)}$ is defined based on $H^{(4)}$.}
	\label{fig:matrices_544}
\end{figure}

\subsection{Proof of \eqref{eq:proof_544_line1}\label{sec:proof_bound1_544}} 

We will be using the rank comparison $\rho(H^{(5)})  \geq  \rho(H^{(4)})$ to prove \eqref{eq:proof_544_line1}. It follows from \eqref{eq:proof_dimakis_1}, \eqref{eq:dimakis_proof_3a} and \eqref{eq:proof_dimakis_rkH_delta_eq} that
\begin{eqnarray}
\rho\left( H^{(5)} \right) & \geq & \rho\left(A^{(5)}_{1,1}\right) \ + \  \sum_{j=2}^{5} \left\{ \left(\rho\left(A_{j,j}^{(5)}\right) - \sum_{\ell=1}^{j-1}\rho\left(A_{j,\ell}^{(5)}\right)\right)^+  \right\} \label{eq:544_H5_rank_ineq}.
\end{eqnarray}
We introduce slack variables $\{\alpha_j, 2 \leq j \leq 5\}$ that take non-negative integer values to convert \eqref{eq:dimakis_proof_3a} into equalities i.e.,

\begin{eqnarray}\label{eq:delta_equality_544}
\delta_j & = & \left(\rho\left(A_{j,j}^{(5)}\right) - \sum_{\ell=1}^{j-1}\rho\left(A_{j,\ell}^{(5)}\right)\right)^+  +   \alpha_j, \ 2 \leq j \leq 5.
\end{eqnarray}
Hence, using \eqref{eq:proof_dimakis_rkH_delta_eq} we have:
\bea
\rho\left( H^{(5)} \right) & = & \rho\left(A^{(5)}_{1,1}\right) \ + \  \sum_{j=2}^{5} \left\{ \left(\rho\left(A_{j,j}^{(5)}\right) - \sum_{\ell=1}^{j-1}\rho\left(A_{j,\ell}^{(5)}\right)\right)^+  +   \alpha_j \right\}. \label{eq:544_H5_rank_eq}
\eea

$\rho\left( H^{(4)} \right)$ can be bounded from below quite similar to \eqref{eq:544_H5_rank_ineq} (see Remark \ref{rem:full_column_rank_diag_sub_matrix} also) to obtain
\begin{eqnarray}
\rho\left( H^{(4)} \right)  & \geq & \rho\left(A^{(4)}_{2,2}\right)  +   \sum_{j=3}^{5} \left\{ \left(\rho\left(A_{j,j}^{(4)}\right) - \sum_{\ell=2}^{j-1}\rho\left(A_{j,\ell}^{(4)}\right)\right)^+  \right\} \label{eq:boundH4_544}.
\end{eqnarray}
Our aim at first is to find a lower bound for $\sum_{j=2}^{5}\alpha_j$. The analysis in Sec.~\ref{sec:dimakis_via_dual} in fact works with the trivial lower bound $\sum_{j=2}^{5}\alpha_j\geq 0$. But here, we substitute \eqref{eq:544_H5_rank_eq} and \eqref{eq:boundH4_544} in 
\bean
\rho(H^{(5)}) & \geq & \rho(H^{(4)})
\eean
to obtain a much tighter lower bound for $\sum_{j=2}^{5}\alpha_j$. Using this tighter bound in \eqref{eq:544_H5_rank_eq}, we will obtain a lower bound for $\rho\left( H^{(5)} \right)$ in terms of $\{ \rho\left(A^{(5)}_{i,j}\right), \rho\left(A^{(4)}_{i,j}\right)\}$. We know that the terms $\{ \rho\left(A^{(5)}_{i,j} \right)\}$ can be expressed in terms of $\alpha$ and $\beta$. In the following Lem.~\ref{lem:544_intersections}, we show how  $\left\{\rho\left(A^{(4)}_{i,j} \right)\right\}$ can be expressed in terms of $\left\{\rho\left(A^{(5)}_{i,j} \right)\right\}$. Finally, all the terms involve $\{ \rho\left(A^{(5)}_{i,j} \right)\}$, and this will lead to the proof of \eqref{eq:proof_544_line1}. 

\begin{lem} \label{lem:544_intersections} The following statements hold:
\begin{enumerate}[a)] \item \begin{eqnarray}
		\rho\left(A^{(4)}_{j,j} \right) & = & \rho\left(A^{(5)}_{j,j} \right) \ - \ \left\{ \left(\rho\left(A_{j,j}^{(5)}\right) - \sum_{\ell=1}^{j-1}\rho\left(A_{j,\ell}^{(5)}\right)\right)^+  +   \alpha_j\right\}, \ 2 \leq j \leq 5.\label{eq:544_key_lemma1}
		\end{eqnarray}
\item \begin{eqnarray}
		\sum_{\ell=2}^{j-1} \rho\left(A^{(4)}_{j,\ell} \right) & \leq & \sum_{\ell=1}^{j-1}\rho\left(A^{(5)}_{j,\ell} \right) \ - \ \rho\left(A^{(4)}_{j,j} \right), \ 3 \leq j \leq 5.\label{eq:544_key_lemma2}
\end{eqnarray}
\end{enumerate}
\end{lem}

\bpf The proof is relegated to Appendix~\ref{app:pf_lem_544}.
\epf

By making use of Lem.~\ref{lem:544_intersections}, we first obtain a lower bound on $\sum_{j=2}^{5}\alpha_j$, and subsequently a lower bound on $\rho\left( H^{(5)} \right)$ all in terms of  $\{ \rho\left(A^{(5)}_{i,j} \right)\}$:
\begin{eqnarray}
\sum_{j=2}^{5}\alpha_j & \geq& \frac{1}{3}\left\{ -\rho\left(A^{(5)}_{1,1}\right) + \rho\left(A^{(5)}_{2,2}\right)  +
2\sum_{j=3}^{5}\rho\left(A^{(5)}_{j,j}\right) -\right. \nonumber\\
&&\left. \left[2\left( \rho\left(A^{(5)}_{2,2}\right) - \rho\left(A^{(5)}_{2,1}\right)\right)^+ + 3\sum_{j=3}^{5}\left( \rho\left(A^{(5)}_{j,j}\right) - \sum_{\ell=1}^{j-1}\rho\left(A^{(5)}_{j,\ell}\right)\right)^+ + \sum_{j=3}^{5} \sum_{\ell=1}^{j-1}\rho\left(A^{(5)}_{j,\ell}\right)\right]
\right\}\label{eq:544_alpha_bound1} \\
\label{eq:544_bound1} \rho\left( H^{(5)} \right) & \geq  & \frac{1}{3}\left\{ 2\sum_{j=1}^{5}\rho\left(A^{(5)}_{j,j}\right) - \sum_{j=2}^{5} \sum_{\ell=1}^{j-1}\rho\left(A^{(5)}_{j,\ell}\right)\right\}.
\end{eqnarray}
Finally, we apply $\rho\left(A^{(5)}_{j,j}\right) = \alpha, 1 \leq j \leq 5$ and $\rho\left(A^{(5)}_{i,j}\right) \leq \beta, 1 \leq i, j \leq 5, \ i \neq j$ to complete the proof of \eqref{eq:proof_544_line1}.

\subsection{Proof of \eqref{eq:proof_544_line2}\label{sec:proof_bound2_544}}

While proving \eqref{eq:proof_544_line1}, we leveraged upon the inequality
\bean
\rho(H^{(5)}) & \geq & \rho(H^{(4)}) .
\eean
Here, we will make use of the chain 
\bean
\rho(H^{(5)}) & \geq & \rho(H^{(4)})\  \geq \ \rho(H^{(3)}),
\eean
to prove \eqref{eq:proof_544_line2}. First, we consider $\rho(H^{(4)}) \geq \rho(H^{(3)})$ and obtain a lower bound on $\rho(H^{(4)})$. This is carried out precisely the same way as how we obtained the lower bound \eqref{eq:544_bound1} on $\rho(H^{(5)})$. The only change required will be to adapt Lem.~\ref{lem:544_intersections} to express $\{A_{i, j}^{(4)}\}$ in terms of $\{A_{i, j}^{(3)}\}$. Thus we obtain that 
\bea \label{eq:544_bound2a}
\rho\left( H^{(4)} \right) & \geq & \frac{1}{3}\left\{ 2\sum_{j=2}^{5}\rho\left(A^{(4)}_{j,j}\right) - \sum_{j=3}^{5} \sum_{\ell=2}^{j-1}\rho\left(A^{(4)}_{j,\ell}\right)\right\}.
\eea
Observe that \eqref{eq:544_bound2a} is same as \eqref{eq:544_bound1} except for that $\{A_{i, j}^{(5)}\}$ are replaced with $\{A_{i, j}^{(4)}\}$. The limits of the summation are also modified accordingly. 

We next consider the inequality $\rho(H^{(5)}) \geq \rho(H^{(4)})$ where $\rho(H^{(4)})$ is lower bounded as in \eqref{eq:544_bound2a} and  $\rho(H^{(5)})$ is equated using \eqref{eq:544_H5_rank_eq}. It follows that 
\begin{eqnarray}
 \rho\left(A^{(5)}_{1,1}\right)  +  \sum_{j=2}^{5} \left\{ \left(\rho\left(A_{j,j}^{(5)}\right) - 
\sum_{\ell=1}^{j-1}\rho\left(A_{j,\ell}^{(5)}\right)\right)^+ + \alpha_j \right\}  \geq  \frac{1}{3}\left\{ 2\sum_{j=2}^{5}\rho\left(A^{(4)}_{j,j}\right) - \sum_{j=3}^{5} \sum_{\ell=2}^{j-1}\rho\left(A^{(4)}_{j,\ell}\right)\right\}. \label{eq:long_544_2}
\end{eqnarray}
After invoking Lem.~\ref{lem:544_intersections}, we obtain the lower bound: 
\begin{eqnarray}
\sum_{j=2}^{5}\alpha_j & \geq & \frac{1}{6}\left\{ -3\rho\left(A^{(5)}_{1,1}\right) + 3\sum_{j=2}^{5}\rho\left(A^{(5)}_{j,j}\right) - \right.\nonumber\\
&&\left. \left[6\sum_{j=2}^{5}\left( \rho\left(A^{(5)}_{j,j}\right) - \sum_{\ell=1}^{j-1}\rho\left(A^{(5)}_{j,\ell}\right)\right)^+ + \sum_{j=2}^{5} \sum_{\ell=1}^{j-1}\rho\left(A^{(5)}_{j,\ell}\right)\right]
\right\}\label{eq:544_alpha_bound2} .
\end{eqnarray}
Substituting \eqref{eq:544_alpha_bound2} back in \eqref{eq:544_H5_rank_eq}, we obtain the following lower bound on $\rho(H^{(5)})$: 
\bea \label{eq:544_bound2}
\rho\left( H^{(5)} \right) & \geq & \frac{1}{6}\left\{ 3\sum_{j=1}^{5}\rho\left(A^{(5)}_{j,j}\right) - \sum_{j=2}^{5} \sum_{\ell=1}^{j-1}\rho\left(A^{(5)}_{j,\ell}\right)\right\}.
\eea
Finally, we apply $\rho\left(A^{(5)}_{j,j}\right) = \alpha, 1 \leq j \leq 5$ and $\rho\left(A^{(5)}_{i,j}\right) \leq \beta, 1 \leq i, j \leq 5, \ i \neq j$ on \eqref{eq:544_bound2} to complete the proof of \eqref{eq:proof_544_line2}.

\section{An Upper Bound On The File Size Of Linear ER Codes for General $(n,k=n-1,d=n-1)$ \label{sec:main_proof}}

In this section, we generalize the result proved for $(n=5,k=4,d=4)$ in  Sec.~\ref{sec:544} to $(n,k=n-1,d=n-1)$. We will only provide a sketch of the proofs, as the techniques remain the same as those presented in Sec.~\ref{sec:544} (see \cite{PraKri_arxiv} for details). Again, the upper bound on the file size is a direct corollary of a lower bound on rank($H$) and the bound is achievable using layered codes (see Sec.~\ref{sec:ach_lin}). In the following theorem, a lower bound on rank($H$) is established.

\begin{thm} \label{thm:rankH_new_bound_k_eq_d} Consider an ER linear regenerating code $\mathcal{C}_{\text{lin}}$ with full-parameter set $\{(n, k=n-1, d=n-1), (\alpha, \beta)\}$ with $n \geq 4$. Let $H$ denote a parity-check matrix of $\mathcal{C}_{\text{lin}}$. Then
\begin{eqnarray} \label{eq:bound_rank_H_gen}
\text{rank}(H) & \geq & \left \{ \begin{array}{cl} \left \lceil \frac{2rn\alpha - n(n-1)\beta}{r^2+r}\right \rceil, & \frac{d\beta}{r} \leq \alpha \leq \frac{d\beta}{r-1}, \ \ 2 \leq r \leq n - 2 \\ 
2\alpha - \beta, & \frac{d\beta}{n-1} \leq \alpha \leq \frac{d\beta}{n-2} \end{array} \right. .
\end{eqnarray}
\end{thm}

The corresponding theorem Thm.\ref{thm:544} for $(n=5,k=4,d=4)$ established that $\textsl{rank}(H)$ is lower bounded by a piecewise linear curve determined by $3$ inequalities. Here, we show such a behavior exists in general i.e., rank($H$) can be lower bounded by a piecewise linear curve determined by $(n-2)$ inequalities. The last inequality
\bea
\text{rank}(H) & \geq & 2\alpha - \beta \label{eq:proof_k_eq_d_eq_nminus1_eq1},
\eea
is already established by \eqref{eq:dimakis_proof_4}, since $\alpha\ge \beta$ and $\left(\alpha - (j-1)\beta \right)^+\ge 0$ for $3\leq j\leq n$. Therefore to complete the proof, it remains to prove the following $(n-3)$ bounds on $\textsl{rank}(H)$, ignoring the range of $\alpha$:
\begin{equation} 
\text{rank}(H)   \geq   \frac{2rn\alpha - n(n-1)\beta}{r^2+r}, \label{eq:general_bound_to_prove}
\end{equation}
parameterized by $2 \leq r \leq n-2$. We will set up some notations, and introduce a key lemma that are essential in describing a sketch of the proof.

\subsection{Notations and a Key Lemma \label{sec:proof_not}}

\subsubsection{The Matrices $\{H^{(t)}, \ 3 \leq t \leq n\}$} For any matrix $M$ over $\mathbb{F}$, we carry over the notation $\rho(M)$, $\mathcal{S}(M)$ from Sec.~\ref{sec:proof_thm_433}. Quite similar to the definition of $H^{(5)}$ in Sec.~\ref{sec:proof_thm_433}, we define $H^{(n)} = H_{repair}$, where $H_{repair}$ is as defined by Lem.~\ref{lem:H_repair}. We denote by $H^{(n)}_j$ the $j^{\text{th}}$ thick column of $H^{(n)}, 1 \leq j \leq n$, i.e.,
\bean
H^{(n)} = [H^{(n)}_1 \ H^{(n)}_2  \ \ldots \  H^{(n)}_n].
\eean
Next, we define the matrices $H^{(t)}, 3 \leq t \leq n-1$ in an iterative manner as follows: 

\vspace{0.1in}
\begin{enumerate}[Step 1.]
\item Let $ t=n-1$.
\item  Define the matrices $H^{(t)}_j,  n-t+1 \leq j \leq n$, such that the columns of $H^{(t)}_j$ form a basis for the vector space $\mathcal{S}\left(H^{(t+1)}_j\right) \cap \mathcal{S}\left(H^{(t+1)}|_{\{n-t, n-t+1, \ldots, j-1\}}\right)$.
\item  Define the matrix $H^{(t)}$ as
\begin{eqnarray} \label{eq:proofgen_Htdef}
H^{(t)} = [H^{(t)}_{n-t+1}  \ H^{(t)}_{n-t+2} \ \ldots \ H^{(t)}_n].
\end{eqnarray}
\item  If $t \geq 4$, decrement $t$ by $1$ and go back to Step $2$. 
\end{enumerate}
\vspace{0.1in}
Clearly, the ranks of the matrices $H^{(t)}, 3 \leq t \leq n$ are ordered as
\begin{eqnarray} \label{eq:proofgen_rankorderHs}
\rho(H^{(t)})  & \geq & \rho(H^{(t-1)}), \ 4 \leq t \leq n.
\end{eqnarray}
We use the notation $H_j^{(t)}, n-t+1 \leq j \leq n$ to refer to the $j^{\text{th}}$ thick column of the matrix $H^{(t)}$. While every thick column of $H^{(n)}$ has exactly $\alpha$ thin columns, thick columns of $H^{(t)}, 3 \leq t \leq n-1$ need not have the same number of thin columns. We point out for clarity that the thick columns of the matrix $H^{(t)}$ are indexed using $\{n-t+1, \ldots, n\}$. We have avoided $\{1, \ldots, t\}$ for the convenience of notation.

\subsubsection{Block Matrix Representation of the Matrix $H^{(t)}$ \label{sec:bsm}}

Since $H^{(n)} = H_{repair}$, it has a block matrix representation as given in \eqref{eq:Hrepair}. We write in short-hand
\begin{eqnarray}
H^{(n)} & = & \left( A^{(n)}_{i,j} , 1 \leq i, j \leq n \ \right ),
\end{eqnarray}
where $A^{(n)}_{i,i} = I_{\alpha}, 1 \leq i \leq n$. We introduce block matrix representations for $H^{(t)}, 3 \leq t \leq n-1$ as
\begin{eqnarray}
H^{(t)} & = & \left( A^{(t)}_{i,j} , 1 \leq i \leq n, n-t+1 \leq j \leq n \ \right ),
\end{eqnarray}
where  $A^{(t)}_{i,j}$ is an $\alpha \times \rho(H^{(t)}_j)$ matrix over $\mathbb{F}$ such that 
\begin{eqnarray}
\mathcal{S}\left(A^{(t)}_{i,j} \right) & \subseteq & \mathcal{S}\left(A^{(t+1)}_{i,j} \right) \ \bigcap \ \sum_{\ell=n-t}^{j-1}\mathcal{S}\left(A^{(t+1)}_{i,\ell}\right).  \label{eq:proofgen_1}
\end{eqnarray}
Note that \eqref{eq:proofgen_1} is a direct consequence of our definition of the matrix $H^{(t)}$. Having set up the notation, we introduce the key lemma that establishes the relations among the ranks of the sub-matrices of $\{H^{(t)}\}$. The lemma is similar in spirit to Lem.~\ref{lem:544_intersections}, and its proof is omitted here.

\begin{lem} \label{lem:nkk_intersections}
\begin{enumerate}[a)]
\item For any $t, j$ such that $3 \leq t \leq n$ and $n-t+1 \leq j \leq n$, we have
\begin{eqnarray} 
\rho\left(H^{(t)}_j\right)  & =  & \rho\left( A^{(t)}_{j,j}\right). \label{eq:proofgen_2}
\end{eqnarray}
\item For any $t, j$ such that $3 \leq t \leq n-1$ and $n-t+1 \leq j \leq n$, we have
\begin{eqnarray}
\rho\left(A^{(t)}_{j,j} \right) & = & \rho\left(A^{(t+1)}_{j,j} \right) \ - \ \left\{ \rho\left(H^{(t+1)}|_{\{n-t, \ldots, j\}}\right) - \rho\left(H^{(t+1)}|_{\{n-t, \ldots, j-1\}}\right) \right\}. \label{eq:proofgen_3}
\end{eqnarray}
\item For any $t, j$ such that $3 \leq t \leq n-1$ and $n-t+2 \leq j \leq n$, we have 
\begin{equation}
\rho\left(A^{(t)}_{j,j} \right) + \sum_{\ell=n-t+1}^{j-1} \rho\left(A^{(t)}_{j,\ell} \right)  \leq  \sum_{\ell=n-t}^{j-1}\rho\left(A^{(t+1)}_{j,\ell} \right) . \label{eq:proofgen_4}
\end{equation}
\end{enumerate}
\end{lem}

\subsection{On The Proof of \eqref{eq:general_bound_to_prove} \label{sec:eq_general_bound_proof}}

The bounds in \eqref{eq:general_bound_to_prove} is obtained as a necessary condition for satisfying the chain of inequalities given by 
\begin{eqnarray} \label{eq:chain_r}
\rho(H^{(n)})  & \geq & \rho(H^{(n-1)}) \ \geq \ \cdots \ \geq \ \rho(H^{(n-r+1)}).
\end{eqnarray}
In the analysis of \eqref{eq:chain_r}, we consider in the first step, the inequality $\rho(H^{(n-r+2)}) \geq \rho(H^{(n-r+1)})$ and obtain a lower bound on $\rho(H^{(n-r+2)})$. In the second step, we move on to the inequality $\rho(H^{(n-r+3)}) \geq \rho(H^{(n-r+2)})$ and obtain a lower bound on $\rho(H^{(n-r+3)})$. In the second step, we would make use of a lower bound on $\rho(H^{(n-r+2)})$ that was derived in the first step. This procedure is continued iteratively until we arrive at lower bound for $\rho(H^{(n)})$. The following theorem is a key intermediate step in this process.

\begin{thm} \label{thm:nkk_rankviainduction}
For any $s$ such that $1 \leq s \leq n-3$, and any $t$ such that $3+s \leq t \leq n$, the rank of the matrix $H^{(t)}$ is lower bounded by
\begin{eqnarray}
\rho\left( H^{(t)} \right) & \geq & \frac{2}{(s+1)(s+2)}\left\{ (s+1)\sum_{j=n-t+1}^{n}\rho\left(A^{(t)}_{j,j}\right) -  \sum_{j=n-t+2}^{n} \sum_{\ell=n-t+1}^{j-1}\rho\left(A^{(t)}_{j,\ell}\right)\right\}. \label{eq:boundgen_rankH}
\end{eqnarray}
\end{thm}

\begin{proof} The proof is by induction on $s$, and see \cite{PraKri_arxiv} for details.
\end{proof}	

One can identify Thm.~\ref{thm:nkk_rankviainduction} in the context of $(n=5,k=4,d=4)$. The bounds would then be associated with $(s=1, t=5)$, $(s=1, t = 4)$ and $(s = 2, t = 5)$, and are precisely those given in \eqref{eq:544_bound1}, \eqref{eq:544_bound2a} and \eqref{eq:544_bound2} respectively. To complete the proof of  \eqref{eq:general_bound_to_prove}, we evaluate the bound in \eqref{eq:boundgen_rankH} for the $(n-3)$ pairs given by $(s, t=n), 1 \leq s \leq n-3$. By substituting the constraints $\rho\left(A^{(n)}_{j,j}\right) = \alpha, 1 \leq j \leq n$ and $\rho\left(A^{(n)}_{i,j}\right) \leq \beta, 1 \leq i, j \leq n, i \neq j$, we finally obtain that
\begin{eqnarray}
\rho\left( H^{(n)} \right)  \geq \frac{2}{(s+1)(s+2)}\left\{ (s+1)\sum_{j=1}^{n}\alpha - \sum_{j=2}^{n} (j-1)\beta\right\}  =  \frac{2(s+1)n\alpha - n(n-1)\beta}{(s+1)(s+2)}, \ 1 \leq s \leq n-3. \label{eq:proofgenz}
\end{eqnarray}
By choosing $r=s+1$, \eqref{eq:general_bound_to_prove} follows from \eqref{eq:proofgenz}. This completes the proof of \eqref{eq:general_bound_to_prove}, and consequently that of the Thm.~\ref{thm:rankH_new_bound_k_eq_d}.

\section{On The Achievability Of The Outer Bounds On Normalized ER Tradeoff \label{sec:achievability}}

The outer bounds presented in the present paper matches with the performance of existing code constructions in certain cases, and we present two such results here. 

\subsection{Characterization of Normalized ER tradeoff for the Case $k=3,d=n-1$\label{sec:ach_nonlin}}

In the case of $k=3$ and $d=n-1$, the repair-matrix bound is achieved by a construction that appeared in \cite{SenSasKum_itw}. We will given an example of  the repair-matrix bound below:

{\em Example:} $(n=6,k=3,d=5):$ Using \eqref{eq:eps}, the bound on the ER file size $B$ can computed as  
\bea
\label{eq:635} B & \leq & \frac{10\alpha}{7} + \frac{34\beta}{7}, \  \ 5\beta \geq \alpha > \frac{13\beta}{4}.
\eea
Based on the bound in \eqref{eq:635}, an outer bound on the normalized ER tradeoff is drawn in Fig.~\ref{fig:plot1a}. It is required to have a single code construction ${\cal C}_{\text{int}}$ for the normalized operating point $(\bar{\alpha}_0,\bar{\beta}_0)= (\frac{13}{38},\frac{2}{19})$ to achieve the entire normalized ER tradeoff, as the remaining points can be achieved by space-sharing of the MSR code ${\cal C}_{\text{MSR}}$, the MBR code ${\cal C}_{\text{MBR}}$ and the code ${\cal C}_{\text{int}}$. The construction of ${\cal C}_{\text{int}}$ was provided in \cite{SenSasKum_itw}, and thus establishing that the repair-matrix bound is tight in this case.

\subsection{Characterization of Normalized ER Tradeoff for the Case $(n,k=n-1,d=n-1)$ under the Linear Setting\label{sec:ach_lin}}

In the case of linear codes, the bound presented in Theorem \ref{thm:new_bound_k_eq_d} is achieved by canonical layered codes that was introduced in \cite{TiaSasAggVaiKum}. When specialized to the case of $d = n-1$, the layered codes achieve points described by
\begin{equation} \label{eq:ach_Biren}
\left(\bar{\alpha} , \bar{\beta} \right)  =  \left(\frac{r}{n(r-1)}, \frac{r}{n(n-1)}\right), \ \ 2 \leq r \leq n-1
\end{equation}
on the $(\bar{\alpha} , \bar{\beta} )$-plane. If one substitutes $r = 2$ in \eqref{eq:ach_Biren}, it corresponds to the MBR point, and the achievable points move closer to the MSR point as $r$ increases. It is also proved that the point corresponding to $r = n-1$ lies on the FR tradeoff in the near-MSR region. An achievable region on the $(\bar{\alpha} , \bar{\beta} )$-plane is obtained by space-sharing codes for values of $r$, $2 \leq r \leq n-1$ along with an MSR-code. We can write the equation of the line segment obtained by connecting two points $\left(\frac{r}{n(r-1)}, \frac{r}{n(n-1)}\right)$ and  $\left(\frac{(r+1)}{n((r+1)-1)}, \frac{(r+1)}{n(n-1)}\right)$, $ 2 \leq r \leq n-2$ as
\begin{eqnarray}
r(r-1)n\bar{\alpha} + n(n-1)\bar{\beta} & = & r^2+r,
\end{eqnarray}
and that of the line segment connecting the MSR point and the point corresponding to $r=n-1$ as
\bean
(n-2)\bar{\alpha} + \bar{\beta} & = &  1.
\eean 
This matches with the equations of line segments as given in Theorem \ref{thm:new_bound_k_eq_d}. The normalized linear tradeoff for $(n=6,k=d=5)$ is given in Fig.~\ref{fig:plot1b}.

\bibliographystyle{IEEEtran}
\bibliography{icot}

\appendices

\section{Proof Of Theorem~\ref{thm:bound1}}

Two different estimates on the joint entropy of certain repair data, expressed as functions of $\{\omega_i\}_{i=0}^{k-1}$, are used to derive a lower bound on $\epsilon = \hat{B} - B$. The repair data considered differ based on the value of $\mu$.

\vspace{0.1in}

{\em Case 1:} $\mu \in \{1,2,\ldots, k-2\}$

We set $r_0=\left\lfloor \frac{k-\mu}{\mu+1} \right\rfloor$. We will have two sub-cases for $r_0 \geq 1$ and $r_0=0$. 

\vspace{0.1in}

{\em Case 1(a):} $r_0 \geq 1$

%\begin{figure}[h!]
%\begin{center}
%\includegraphics[width=2.9in]{proof-1.pdf}
%\caption{The sub-trapezoid region $EFGH$ marked in red in the repair matrix indicates the collection of random variables, ${\cal S}$, considered in {\em Case 1(a)}. Here $r_0=\left\lfloor \frac{k-\mu}{\mu+1} \right\rfloor$. } \label{fig:case1a}
%\end{center}
%\end{figure}

\begin{figure}[h!]
\begin{center}
  \subfigure[The sub-trapezoid region considered in Case 1(a)]{\label{fig:case1a}\includegraphics[width=2.7in]{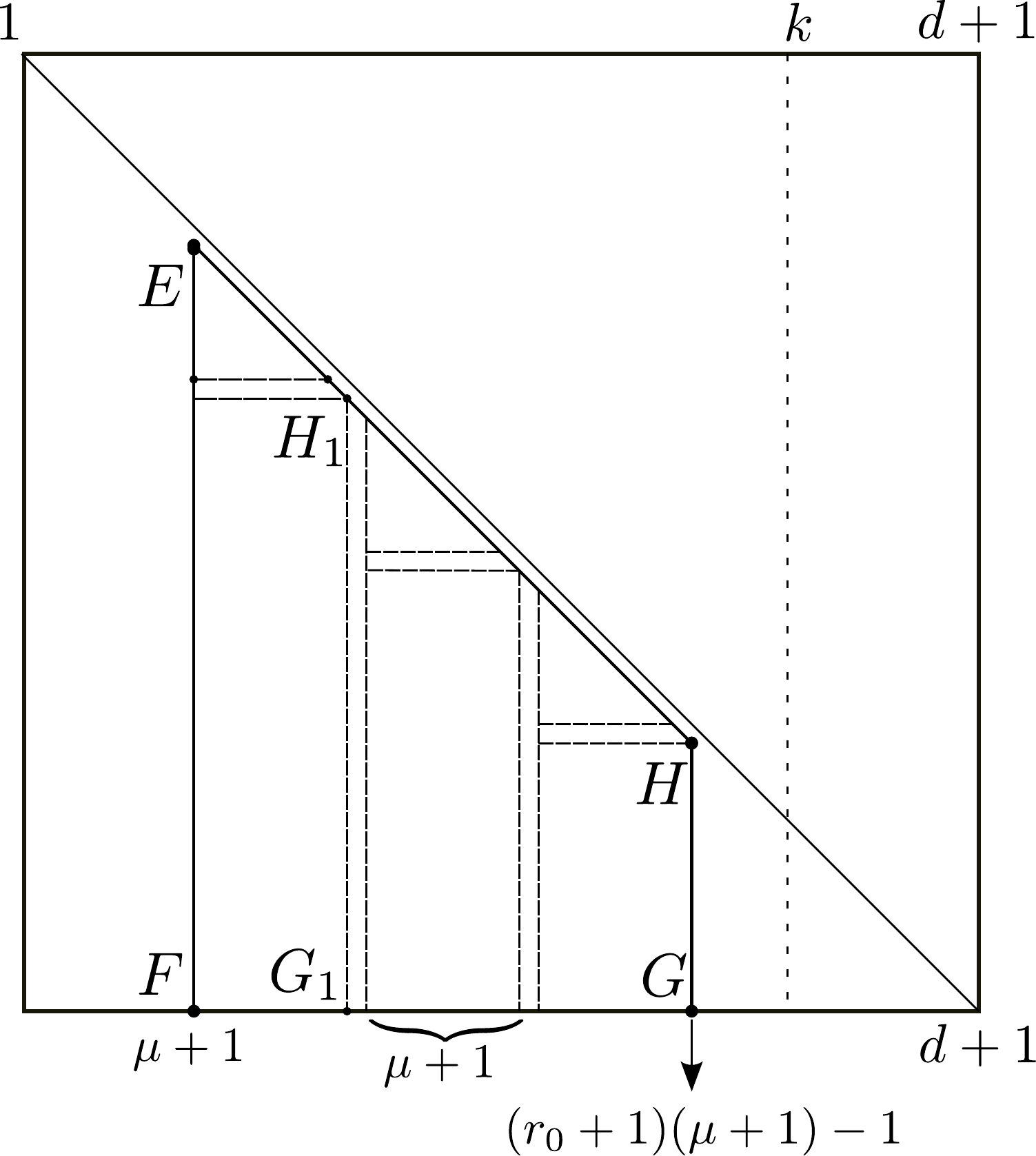}}
  \hspace{0.2in}
  \subfigure[The trapezoid region considered in Case 1(b)]{\raisebox{4mm}{\label{fig:case1b}\includegraphics[width=2.7in]{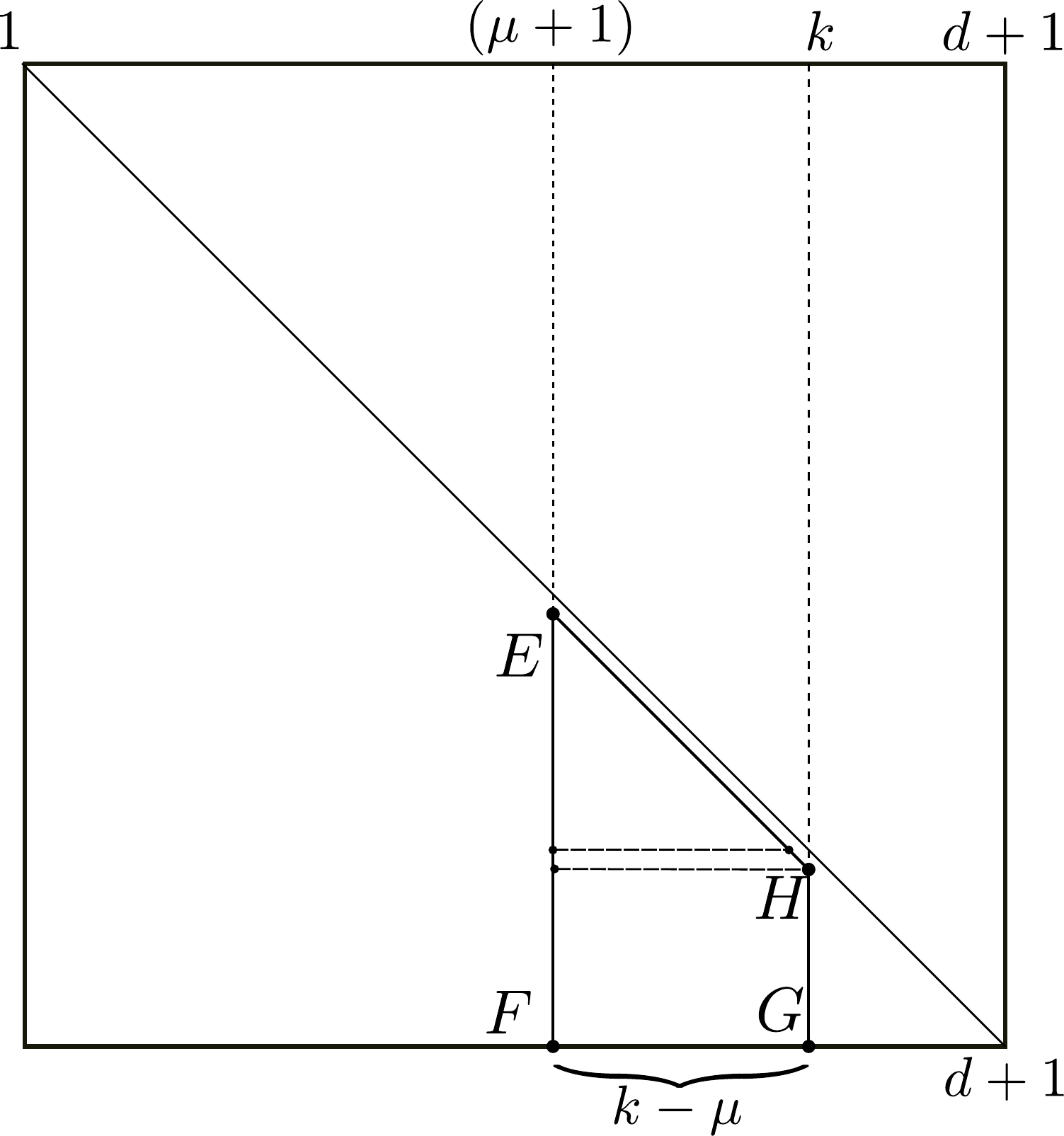}}}
\caption{The illustration of the trapezoid regions considered for Case 1. }
\end{center}
\label{fig:case1}
\end{figure}

We consider the sub-trapezoid $Z_{q,t}$ with parameter $q = \mu$ and $t = r_0(\mu+1)$. Pictorially, it is marked as a trapezium $EFGH$ in the repair matrix shown in Fig.~\ref{fig:case1a}. The set of nodes $T = \{\mu+1, \mu+2, \ldots , (r_0+1)(\mu+1)-1\}$ that are repaired by $Z_{q,t}$ is split into $r_0$ groups of $(\mu+1)$ nodes in order, and the corresponding subsets of $Z_{q,t}$ are denoted by ${\cal E}_i, i = 1, 2, \ldots, r_0$. Pictorially, ${\cal E}_1$ is associated with the trapezium $EFG_1H_1$ in Fig.~\ref{fig:case1a}. Similarly every ${\cal E}_i$ is associated with a smaller trapezium contained within $EFGH$. The set ${\cal E}_i$ can again be viewed as the union of two subsets ${\cal V}_i$ and ${\cal T}_i$, respectively associated with the largest rectangle within the trapezium, and the remaining triangular region. These sets are formally defined as
\bean
{\cal E}_i & = & \{ S_x^y \mid S_x^y \in Z_{q,t}, (\mu+1)i \leq y \leq (\mu+1)(i+1)-1 \}, \ \ i = 1,2,\ldots , r_0 \\
{\cal V}_i & = & \{ S_x^y \mid S_x^y \in {\cal E}_i, (\mu+1)(i+1) \leq x \leq d+1 \}, \ \ i = 1,2,\ldots , r_0 \\
{\cal T}_i & = & \{ S_x^y \mid S_x^y \in {\cal E}_i, (\mu+1)i+1 \leq x \leq (\mu+1)(i+1)-1 \}, \ \ i = 1,2,\ldots , r_0 .
\eean
Note that ${\cal E}_i = {\cal V}_i \cup {\cal T}_i$. Next, we bound the joint entropy $H(Z_{q,t})$ as
\bea
\nonumber H(Z_{q,t}) &  \leq & \sum_{i=1}^{r_0} H({\cal V}_i) + \sum_{i=1}^{r_0} H({\cal T}_i) \\
& \leq & \sum_{i=1}^{r_0} \left(d-(i+1)(\mu+1)+2\right) \cdot  [\beta + \mu\theta + \mu\omega_{\mu} + ((\omega_{\mu}+\omega_{\mu+1}) ] + \sum_{i=1}^{r_0} \frac{(\mu+1)\mu\beta}{2} \label{eq:upper1a}.
\eea
In the second inequality, we use \eqref{eq:rb1} of Cor.~\ref{cor:rowbound} to obtain the upper bound on $H({\cal V}_i)$. On the other hand, using Lem.~\ref{lem:colsum}, we also have,
\bea
\nonumber H(Z_{q,t}) \ \geq \ H(Z_{q,t}\mid W_Q) &  \geq &  \sum_{i=\mu}^{(r_0+1)(\mu+1)-2} \min \{\alpha , (d-i)\beta \} - \sum_{i=\mu}^{(r_0+1)(\mu+1)-2} \omega_i \\
& = & \left[ \sum_{i=\mu}^{(r_0+1)(\mu+1)-2} (d-i)\beta \right] - \theta - \sum_{i=\mu}^{(r_0+1)(\mu+1)-2} \omega_i \label{eq:lower1a}.
\eea
Matching the bounds in \eqref{eq:upper1a} and \eqref{eq:lower1a} and using the identity \eqref{eq:basic1}, we obtain that
\bea
\epsilon & \geq & \frac{\left(d-\frac{(\mu+1)(r_0+3)}{2}+2 \right)r_0\mu(\beta - \theta) \ - \ \theta}{\left(d-\frac{(\mu+1)(r_0+3)}{2}+2\right)r_0(\mu+1) \ + \ 1} \label{eq:eps1a}.
\eea

\vspace{0.1in}

{\em Case 1(b):} $r_0=0$

%\begin{figure}[h!]
%\begin{center}
% \includegraphics[width=2.8in]{proof-2.pdf}
%\caption{The trapezoid $EFGH$ marked in red in the repair matrix indicates the collection of random variables, ${\cal S}$, considered in {\em Case 1(b)}. } \label{fig:case1b}
%\end{center}
%\end{figure}

The collection $Z_{q}$ of repair data considered in this case corresponds to the trapezoid configuration $Z_q$ with $q=\mu$. The set $Z_{q}$ is written as $Z_{q} = {\cal V} \cup {\cal T}$, where
\bean
{\cal V} & = & \{ S_x^y \mid S_x^y \in Z_{q}, k+1 \leq x \leq d+1 \}, \\
{\cal T} & = & \{ S_x^y \mid S_x^y \in Z_{q}, \mu+2 \leq x \leq k \} 
\eean
Pictorially, $Z_{q}$ is represented by the trapezium $EFGH$ in Fig.~\ref{fig:case1b}. Quite similar to the {\em Case 1(a)}, we invoke Cor.~\ref{cor:rowbound} to bound $H(Z_{q})$ as 
\bea
\nonumber H(Z_{q}) \ \geq \ H(Z_{q}\mid W_Q) &  \leq & H({\cal V}) + H({\cal T}) \\
\nonumber & \leq & (d-k+1) \cdot [\beta + (k-\mu-1)\theta + (k-\mu-1)\omega_{\mu} + (\omega_{\mu}+\omega_{\mu +1}) ] + \\
&& \ \ \ \ \ \ \ \ \ \  \ \frac{(k-\mu-1)(k-\mu)\beta}{2} . \label{eq:upper1b}
\eea
On the other hand, using Lem.~\ref{lem:colsum},
\bea
\nonumber H(Z_{q}) &  \geq &  \sum_{i=\mu}^{k-1} \min \{\alpha , (d-i)\beta \} - \sum_{i=\mu}^{k-1} \omega_i \\
& = & \left[ \sum_{i=\mu}^{k-1} (d-i)\beta \right] - \theta - \sum_{i=\mu}^{k-1} \omega_i \label{eq:lower1b} .
\eea
Matching the bounds in \eqref{eq:upper1b} and \eqref{eq:lower1b} and using the identity \eqref{eq:basic1}, we obtain that
\bea
\epsilon & \geq & \frac{(d-k+1)(k-\mu-1)(\beta - \theta) \ - \ \theta}{(d-k+1)(k-\mu) \ + \ 1} \label{eq:eps1b}
\eea

\vspace{0.1in}

{\em Case 2:} $\mu \in \{0,1,\ldots, k-3\}$

We set $r_1=\left\lfloor \frac{k-\mu-1}{\mu+2} \right\rfloor$. We will have two sub-cases for $r_1 \geq 1$ and $r_1=0$. In contrast with {\em Case 1}, we consider a different trapezoid configuration $(Q, Z_q)$ with $q = (\mu+1)$ in {\em Case 2}. It turns out that this change will help in getting a tighter bound in certain regions of $(\mu, \theta)$.

\vspace{0.1in}

{\em Case 2(a):} $r_1 \geq 1$

%\begin{figure}[h!]
%\begin{center}
%\includegraphics[width=3.1in]{proof-3.pdf}
%\caption{The sub-trapezoid region $EFGH$ marked in red in the repair matrix indicates the collection of random variables, ${\cal S}$, considered in {\em Case 2(a)}. Here $r_1=\left\lfloor \frac{k-p}{p+1} \right\rfloor$. } \label{fig:case2a}
%\end{center}
%\end{figure}
\begin{figure}[h!]
\begin{center}
  \subfigure[The sub-trapezoid region considered in Case 2(a)]{\label{fig:case2a}\includegraphics[width=2.7in]{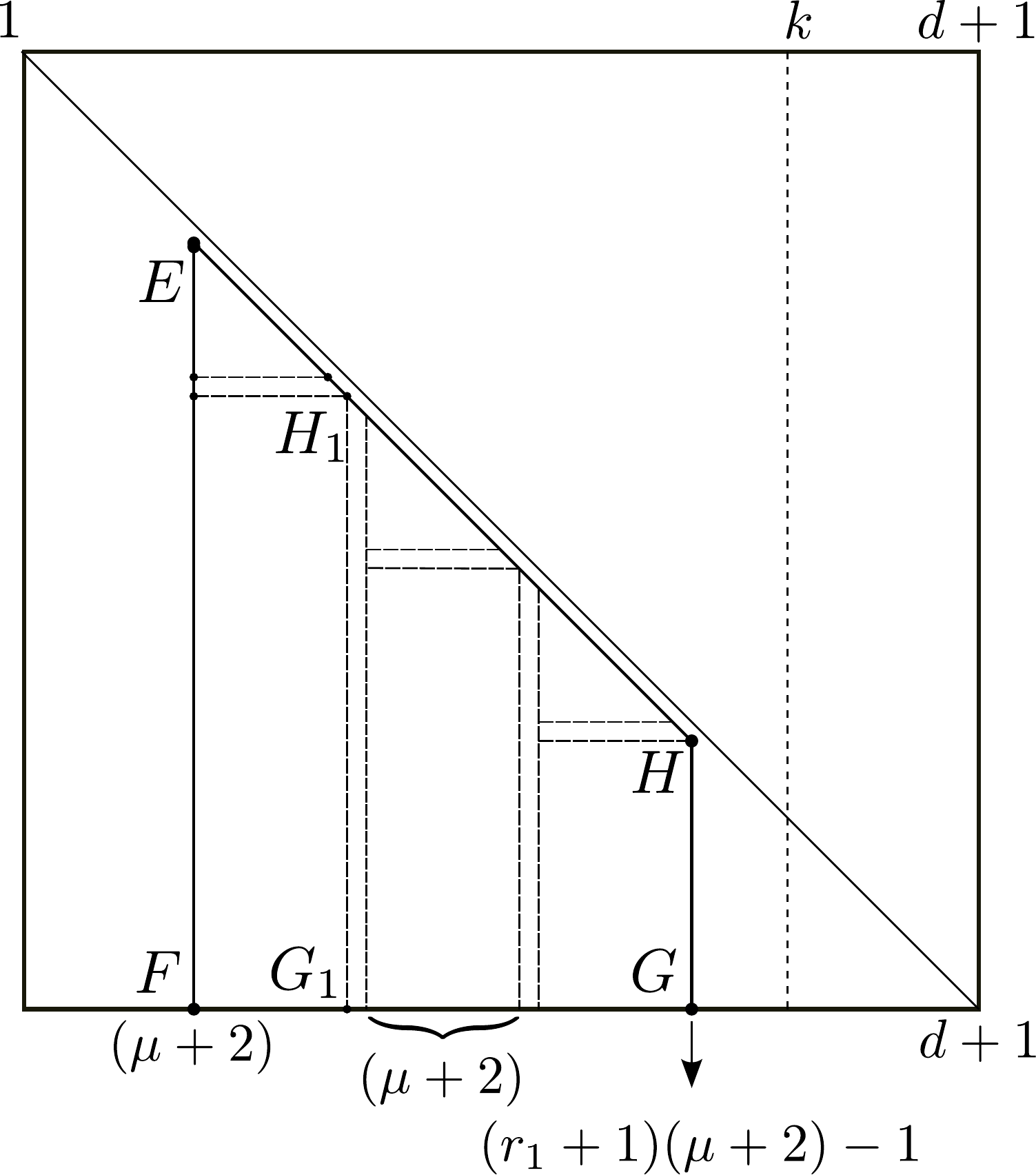}}
  \hspace{0.2in}
  \subfigure[The trapezoid region considered in Case 2(b)]{\raisebox{5mm}{\label{fig:case2b}\includegraphics[width=2.7in]{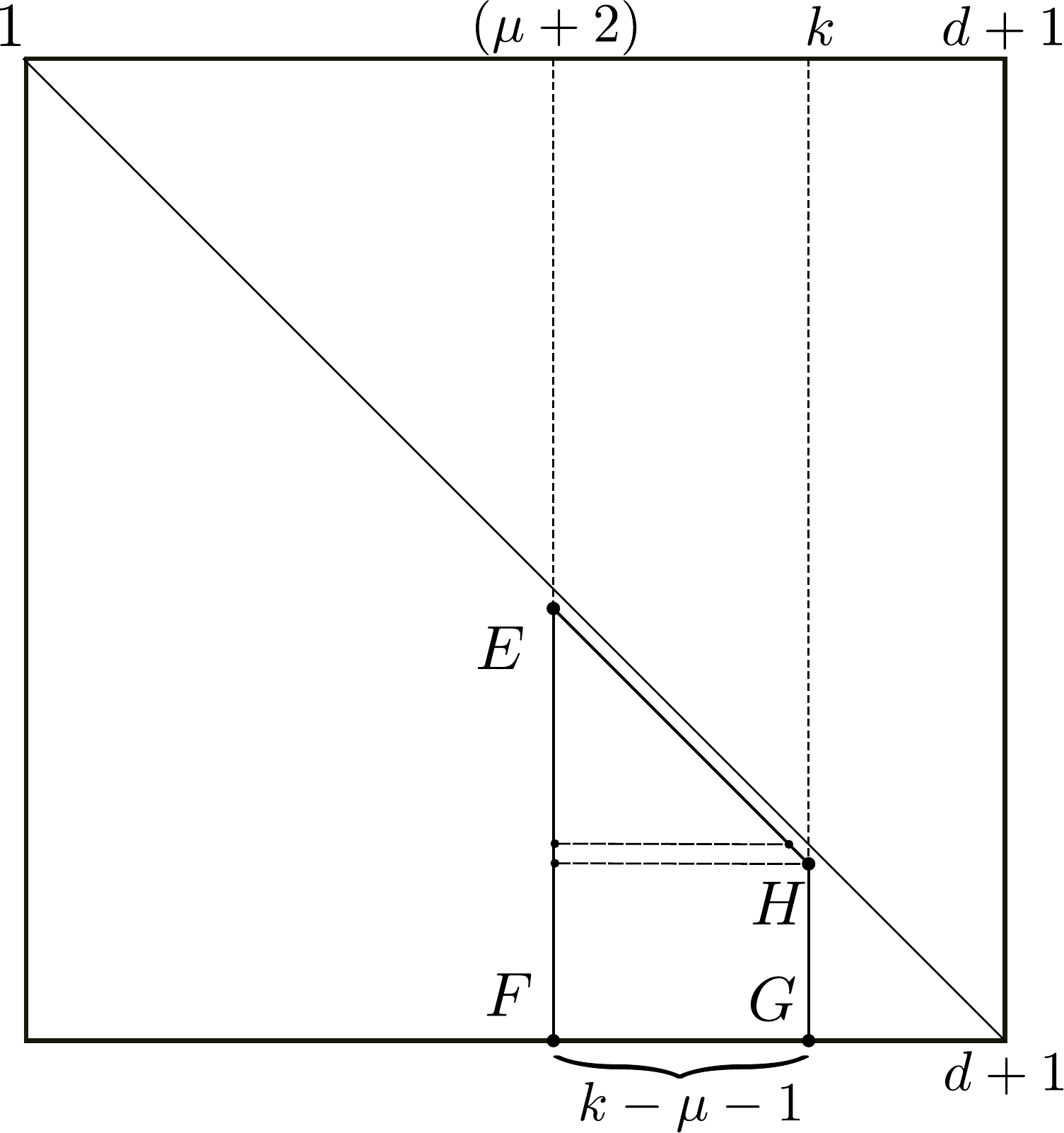}}}
\caption{The illustration of the trapezoid regions considered for Case 2. }
\end{center}
\label{fig:case2}
\end{figure}

In this case, we consider the set $Z_{q,t}$ with parameter $q = \mu+1$, $t = r_1(\mu+2)$. The set of nodes $T = \{\mu+2, \mu+2, \ldots , (r_1+1)(\mu+2)-1\}$ that are repaired by $Z_{q,t}$ is split into $r_1$ groups of $(\mu+2)$ nodes in order, and the corresponding subsets of $Z_{q,t}$ are denoted by ${\cal E}_i, i = 1, 2, \ldots, r_1$. A pictorial illustration is given in Fig.~\ref{fig:case2a}. Every ${\cal E}_i$ is further viewed as the union of two subsets ${\cal V}_i$ and ${\cal T}_i$, respectively associated with the largest rectangle within the trapezium, and the remaining triangular region. The sets of interest are formally defined as
\bean
{\cal E}_i & = & \{ S_x^y \mid S_x^y \in Z_{q,t}, (\mu+2)i \leq y \leq (\mu+2)(i+1)-1 \}, \ \ i = 1,2,\ldots , r_1 \\
{\cal V}_i & = & \{ S_x^y \mid S_x^y \in {\cal E}_i, (\mu+2)(i+1) \leq x \leq d+1 \}, \ \ i = 1,2,\ldots , r_1 \\
{\cal T}_i & = & \{ S_x^y \mid S_x^y \in {\cal E}_i, (\mu+2)i+1 \leq x \leq (\mu+2)(i+1)-1 \}, \ \ i = 1,2,\ldots , r_1 ,
\eean
where ${\cal E}_i = {\cal V}_i \cup {\cal T}_i$. Similar to {\em Case 1(a)}, we bound the joint entropy $H(Z_{q,t})$ as
\bea
H(Z_{q,t}) &  \leq & \sum_{i=1}^{r_1} H({\cal V}_i) + \sum_{i=1}^{r_1} H({\cal T}_i) \\
\nonumber & \leq & \sum_{i=1}^{r_1} \left(d-(i+1)(\mu+2)+2\right) \cdot  [2\beta - \theta + (\mu+1)\omega_{\mu+1} + (\omega_{\mu+1}+\omega_{\mu+2}) ] + \\
 & & \ \ \ \ \ \ \ \ \ \  \ \sum_{i=1}^{r_1} \frac{(\mu+2)(\mu+1)\beta}{2} \label{eq:upper2a} .
\eea
In the last inequality, we have used \eqref{eq:rb2} of Cor.~\ref{cor:rowbound}. On the other hand, using Lem.~\ref{lem:colsum}, we also have,
\bea
H(Z_{q,t}) \ \geq \ H(Z_{q,t}\mid W_Q) &  \geq &  \sum_{i=\mu+1}^{(r_1+1)(\mu+2)-2} \min \{\alpha , (d-i)\beta \} - \sum_{i=\mu+1}^{(r_1+1)(\mu+2)-2} \omega_i \\
& = & \left[ \sum_{i=\mu+1}^{(r_1+1)(\mu+2)-2} (d-i)\beta \right] - \sum_{i=\mu+1}^{(r_1+1)(\mu+2)-2} \omega_i \label{eq:lower2a}
\eea
Matching the bounds in \eqref{eq:upper2a} and \eqref{eq:lower2a} and using the identity \eqref{eq:basic1}, we obtain that
\bea
\epsilon & \geq & \frac{\left(d-\frac{(\mu+2)(r_1+3)}{2}+2 \right)r_1 \left[\mu\beta \ + \ \theta\right] }{\left(d-\frac{(\mu+2)(r_1+3)}{2}+2\right)r_1(\mu+2) \ + \ 1} \label{eq:eps2a}
\eea

{\em Case 2(b):} $r_1 =0$

%\begin{figure}[h!]
%\begin{center}
%\includegraphics[width=2.8in]{proof-4.pdf}
%\caption{The trapezoid $EFGH$ marked in red in the repair matrix indicates the collection of random variables, ${\cal S}$, considered in {\em Case 2(b)}.}   \label{fig:case2b}
%\end{center}
%\end{figure}
The set $Z_{q}$ with $q=\mu+1$ is considered in this case. We can write $Z_{q} = {\cal V} \cup {\cal T}$, where
\bean
{\cal V} & = & \{ S_x^y \mid S_x^y \in Z_{q}, k+1 \leq x \leq d+1 \}, \\
{\cal T} & = & \{ S_x^y \mid S_x^y \in Z_{q}, \mu+3 \leq x \leq k \} .
\eean
A pictorial illustration is given in Fig.~\ref{fig:case2b}. Following the same line of arguments as in {\em Case 1(a)}, we obtain that 
\bea
H(Z_{q}) &  \leq & (d-k+1) \cdot [2\beta - \theta + (k-\mu-1)\epsilon ] + \frac{(k-\mu-2)(k-\mu-1)\beta}{2}  \label{eq:upper2b},\\
H(Z_{q}) &  \geq &  \left[ \sum_{i=\mu+1}^{k-1} (d-i)\beta \right] - \sum_{i=\mu+1}^{k-1} \omega_i \label{eq:lower2b}.
\eea
Matching the above two bounds and using the identity \eqref{eq:basic1}, we obtain the lower bound for $\epsilon$:
\bea
\epsilon & \geq & \frac{(d-k+1) \left[ (k-\mu-3)\beta + \theta \right]}{(d-k+1)(k-\mu-1) \ + \ 1} \label{eq:eps2b}
\eea

\section{Proof of Lem.~\ref{lem:544_intersections}\label{app:pf_lem_544}}

By definition of $\delta_j$ in \eqref{eq:proof_diamkis_0}, we have that 
\begin{eqnarray}
	\delta_j & = & \rho\left(H^{(5)}|_{[j]}\right) - \rho\left(H^{(5)}|_{[j-1]}\right) \\
	& = & \text{dim}\left(\mathcal{S}\left(H^{(5)}|_{[j-1]} \right) + \mathcal{S}\left(H^{(5)}_j \right)\right) -  \text{dim}\left(\mathcal{S}\left(H^{(5)}|_{[j-1]}\right)\right) \label{eq:544_proof_1emma_1_1} \\ 
	& = & \text{dim}\left(\mathcal{S}\left(H^{(5)}_j \right)\right) - \text{dim}\left(\mathcal{S}\left(H^{(5)}|_{[j-1]} \right) \cap \mathcal{S}\left(H^{(5)}_j \right)\right) \label{eq:544_proof_1emma_1_2} \\
	& = & \rho\left(H^{(5)}_j \right) - \rho\left(H^{(4)}_j \right) \label{eq:544_proof_1emma_1_3} \\
	& = & \rho\left(A^{(5)}_{j,j} \right) - \rho\left(A^{(4)}_{j,j} \right), \label{eq:544_proof_1emma_1_4}
	\end{eqnarray} 
where in \eqref{eq:544_proof_1emma_1_2} we used the identity $\text{dim}(W_1 + W_2)  = \text{dim}(W_1) + \text{dim}(W_2) - \text{dim}(W_1 \cap W_2)$ for any two subspaces $W_1, W_2$. While \eqref{eq:544_proof_1emma_1_3} follows from the definition of $H^{(4)}_j$, \eqref{eq:544_proof_1emma_1_4} from Remark~\ref{rem:full_column_rank_diag_sub_matrix}. The first assertion \eqref{eq:544_key_lemma1} of the lemma now follows from \eqref{eq:544_proof_1emma_1_4} and \eqref{eq:delta_equality_544}.
	
By definition of $H^{(4)}_j$, we have that $\mathcal{S}\left(A^{(4)}_{j,j}\right) \subseteq  \sum_{\ell = 1}^{j-1}\mathcal{S}\left(A^{(5)}_{j,\ell}\right)$, and it follows that 
\begin{eqnarray}
\rho \left(A^{(4)}_{j,j}\right) & \leq &  \text{dim} \left( \sum_{\ell = 1}^{j-1}\mathcal{S}\left(A^{(5)}_{j,\ell}\right) \right). \label{eq:544_proof_1emma_2_1}
\end{eqnarray}
The RHS of \eqref{eq:544_proof_1emma_2_1} is further upper bounded as follows:
\begin{eqnarray}
\text{dim} \left( \sum_{\ell = 1}^{j-1}\mathcal{S}\left(A^{(5)}_{j,\ell}\right) \right) & = & \text{dim} \left( \sum_{\ell = 1}^{j-2}\mathcal{S}\left(A^{(5)}_{j,\ell}\right)  + \mathcal{S}\left(A^{(5)}_{j,j-1}\right)\right) \label{eq:544_proof_1emma_2_2} \\
	\nonumber & = & \text{dim} \left( \sum_{\ell = 1}^{j-2}\mathcal{S}\left(A^{(5)}_{j,\ell}\right) \right) + \text{dim} \left(  \mathcal{S}\left(A^{(5)}_{j,j-1}\right)\right) - \\
		& & \ \ \ \ \ \ \ \ \   \text{dim} \left( \sum_{\ell = 1}^{j-2}\mathcal{S}\left(A^{(5)}_{j,\ell}\right)  \cap \mathcal{S}\left(A^{(5)}_{j,j-1}\right)\right) \label{eq:544_proof_1emma_2_3} \\
	& \leq & \text{dim} \left( \sum_{\ell = 1}^{j-2}\mathcal{S}\left(A^{(5)}_{j,\ell}\right) \right) + \text{dim} \left(  \mathcal{S}\left(A^{(5)}_{j,j-1}\right)\right) - \text{dim} \left( \mathcal{S}\left(A^{(4)}_{j,j-1}\right)\right) \label{eq:544_proof_1emma_2_4} \\
	& = & \text{dim} \left( \sum_{\ell = 1}^{j-2}\mathcal{S}\left(A^{(5)}_{j,\ell}\right) \right) + \rho\left(A^{(5)}_{j,j-1}\right) - \rho\left(A^{(4)}_{j,j-1}\right), \label{eq:544_proof_1emma_2_5}
\end{eqnarray}
where \eqref{eq:544_proof_1emma_2_4} follows from the definition of $H^{(4)}_j$ and $A^{(4)}_{j,j-1}$. If $j=3$, \eqref{eq:544_proof_1emma_2_5} completes the proof of the second assertion. Else, for the case $j \geq 4$, the term $\text{dim} \left( \sum_{\ell = 1}^{j-2}\mathcal{S}\left(A^{(4)}_{j,\ell}\right) \right)$ can further be upper bounded by following a similar sequence of steps as in \eqref{eq:544_proof_1emma_2_2} -	\eqref{eq:544_proof_1emma_2_5}. This completes the proof.

\end{document}